\documentclass[a4paper, 12pt]{article}
\usepackage{amsmath,amssymb, amscd}
\usepackage{pdfsync}
\usepackage{graphicx,subfigure}
\usepackage{color}
\numberwithin{equation}{section}

\newtheorem{theorem}{Theorem}[section]

\newtheorem{conjecture}[theorem]{Conjecture}
\newenvironment{proof}{\noindent {\it Proof }}{\hfill $\square$}

\newcommand{\eqa}{\begin{eqnarray}}
\newcommand{\eeqa}{\end{eqnarray}}
\newcommand{\beq}{\begin{equation}}
\newcommand{\eeq}{\end{equation}}

\begin{document}

\title{On critical behaviour in    generalized Kadomtsev--Petviashvili equations }
\author{B.~Dubrovin,\\
SISSA, Via Bonomea 265, I-34136 Trieste, Italy\\
T. Grava,\\ School of Mathematics, University of Bristol, University 
Walk,\\
Bristol BS8 1TW, United Kingdom,\\ and\\
SISSA, Via Bonomea 265, I-34136 Trieste, Italy,\\ C. 
Klein,\\ 
Institut de Math\'ematiques de Bourgogne, Universit\'e de Bourgogne, 
\\
9 avenue Alain Savary, 21078 Dijon Cedex, France 
}

\maketitle

\begin{abstract}
An asymptotic description of  the formation of  dispersive shock  waves in 
solutions to the generalized  Kadomtsev--Petviashvili (KP) equation is conjectured.
 The  asymptotic description based on a multiscales expansion is given in 
terms of a special solution to an ordinary differential equation of 
the Painlev\'e I hierarchy. Several examples are discussed 
numerically to provide strong evidence for the validity of the 
conjecture. The numerical study of the long time behavior of these 
examples  indicates persistence of dispersive shock waves  in solutions to the 
(subcritical) KP  equations,   while in the  supercritical 
KP equations a blow-up occurs  after the formation of the dispersive shock waves.

\end{abstract}

\section{Introduction}
In this manuscript we consider the Cauchy problem for the   
generalized  Kadomtsev--Petviashvili (KP) equations
\begin{equation}\label{KP}
(u_t+u^n u_x+\epsilon^2u_{xxx}  )_x= \sigma u_{yy},\quad n\in\mathbb{N},
\end{equation}
where $\epsilon$ is a small positive parameter and $\sigma=\pm 1$. We are interested in 
studying the behaviour of solutions $u(x,y,t;\epsilon)$ for 
$\epsilon\to0$ when the initial data $u(x,y,t=0;\epsilon)=u_0(x,y)$ 
are independent from $\epsilon$.  Generically  the solution develops oscillations 
that are  called \emph{Dispersive Shock Waves} \cite{GP} in the  nonlinear 
wave community or   \emph{undular bore} in the fluid dynamics community.
The related mathematical problem is now well understood for several dispersive equations in one space dimension, starting with the seminal work of Lax and Levermore \cite{LL} (see also  \cite{Venakides} and  \cite{DVZ}) for the Korteweg--de Vries (KdV) equation (see \cite{GK12} for a numerical  treatment of the problem).
However, a rigorous mathematical study of two-dimensional dispersive shock waves remains an open problem  even though there are some heuristic  arguments for several  dispersive equations in two spatial dimension like \cite{Ablowitz}, \cite{Hoefer}, \cite{El}.

Our main result is a description of the solution $u(x,y,t;\epsilon)$ 
of the generalized KP equation  for  generic initial data,  on  the 
onset of the oscillations in terms of a particular solution of  an 
ordinary differential equation, the so-called Painlev\'e I2 equation (PI2), the second  member of the Painlev\'e I hierarchy.
This description extends the universality results on critical behaviour 
of Hamiltonian PDEs, first obtained for one-dimensional evolutionary 
equations \cite{Dub1}, \cite{Dub2}, \cite{DGK}, \cite{DGKM}, to PDEs in two spatial dimensions.

For $n=1$ the  equations (\ref{KP})  are  known as  KP equations (KP I 
for $\sigma=+1$ and KP II for $\sigma=-1$). They were introduced in \cite{KP70} to study the transverse stability of the solitary wave solution of the 
KdV in a $2+1$ dimensional setting. Both cases 
can be derived as models for nonlinear dispersive waves on the 
surface of fluids \cite{KP70} (see also \cite{Joh,KS12} for further 
references). The settings studied  in this context are essentially 
one-dimensional waves with weak transverse modulation. 
 KP I is applicable in the case of strong
surface tension, whereas KP II is a model for weak 
surface tension. Note that KP I has a focusing effect, whereas KP II 
is defocusing. KP type equations also arise as a 
model for sound waves in ferromagnetic media \cite{TF} and in the 
description of two-dimensional nonlinear matter-wave pulses in 
Bose--Einstein condensates, see e.g. \cite{HMV,JR}. As for the KdV 
equation, KP equations appear as an asymptotic description in the 
limit of long wavelengths. In general the dispersion of KP equations 
is too strong 
compared to what is found in applications. A way to tilt the balance between 
dispersion and nonlinearity towards the nonlinearity is to consider 
generalized KP equations (\ref{KP}) with a stronger nonlinearity $n>1$. 
Interestingly  the generalized KP for $n = 2$ appears as a model for the evolution 
of sound waves in antiferromagnetic materials, see \cite{TF}. 

In the dimensionless generalized KP equation, i.e., equation (\ref{KP}) with 
$\epsilon=1$, a parameter $\epsilon$ can be introduced in the 
following way: a possible approach to studying the long time behavior of 
solutions of this dimensionless generalized KP equation is to consider 
slowly varying initial data of the form $u_{0}(\epsilon x,\epsilon y)$
where $0<\epsilon \ll 1$ is a small  parameter and $u_{0}(x,y)$ is some given initial profile. As $\epsilon \to 0$ the 
initial data approach a constant value. Hence, in order to see nontrivial effects one has to wait until sufficiently long times of order $t\sim O(1/\epsilon)$, which 
consequently requires to rescale the spatial variables onto 
macroscopically large scales $x\sim O(1/\epsilon)$, too. 
In other words, we consider $x\mapsto \tilde x= x\epsilon$, $x\mapsto 
\tilde y= y\epsilon$, $t\mapsto \tilde t= t\epsilon$ and put 
$u^\epsilon(\tilde t,\tilde x, \tilde{y}) = u(\tilde t/\epsilon, 
\tilde x/ \epsilon,\tilde{y}/\epsilon)$ to obtain equation (\ref{KP}) 
(we omit the `tildes' for simplicity). The limit of small $\epsilon$ 
is called the \emph{small dispersion limit}. 

We expect that for a reasonable class of smooth initial data $u_0(x,y)$ the solution $u(x,y,t; \epsilon)$ to the Cauchy problem for the  generalized KP equation  remains smooth and  depends continuously on the sufficiently small scaling parameter $\epsilon$ on a finite time interval. 
 The solution $u(x,y,t,\epsilon)$ of the generalized KP equation  (\ref{KP}) is expected  to converge in the limit $\epsilon\to 0$ to 
the solution $u(x,y,t)$  of the generalized dKP equation (\ref{dKP})
\begin{equation}\label{dKP}
(u_t+u^n u_x)_x= \pm u_{yy},\quad n\in\mathbb{N},
\end{equation}
  for $0<t<t_{c}$ where $t_c$ is the time where the solution of 
  equation (\ref{dKP}) first  develops a singularity (blow-up of gradient) which generically appears in one point $(x_c,y_c)$ of the plane. 
Equations (\ref{dKP}) are  called  generalized dKP equations or generalized dispersion-less KP equations, even though the equations (\ref{dKP}) have dispersion.
For $n=1$ these equations were derived  earlier than the KP  equation by  Lin, Reissner and Tsien \cite{LRT} and Khokhlov and  
Zabolotskaya  \cite{ZK69} in a $3+1$ dimensional setting. However, in this manuscript we call (\ref{dKP}) generalized dKP equations.

Local well-posedness of the Cauchy problem for  dKP equation ($n=1$) has been proved in certain Sobolev spaces in \cite{Rozanova}, while for the generalized dKP equation a similar result is still missing to the best of our knowledge. Furthermore  although many strong results about the Cauchy problem for the KP equation in various functional spaces are now available (see, e.g., \cite{Bourgain, Liu, MST, Saut93}),  these results are still insufficient to rigorously justify the  small $\epsilon$  behaviour of solutions to KP and its generalizations. Namely, for $t<t_c$ the  
 solution of the generalized KP equation  $u(x,y,t;\epsilon)$  is conjectured  to be  approximated in the limit $\epsilon\to 0$ by the 
solution $u(x,y,t)$ of the generalized dKP equation with the same initial data.
The numerical results of Sections \ref{sec:KP}, 
 \ref{sec:gKP} provide a rather convincing motivation for the above 
 conjectural statement.  Such conjectural small $\epsilon$ behaviour is somehow  expected  in the subcritical  case, where solutions of generalized KP do no have  blow-ups, while 
in the supercritical case, when solutions of generalized KP equations 
can have a blow-up in finite time $T<\infty$, our conjecture implies  that $T>t_c$.

Moreover, our numerical results  strongly support the conjectural analytic description of the leading term in the asymptotic expansion of solutions \emph{near} the critical time  $t_c$ that we will explain now.
Let   $(x_c,y_c)$ be the point 
where the solution of  generalized dKP equation (\ref{dKP})  
develops  a singularity and   let  $u_c=u(x_c,y_c,t_c)$. 
Furthermore let us denote by $\bar{x}:=x-x_c$, $\bar{y}:=y- y_c $ and $\bar{t}:=t= t_c$ the shifted coordinates in the  neighbourhood of the critical point   and  let 
us introduce the  following variables
\begin{equation}
\label{XT00}
{X}=\bar{x}-u_c^n\bar{t}+c_1\bar{t}\bar{y}+c_2\bar{y}+c_3\bar{y}^2+c_4\bar{y}^3,\quad T=\bar{t}+b\, \bar{y}^2
\end{equation}
where $c_1,\dots, c_4$ and $b $ are constants that depend on the equation and the initial data. 
Then in the double scaling  limit 
$\bar{x}\to 0$, $\bar{y}\to 0$, $\bar{t}\to 0$ and $\epsilon\to 0$ in such a way that $X/\epsilon^{\frac{6}{7}}$ and $T/\epsilon^{\frac{4}{7}}$ remain finite,
 the solution of the generalized KP equation  $u(x,y,t;\epsilon)$ in  a neighbourhood of the point $(x_c,y_c,t_c)$  at the onset of the oscillations has the following expansion 
\begin{equation}
\label{asymp1}
u(x,y,t;\epsilon)=u_c+\dfrac{6}{nu_c^{n-1}} \left(\dfrac{\epsilon^2}{\kappa^2}\right)^{\frac{1}{7}}U\left(\frac{X}{( \kappa\epsilon^6)^{1/7}},\frac{T}{(\kappa^3\epsilon^4)^{1/7}}\right)+\beta\bar{y}+O(\epsilon^{\frac{4}{7}}),
\end{equation}
where $\beta$ and $\kappa$ are constants and $X$ and $T$ have been defined above in (\ref{XT00}).
%\[
%k=-36 t_c^4\left.\dfrac{\partial^3 F^n(\xi,y_c,t_c)}{\partial \xi^3}\right|_{\xi=\xi_c}
%\]
%and $U(X,T)$ solves the P1-2 equation (\ref{PI2asym}).

 The function $U=U({\cal X},{\cal T})$ satisfies the PI2 equation
\begin{equation}
\label{PI2}
{\mathcal X}=6{\mathcal T}U-\left(U^3+\frac{1}{2}U_{{\mathcal 
X}}^2+UU_{{\mathcal XX}}+\dfrac{1}{10}U_{{\mathcal XXXX}}\right).
\end{equation}
The relevant solution is uniquely determined by  the asymptotic conditions
\begin{equation}
\label{PI2asym}
U({\cal X},{\cal T})=\mp |{\cal X}|^{\frac{1}{3}}\mp \frac{2{\cal T}}{|{\cal X}|^{\frac{1}{3}}}+O(|{\cal X}|^{-1}),\quad |{\cal X}|\to\infty.
\end{equation}
The existence  for real ${\cal X}$ and ${\cal T}$ of a smooth  real solution of the PI2 equation (\ref{PI2})  satisfying the boundary conditions 
(\ref{PI2asym}) has been conjectured in \cite{Dub1} and proved in 
\cite{CV}. The solution has an oscillatory region  that is formed 
around the point ${\cal X}=0$ and at about  the time ${\cal T}=0$  and it is developing for  ${\cal T}\gg 1$ into the Gurevich--Pitaevski solution \cite{GP}, \cite{Claeys0}, \cite{GK12}, \cite{Suleimanov}. 
We remark that  from (\ref{asymp1})  for  generic initial data  the 
oscillatory front  is,   approximately, a straight line in the $(x,y)$ plane, while for  initial data having a  $y$-symmetry,  the oscillatory front   has  a parabolic shape.

% For initial data satisfying the symmetry $u_0(x, -y)=u_0(x,y)$, one obtains that 
% \begin{equation}
% \label{asymp2}
% u(x,y,t;\epsilon)=u_c+\dfrac{6}{nu_c^{n-1}} 
% \left(\dfrac{\epsilon^2}{\kappa^2}\right)^{\frac{1}{7}}U\left(\frac{\bar{x}-u_c^n\bar{t}+c_3\bar{y}^2}{( \kappa\epsilon^6)^{1/7}},\frac{\bar{t}+b\bar{y}^2}{(\kappa^3\epsilon^4)^{1/7}}\right)+O(\epsilon^{\frac{4}{7}}).
% \end{equation}
% and this fact is consistent with the ansatz made in \cite{Ablowitz} in the study of dispersive shock waves in KP equation.

Several numerical examples are presented to provide strong numerical 
support for our conjecture. Furthermore we will show in section~\ref{outlook} that the formula ({\ref{PI2asym}) catches some of the qualitative features of the oscillatory regime also 
for some time considerably  bigger then the critical time $t_c$.
 
The paper is organised as follows: in section \ref{prelim} we 
collect some mathematical facts about KP equations and present the main 
conjecture of this paper. In section 
\ref{sec2} we briefly summarize the used numerical techniques. 
Solutions to the KP equations are discussed for various initial data 
in the vicinity of the critical points in section \ref{sec:KP} and 
compared to the asymptotic description. In 
section \ref{sec:gKP} we study solutions to the generalized KP equation for 
$n=3$ for the same initial data near the critical points. We conclude 
in section \ref{outlook} with a preliminary discussion of the 
longtime behavior of solutions to KP and generalized KP equations, i.e., the 
formation of dispersive shocks for all times in the former and eventual blow-up in 
the later, and we outline directions of further research. 

\section{Asymptotic description of break-up in KP solutions}
\label{prelim}
In this section we present an asymptotic description of the formation 
of a dispersive shock wave  in solutions to the generalized  KP equations for initial data in the 
Schwartz class. We first collect some mathematical facts about generalized  KP 
equations, then summarize previous work on the formation of  dispersive shock waves  in 
KdV solutions. Finally the latter results are generalized to the case of 
KP equations. 

\subsection{Mathematical preliminaries}
The integrability of the KP equations  was obtained in 1974 
\cite{Dryuma} via Lax pair. However, the problem to effectively  integrate the 
equation proved to be quite difficult and the first achievements were 
obtained in \cite{ZS,Manakov,FokasAblowitz,Ablowitzetall}.  
%For initial data in the Schwartz class and $\epsilon$ fixed the solution of the  KP equation (the case $n=1$) can be solved via inverse scattering transform \cite{Manakov}.
For initial data $u_0(x,y)$ in the Schwartz class,  the solution $u(x,y,t;\epsilon)$ remains a smooth function of $x$ and $y$.
However, the solution $u(x,y,t;\epsilon)$ immediately leaves the Schwartz space of rapidly decreasing functions as soon as $t>0$ even though it  remains in $L^2(\mathbb{R})$ since the  quantity
\[
\int_{\mathbb{R}^{2}}u^2(x,y,t;\epsilon)dxdy=\int_{\mathbb{R}^{2}} u^2_0(x,y)dxdy,
\]
is conserved in time. 
% Indeed by multiplying (\ref{KP}) by $u$ and integrating in $x$ and $y$  one can see that the integral is constant. 

 For $x\to -t\infty$ the solution is still rapidly  decreasing while for  
$x\to t\infty$  one has for KP I \cite{BPP94}
\begin{equation}
\label{asymp}
u(x,y,t;\epsilon)=\dfrac{c}{\sqrt{tx}x}\int_{\mathbb{R}^{2}} dx'dy'u_0(x',y')+o(|x|^{-\frac{3}{2}})
\end{equation}
 with $c$ a constant.
 Furthermore the solution  for $t>0$ satisfies an infinite number of dynamical constraints, the  first two, taking the form \cite{LinChen}
 \begin{align}
 \label{const1}
 &\int_{\mathbb{R}} u(x,y,t;\epsilon)dx=0,\\
  \label{const2}
 & \int_{\mathbb{R}} x u_y(x,y,t;\epsilon)dx=0.
\end{align}
%where the $y$-derivative cannot be taken out of the integral.
%The first constraint can be easily  derived by integrating the KP equation in $x$, obtaining 
%\[
% \int u_{yy}dx=0
%\]
%which implies (\ref{const1}) by (\ref{asymp}).  
The dynamical constraints   are satisfied even if the initial data do not satisfy them.
This is a manifestation of the infinite speed of propagation inherent 
to the KP equations. The constraint (\ref{const1})  is  satisfied also for the generalized KP equation \cite{MST07}.
 For  initial data  in the Schwartz space, the  solution   $u(x,y,t; \epsilon)$  of the KP equation $u(x,y,t;\epsilon)$ is a $C^{\infty}$ function  in $\mathbb{R}^3_+=\{(x,y,z)\,;\,t>0\}$  even if the constraints are not satisfied at $t=0$. 
This  result has been proved via inverse scattering for the KP~I  
equation for small norm initial data  \cite{FS97} and for the KP~II  equation in \cite{Ablowitzetall}.
Global well-posedness for KP~II  was shown 
in \cite{Bourgain}  in the Sobolev space $H^s(\mathbb{R}^2)$, $s\geq 0$,  so that for $s\geq 4$ one gets classical solutions while the global well-posedness for KP~I was obtained
in some subset of the Sobolev space $H^s(\mathbb{R}^2)$ \cite{MST}.
For the  generalized KP equation, local well-posedness results  have been established in some weighted Sobolev space, see e.g. \cite{IN}.
 
 We can conclude that for  initial data in the Schwarz space the solution  $u(x,y,t;\epsilon)$ of the generalized KP equation is  sufficiently regular  in  space  and  time for  $0<t<T$, where in general  $T<\infty$ for the generalized KP equations, while $T=\infty$ for the KP equations.
 We also assume that the solution of the generalized  KP equation depends continuously on the small parameter $\epsilon$.

\subsection{Break-up in dKP solutions}

The same concepts can be applied to the solution of the generalized 
dKP equation (\ref{dKP}). For generic initial data the solution of 
(\ref{dKP}) exists  till a time $t_c$ where a gradient catastrophe 
occurs, however for $t<t_c$ the solution $u(x,y,t)$ of the generalized dKP equation 
is expected to be smooth in  both in time and space   for initial data in the Schwartz class.  
The $L^2$ norm  of the solution is conserved as for the generalized 
KP equation.
% \begin{equation}
% \label{L2dKP}
% \iint u^2(x,y,t)dxdy=\iint u^2_0(x,y)dxdy.
% \end{equation}
Integrating the generalized dKP equation with respect to $x$
one obtains the constraint  
\[
\int_{\mathbb{R}} u_{yy}(x,y,t)dx=0.
\]
%and since the evolution preserves the $L^2$ norm and $u(x,y,t)$ is expected to be smooth for $t<t_c,  one can conclude that also for generalized dKP one has the constraint 
%\[
%\int u(x,y,t)dx=0,\quad t<t_c.
%\]

 While the Cauchy problem for the generalized KP equation has seen considerable attention,
the generalized dKP initial valued problem has been less studied. 
A Lax pair for dKP equation ($n=1$) was obtained in  \cite{TT}.
A definition of integrability for the dKP equation was obtained 
in \cite{Ferapontov} where the method of hydrodynamic reduction is 
used. This method was introduced in  \cite{GK}, to obtain particular solutions of the dKP equation while more general solutions have been obtained in  \cite{MS08}, \cite{MS12} using  inverse scattering. General  solutions have also been obtained \cite{Raimondo}.
 In  \cite{Rozanova}  it is proved  that the Cauchy problem  of dKP ($n=1$) is well posed in the Sobolev  spaces $H^{s}(\mathbb{R}^2),  $ for $t<t_c$ and $s>2$, while a similar statement is missing for  generalized dKP equation.
 
In a recent paper, two of the present authors, inspired by the works  \cite{MS08}, \cite{MS12} have  studied the Cauchy problem for the dKP equation with 
smooth initial data using a change of the independent variable $x$ 
suggested by the method of characteristics for the $1+1$ dimensional 
case.
In the following we are using the same idea to obtain the solution of the generalized dKP equation (\ref{dKP}) in the following form
\begin{equation}
\label{implicit0}
\left\{
\begin{array}{ll}
&u(x,y,t)=F(\xi,y,t)\\
&x= tF^n(\xi,y,t)+ \xi\\
&F(x,y,0)=u_0(x,y)
\end{array}\right.
\end{equation}
where $u_0(x,y)$ is an initial datum in the Schwartz class 
$\mathcal{S}(\mathbb{R}^{2})$ of rapidly decreasing smooth functions.  
Plugging the above ansatz into the generalized dKP equation one obtains an equation for the function $F(\xi,y,t)$
% From (\ref{implicit0}) one has
% \begin{align}
% &\xi_x=\dfrac{1}{\Delta},\quad \Delta =  1 +  nt F^{n-1}F_{\xi}\\
% &\xi_t=-\dfrac{F^n+ntF^{n-1}F_t}{\Delta}\\
% &\xi_y=-\dfrac{ntF^{n-1} F_y}{\Delta}\\
% \end{align}
% so that 
% \beq
% \label{ut}
% u_t = F_{\xi}\xi_t + F_t = \frac{F_t- F^nF_{\xi}}{\Delta}, \quad u_x=\dfrac{F_\xi}{\Delta}.
% \eeq
% From the above relation one has
% \[
% (u_t+u^nu_x)_x=\dfrac{1}{\Delta}\left(\frac{ F_t}{\Delta}\right)_{\xi}.
% \]
% In the same way
% \begin{equation}
% \label{uxy}
% \quad u_{y}=\dfrac{ F_y}{\Delta}. 
% \end{equation}
% Differentiating
% (\ref{uxy}) a second time, we find
% \[
%   u_{yy} = \left(\frac{F_y}{\Delta}\right)_y - 
% \left(\frac{F_y}{\Delta}\right)_{\xi} 
% \frac{ntF^{n-1} F_y}{\Delta} = \frac{1}{\Delta}\left[
% F_{yy} - \left(\frac{ntF^{n-1} F_y^2}{\Delta}\right)_{\xi}\right],
% \]
% after some manipulations. But this means that if $u(x,y,t)$ satisfies 
% the generalized  dKP equation (\ref{dKP}) then  $F(\xi,y,t)$ satisfies 
\begin{equation}
\label{eqF}
\left(\frac{F_{t} \pm  ntF^{n-1}F_{y}^{2}}{1+nt F^{n-1}F_{\xi}}\right)_{\xi}=\pm F_{yy},
\end{equation}
with initial condition
\[
F(x,y,0)=u_0(x,y).
\]
One can easily check that for $y$-independent initial data 
equation (\ref{implicit0}) is equivalent  to the method of characteristics.
The advantage of such a transformation is that while the solution 
$u(x,y,t)$  of the generalized dKP equation develops a 
singularity at a certain critical time $t_c$ and at the point 
$(x_c,y_c)$,  the solution $F(\xi,y,t)$ seems to exist for much 
longer times, at least numerically. 
%
%Another way of writing the solution is 
%\[
%u(x,y,t)=F(x-2tu(x,y,t),y,t).
%\]
%If $u(x,y,t)$ solves the dKP equation (\ref{dKP}) then $F(\xi,y,t)$ satisfies the equation
%\begin{equation}
%\begin{split}
%(1+2F_{\xi})^2F_{yy}+(1+2tF_{\xi})(F_\xi^2-F_{\xi t}-4tF_yF_{\xi y})+(4t^2F^2_y+F+2tF_t)F_{\xi\xi}=0.
%\end{split}
%\end{equation}
%This equation can also be written in the form
%\begin{equation}
%\left(\frac{F_{t}-FF_{\xi}+2tF_{y}^{2}}{1+2tF_{\xi}}\right)_{\xi}=F_{yy}.
%\end{equation}
%

%
%As we explained in the introduction, the solution of the generalized  dKP equation is constructed by introducing the function $F(\xi,y,t)$ that satisfies equation (\ref{eqF}). 
%Form the conservation of the $L^{2}$ norm of a dKP  solution, can easily deduce that 
%% \[
%% \iint F^2(\xi,y,t)d\xi dy=\iint F^2(\xi,y,0)d\xi dy,
%% \]
%the $L^2$ norm of $F(\xi,y,t)$ is conserved with time. Integrating  
%equation (\ref{eqF}) with respect to $\xi$ one obtains the constraint
%\begin{equation}
%\label{constF}
%\int_{\mathbb{R}} F_{yy}(\xi,y,t)d\xi=0.
%\end{equation}
For initial data in the Schwartz class, we assume that the solution of 
equation (\ref{eqF}) is sufficiently smooth in time and space for 
certain time $0<t<T$ with $ T>t_c$.  The change of the independent  variable from 
$x$ to $\xi$ has a smoothing effect on the solution $F(\xi,y,t)$ in 
the sense that, numerically, the solution stays regular for much longer times than the critical time $t_c$  where the derivatives of $u(x,y,t)$ blow up.
%Numerically the existence in time of the solution $F(\xi,y,t)$  of 
%equation (\ref{eqF}) seems to be much longer than the critical $t_c$.
%For any fixed $t_0>0$ we can look for a solution  in a power series in $t-t_0$ of the form
%
%\begin{equation}
%\label{Taylor}
%F(\xi,y,t)=F(\xi,y,t_0)+(t-t_0)F_1(\xi,y,t_0)+(t-t_0)^2F_2(\xi,y,t_0)+\dots.
%\end{equation}
%Since  the solution satisfies (\ref{constF})  for $t>0$  one must have
%\begin{align}
%&0=\int  F_{1}(\xi,y,t_0)dx=\pm\int \partial_{\xi}^{-1}F_{yy}(\xi,y,t_0)dx\\
%%&=\mp\int \xi F_{yy}(x,y,t_0)dx,\\
%\nonumber
%&\dots
%%\quad \int u_2dx=\int \partial_x^{-1} u_{1yy}dx
%\end{align}
%where the operator $\partial_{\xi}^{-1}$ is defined as 
%\[
%\partial_{\xi}^{-1}=\int_{-\infty}^{\xi}
%\]
%Namely analyticity in time of the solution $F(\xi,y,t)$ is guaranteed by the fact that the solution $F(\xi,y,t)$ satisfies an infinite number of constraints.

The first singularity appears when the change of coordinates $x=tF(\xi,y,t)+\xi$ is not invertible any more.
The equations that describe the   singularity  formation have been considered in \cite{Alinhac} using PDE techniques.
To study the local behaviour of the function  $u(x,y,t)$  (as  a multivalued function)   around the critical time  $t_c$  when the first  singularity appears we make  a Taylor expansion of  (\ref{implicit0}) near the critical point $u_c(x_c,y_c,t_c)$ where $\xi=\xi_c$.
Let us consider the gradients
\begin{equation}
u_x=\dfrac{F_\xi}{\Delta}\quad u_{y}=\dfrac{F_y}{\Delta},\quad \Delta=1+nt F^{n-1}F_{\xi}(\xi,y,t).
\end{equation}
The time  of gradient blow-up is the  smallest time $t_c$ where the 
gradient goes to infinity, namely where
\[
\Delta=1+nt F^{n-1}F_{\xi}(\xi,y,t)=0.
\]
For simplicity let us introduce the quantity
\[
 G(\xi,y,t):=F^n(\xi,y,t).
 \]
Since the  quantity  $\Delta(\xi,y,t)$, for $t<t_c$, has a definite sign in the $\xi$ and $y$ plane, the first point  where it vanishes is a double zero, therefore the point of gradient catastrophe is characterised by the equations
\begin{equation}
\label{GC}
\begin{split}
&\Delta=\pm 1 +nt F^{n-1}F_{\xi}(\xi,y,t)=1+t G_{\xi}(\xi,y,t)=0\\
&\Delta_{\xi}=ntF^{n-2}(FF_{\xi\xi}+(n-1)F_{\xi}^2)=tG_{ \xi \xi  }=0\\
&\Delta_{y}=ntF^{n-2}(FF_{\xi y}+(n-1)F_yF_{\xi})=tG_{\xi y}=0\\
&u(x,y,t)=F(\xi,y,t)\\
&x=tF^n(\xi,y,t)+\xi.
\end{split}
\end{equation}
The point of gradient catastrophe is generic if
\[
\Delta_{\xi\xi}(\xi_c,y_c,t_c)\neq 0,\quad \Delta_{\xi y }(\xi_c,y_c,t_c)\neq 0 \quad \Delta_{ yy}(\xi_c,y_c,t_c)\neq 0.
\]
Furthermore, at the point of gradient catastrophe  it turns out  (numerically) that $F_{yy}$ remains bounded.
So  the solution  of the equation (\ref{eqF}) at the critical time 
satisfies  the necessary conditions
\begin{equation}
\label{constraints}
\begin{split}
&F^c_{t} \pm t_cG_y^c F^c_{y}=0,\quad F^c_{yt} \pm( t_c( G_y^cF^c_{yy}+G^c_{yy}F^c_{y}))=0,\quad  \\&
 F^c_{\xi t} \pm t_c( G_y^cF^c_{y\xi }+G^c_{y\xi }F^c_{y})=0.
\end{split}
\end{equation}
%The solution generically breaks at the point $(x_c,y_c,t_c)$ where  all  the derivatives of $u(x,y,t)$ breaks  except in the transversal direction carachterized by the vector field
%\[
%2t_cF_y^c\partial_x+\partial_y
%\]
Now we are going to study the analytic behaviour of the solution (\ref{dKP}) near the point $(x_c,y_c,t_c)$. Introducing the shifted variables
\[
x-x_c=\bar{x},\quad t-t_c=\bar{t},\quad y-y_c=\bar{y},\quad \xi-\xi_c=\bar{\xi}
\]
%and the rescaling 
%\begin{align}
%&\bar{x}\to\lambda \bar{x},\\
%\bar{t}\to \lambda\bar{t},\\
%
 
  we obtain the equation
\begin{equation}
\label{char1}
\begin{split}
&\bar{x}-\bar{t}(G^c+t_cG_t^c)-\bar{t}\bar{y}(G_y^c+t_cG^c_{yt})-t_c(G^c_{y}\bar{y}+\dfrac{1}{6}G_{yyy}^c\bar{y}^3+\dfrac{1}{2}G^c_{yy}\bar{y}^2)\\
&=\dfrac{t_c}{6}G^c_{\xi\xi\xi}\bar{\xi}^3+\dfrac{1}{2}t_cG^c_{\xi\xi y}\bar{y}\bar{\xi}^2+\dfrac{1}{2}(t_c\bar{y}^2G_{\xi yy}^c+(2 t_cG^c_{\xi t} +2G^c_{\xi})\bar{t})\bar{\xi}
+o(\bar{t}^2,\bar{y}^4,\bar{\xi}^4,\bar{t}(\bar{y}^2+\bar{\xi}^2))\\
%\bar{X}=\bar{x}-2\bar{t}(F^c+t_cF_t^c)-2\bar{t}\bar{y}(F_y^c+t_cF^c_{yt})-2t_c(F^c_{y}\bar{y}+\dfrac{1}{6}F_{yyy}^c\bar{y}^3+\dfrac{1}{2}F^c_{yy}\bar{y}^2)
\end{split}
\end{equation}
where the notation $G^c_{\xi\xi\xi}$ stands for $\left.\dfrac{\partial^3}{\partial \xi^3}G(\xi,y_c,t_c)\right|_{\xi=\xi_c}$ and analogous notations hold for the other quantities.
This suggests to introduce shifted variables 
(using $t_c = -1/G_{\xi}^c$):
\begin{equation}
\label{XT}
\begin{split}
&\zeta=G^c_\xi\left(\bar{\xi}+
\dfrac{G^c_{\xi\xi y}}{G^c_{\xi\xi\xi}}\bar{y}\right)\\
&X=\left[\bar{x}-\bar{t}(G^c+t_cG_t^c)-
\bar{t}\bar{y}(G_y^c+t_cG^c_{yt})-t_c\left(G^c_{y}\bar{y}+
\dfrac{1}{2}G^c_{yy}\bar{y}^2 + \dfrac{1}{6}G_{yyy}^c\bar{y}^3\right)\right.\\
&\left.-\frac{1}{3}t_c\dfrac{(G^c_{\xi\xi y})^3}{(G^c_{\xi\xi\xi})^2}\bar{y}^3+
\frac{1}{2}t_c\dfrac{G^c_{\xi\xi y}G^c_{\xi y y}}{G^c_{\xi\xi\xi}}\bar{y}^3+
G^c_\xi\dfrac{G^c_{\xi\xi y}}{G^c_{\xi\xi\xi}}\bar{y}\bar{t}\right]\\
&T=\left[\bar{t}+\frac{t^2_c}{2}\bar{y}^2
\left(\dfrac{(G^c_{\xi\xi y})^2}{G^c_{\xi\xi\xi}}-G^c_{\xi yy}\right)\right],\\
\end{split}
\end{equation} 
so that in the variable $\zeta$, (\ref{char1}) takes the form
\begin{equation}
\label{s_orig}
-\frac{k}{6}\zeta^3 + T\zeta=
X+o(\bar{t}^2,\bar{y}^4,\bar{\xi}^4,\bar{t}(\bar{y}^2+\bar{\xi}^2)),
\end{equation}
where
\begin{equation}
\label{k}
k=t_c^4 G^c_{\xi\xi\xi}.
\end{equation}
Using the estimates $\bar{\xi}\sim\bar{y}\sim\bar{t}^{1/2}$ and 
$\bar{x}\sim\bar{t}^{3/2}$ identified previously, the rescaling
\begin{equation}
\label{scaling_lambda}
\begin{split}
&X\to \lambda X\\
&\bar{t}\to\lambda^{\frac{2}{3}} \bar{t}\\
&\bar{y}\to\lambda^{\frac{1}{3}}\bar{y}\\
&\zeta\to\lambda^{\frac{1}{3}}\zeta,
\end{split}
\end{equation}
in the limit $\lambda\to 0$
reduces (\ref{s_orig}) to the 
universal  cusp singularity
\begin{equation}
\label{char2}
-\frac{k}{6}\zeta^3+T\zeta=X,
\end{equation}
 of the solution of a one-dimensional hyperbolic equation.
%{\bf Why is the cubic equation called ``one-dimensional hyperbolic equation"?}

To leading order in the limit $\lambda \rightarrow 0$, it is consistent 
to expand $u(x,y,t)$ to linear order in $\bar{\xi},\bar{y}$:
\beq
u(x,y,t)-u_c = F(\xi,y,t) - F^c \simeq F_{\xi}^c\bar{\xi}+F_y^c\bar{y}=
\dfrac{F_\xi^c}{G_{\xi}^c}\zeta(X,T) + \bar{\beta}\bar{y},
\label{u_exp}
\eeq
with 
\begin{equation}
\label{beta}
\bar{\beta} = F_y^c - \frac{F_{\xi}^cG_{\xi\xi y}^c}{G_{\xi\xi\xi}^c}.
\end{equation}
%%%%%%%
%%%%%%%%
%%%%%%%%
%%%%%%%%%
%%%%%%%%%
%%%%%%%%%%
%%%%%%%%%
% \subsection{Solution of the KP equation in the limit $\epsilon\to 0$}
\subsection{The small dispersion limit of the KdV equation}

%Higher order dispersion in NLS
%
%\[
%iu_z-\beta_2 u_{tt}-\frac{i}{6}u_{ttt}+\gamma |u|^2u=0
%\]
We conjecture that, before the critical time $t_c$ the solution  $u(x,y,t;\epsilon)$ of the generalized  KP equation (\ref{KP}) in the limit $\epsilon\to 0$ can be approximated by the solution  $u(x,y,t)$ of the generalized  dKP equation
 (\ref{dKP})  with the same $\epsilon$-independent initial data as long as the  gradients remain bounded, namely for $t<t_c$ one  is expected to have
\[
u(x,y,t;\epsilon)=u(x,y,t)+O(\epsilon),\quad t<t_c.
\]
We are interested to understand the behaviour of the KP solution $u(x,y,t;\epsilon)$ in a neighborhood of the critical point $(x_c,y_c,t_c)$.
For this purpose we first recall some results from the theory of the KdV equation.

We consider the small dispersion limit of the KdV equation

\[
u_t+\beta\, uu_x+\epsilon^2\rho \,u_{xxx}=0.
\]
where $\beta$ and $\rho $ are constants and $\epsilon$ is a small parameter.

For given  smooth rapidly decreasing initial data $u_0(x)$, such a 
limit has been extensively studied in the works \cite{LL}, 
\cite{Venakides}, \cite{DVZ}. It was pointed out in \cite{Dub1} and 
numerically shown in \cite{GK12} that some 
regions of  the $(x,t)$ plane escape the analysis of the small dispersion limit of the KdV equation. In particular one of these regions  is a neighbourhood of the critical point $(x_c,t_c)$ where 
the solution of the Hopf equation
\[
u_t+\beta uu_x=0
\]
obtained by setting $\epsilon=0$ in the KdV equation
has  a singularity. The solution of this equation for initial data $u_0(x)$ 
takes the form
\[
\left\{
\begin{array}{ll}
&u(x,t)=u_0(\xi)\\
&x=\beta u_0(\xi)t+\xi.
\end{array}
\right.
\]
The solution has a point of gradient catastrophe  at the time $t_c$ with
\[
t_c=\min_{\xi}\left(-\dfrac{1}{\beta u_0'(\xi) } \right).
\]
The minimum point $\xi_c$ gives $u_c=u_0(\xi_c)$ 
and the position  $x_c=t_c u_0(\xi_c)+\xi_c$. 
 The expansion near the critical points of the Hopf solution gives 
\[
\bar{x}:=x-x_c-\beta u_c(t-t_c)\simeq\beta(t-t_c)(u-u_c)-\dfrac{t_c^4\beta^3}{6}u_0'''(\xi_c)(u-u_c)^3.
\]
 It has been conjectured in \cite{Dub1} and proved in \cite{CG1}  
 (and for the KdV hierarchy in \cite{CG4}) that near the point $(x_c,t_c)$  the solution of KdV in the limit $\epsilon\to 0$  is approximated by 

\[
u(x,t;\epsilon)=u_c+\left(\dfrac{18\epsilon^2b}{\gamma^2}\right)^{\frac{1}{7}}U\left(\dfrac{(48)^{1/7}\bar{x}}{(\epsilon^6b^3\gamma)^{\frac{1}{7}}},\beta\dfrac{t-t_c}{(3^42^2\epsilon^4 b^2\gamma^3)^{\frac{1}{7}}}\right)+O(\epsilon^{\frac{4}{7}}),
\]
where 
\[
b=12\dfrac{\rho}{\beta} \quad \gamma=-t_c^4\beta^3 u_0'''(\xi_c)
\]
and the function $U({\cal X},{\cal T})$ satisfies the ODE (\ref{PI2}) 
with the  asymptotic conditions (\ref{PI2asym}).
%The existence and uniqueness for real ${\cal X}$ and ${\cal T}$ of a smooth  real solution of the PI2 equation (\ref{PI2})  satisfying the boundary conditions 
%(\ref{PI2asym}) has been conjectured in \cite{Dub1} and proved in \cite{CV}.
We remind that this  solution of the PI2 equation also satisfies the KdV equation
\[
U_{\cal T}+6UU_{\cal X}+U_{{\cal XXX}}=0.
\]

\subsection{The small dispersion limit of the KP equation}
Now let us consider the KP equation (\ref{KP}).
 We are looking for a solution of $u(x,y,t;\epsilon)$  near the point of  gradient catastrophe $(x_c,y_c,t_c)$ for the dKP equation  of the form
\[
u(x,y,t;\epsilon)= u_c+\dfrac{1}{nu_c^{n-1}} h(X,T;\epsilon)+\beta \bar{y}
\]
with $X$ and $T$ defined in (\ref{XT}) and $\beta$ is defined in (\ref{beta}) and $h(X,Y;\epsilon)$ a function to be determined.
We are interested in the multiscale expansion of this function of the form
\[
h(X,T;\epsilon)=\lambda^{\frac{1}{3}}H(\cal{X},\cal{T};\varepsilon)+O(\lambda)\]
\begin{equation}
\label{subs}
X=\lambda{ \cal X},\quad T=\lambda^{\frac{2}{3}}{ \cal T},\quad \epsilon=\lambda^{\frac{7}{6}}\varepsilon,\quad \bar{y}=\lambda^{\frac{1}{3}}{\cal Y}.
\end{equation} \begin{theorem}

Let 
\begin{equation}
\label{uexp}
u(x,y,t;\epsilon)= u_c+ \dfrac{1}{nu_c^{n-1}} h(X,T;\epsilon)+\beta \bar{y}\end{equation}
 be a solution of the generalized  KP equation (\ref{KP}).  
Suppose that the limit
\[
H({ \cal X}, {\cal T};\varepsilon)=\lim_{\lambda\to 0}\lambda^{-\frac{1}{3}}h(\lambda{ \cal X},\,\lambda^{\frac{2}{3}}{\cal T};\lambda^{\frac{7}{6}}\varepsilon)\]
exists. Then the function $H({\cal X},{\cal T};\varepsilon)$ satisfies the KdV equation
\begin{equation}
\label{KdV2}
H_{\cal{T}}+HH_{\cal{X}}+\varepsilon^2H_{\cal{X}\cal{X}\cal{X}}=0.
\end{equation}
\end{theorem}
\begin{proof}
Plugging the ansatz into the generalized  KP equation one obtains
\[
 \left(\dfrac{\partial }{\partial t} h(X,T;\epsilon)+(u_c+ \dfrac{F_{\xi}^c}{G_{\xi}^c} h(X,T;\epsilon)+\beta \bar{y})^n\dfrac{\partial }{\partial x} h(X,T;\epsilon)+\epsilon^2 \dfrac{\partial^3 }{\partial x^3} h(X,T;\epsilon)\right)_x=\pm 
 \dfrac{\partial^2 }{\partial y^2} h(X,T;\epsilon)
 \]
 Performing the rescalings (\ref{subs}) one obtains
\begin{align*}
&(H_{\cal{T}}+HH_{\cal{X}}+\varepsilon^2H_{\cal{X}\cal{X}\cal{X}})_{{\cal X}}+\lambda^{-\frac{1}{3}}H_{{\cal X}{\cal X}}\left( \dfrac{\partial X}{\partial t}\mp\left(\dfrac{\partial X}{\partial y}\right)^2+G_c+\beta nu_c^{n-1}\bar{y}\right)\\
&=\pm 2H_{\cal{T}\cal{X}}\dfrac{\partial T}{\partial y}\dfrac{\partial X}{\partial y}\pm \lambda^{\frac{1}{3}}H_{\cal{T}\cal{T}}\left(\dfrac{\partial T}{\partial y}\right)^2\pm \lambda H_{T}\dfrac{\partial^2  T}{\partial y^2}\pm\lambda^{\frac{2}{3}}H_{\cal{X}}\dfrac{\partial^2  X}{\partial y^2}\\
&-\sum_{k=2}^n{n \choose k} \lambda^{\frac{k-1}{3}}u_c^{n-k}\left((\dfrac{F_{\xi}^c}{G_{\xi}^c}\mathcal{H}+\beta \mathcal{Y})^k{\mathcal H}_{{\mathcal X}}\right)_{{\mathcal X}}
\end{align*}
Using  (\ref{XT}) and taking into account the constraints (\ref{constraints})  one arrives at the relation
\[
(H_{\cal{T}}+HH_{\cal{X}}+\varepsilon^2H_{\cal{X}\cal{X}\cal{X}})_{{\cal X}}=O(\lambda^{\frac{1}{3}}),
\]
which in the limit $\lambda\to 0$ implies that 
\begin{equation}
\label{KdVH}
H_{\cal{T}}+HH_{\cal{X}}+\varepsilon^2H_{\cal{X}\cal{X}\cal{X}}=const
\end{equation}
In order to show that the constant is equal to zero it is sufficient to observe that the solution  $H({\cal X}, {\cal T};\epsilon) $ of (\ref{KdVH})  has to match  the outer solution (\ref{u_exp}) when ${\cal X}\to \infty$.
\end{proof}

 We observe that choosing $\lambda=\epsilon^{\frac{6}{7}}$  one has $\varepsilon=1$ in (\ref{KdV2}). 
 %The rescaled function $h(X,T)$ satisfies the KdV equation near the critical point $(x_c,y_c, t_c)$. Indeed  
 In the rescaled variables
 \[
{\cal X}= \frac{X}{\epsilon^{\frac{6}{7}}},\quad {\cal T}=  \dfrac{T}{\epsilon^{\frac{4}{7}}},
 \]
 the function $H({\cal X},{\cal T};\varepsilon)$ satisfies the KdV equation (\ref{KdV2}) with  $\varepsilon=1$.
 Furthermore for $\epsilon\to 0$ and fixed $X$ and $T$ the solution has to match the outer solution  (\ref{u_exp}), namely 
 \[
H({\cal X},{\cal T};\varepsilon)\simeq \mp \left(\dfrac{6}{k}\right)^{\frac{1}{3}}|{\cal X}|^{\frac{1}{3}}\mp 2{\cal T}\left(\dfrac{6}{k}\right)^{-\frac{1}{3}}|{\cal X}|^{-\frac{1}{3}}+O({\cal X}^{-1}),\quad |{\cal X}|\to \infty.
  \]
%  Since there is only one smooth solution of the KdV equation satisfying such asymptotic conditions and it is given by the PI2 solution,
Using the results of the previous section on the KdV equation, we arrive at the following conjecture.
\begin{conjecture}   Let us assume that the solution $u(x,y,t;\epsilon)$  to the Cauchy problem for  the generalized  KP  equation
\[
(u_t+u^nu_x+\epsilon^2 u_{xxx})_x=\pm u_{yy}
\]
with $\epsilon$-independent initial data
\[
u(x,y,t=0;\epsilon)=u_0(x,y),
\]
is at least $C^4$ in $\mathbb{R}^3_+=\{(x,y,t)\,|\,t>0\}$. Then the solution $u(x,y,t;\epsilon)$ 
admits the following expansion near the critical point $(x_c,y_c,t_c)$ and $u_c=u(x_c,y_c,t_c)$ for the solution of the dKP equation $ (u_t+u^nu_x)_x=\pm u_{yy}.$
 In the limit $\epsilon\to0$ and $x\to x_c$, $y\to y_c$ and $t\to t_c$ in such a way that the limits 
\[
\lim\frac{X}{\epsilon^{\frac{6}{7}}},\quad \lim \dfrac{T}{\epsilon^{\frac{4}{7}}},\quad \lim \dfrac{y-y_c}{\epsilon^{\frac{2}{7} } }
\]
remain finite, with $X$ and $T$ defined in (\ref{XT}), the solution of the generalized KP equation (\ref{KP}) is approximated by

%\[
%u=u_c\pm \left(\dfrac{18\epsilon^2b}{\gamma^2}\right)^{\frac{1}{7}}U\left(\dfrac{(48)^{1/7}\bar{x}}{(\epsilon^6b^3\gamma)^{\frac{1}{7}}},\beta\dfrac{t-t_c}{(3^42^2\epsilon^4 b^2\gamma^3)^{\frac{1}{7}}}\right)+O(\epsilon^{\frac{4}{7}})
%\]
%where 
%\[
%b=12\dfrac{\rho}{\beta} \quad \gamma=-t_c^3\beta^3 u_0'''(\xi_c)
%\]

\begin{equation}
\label{KP12}
u(x,y,t;\epsilon)= u_c+\dfrac{6}{nu_c^{n-1}} \left(\dfrac{\epsilon^2}{\kappa^2}\right)^{\frac{1}{7}}U\left(\frac{X}{( \kappa \epsilon^6)^{1/7}},\frac{T}{( \kappa^3\epsilon^4)^{1/7}}\right)+\bar{y}(F_y-F_\xi\dfrac{G_{\xi\xi y}}{G_{\xi\xi\xi}})+O(\epsilon^{\frac{4}{7}})
\end{equation}
where 
\[
\kappa=-36\,G^c_{\xi\xi\xi}t_c^4
\]
and $U({\cal X},{\cal T})$  is the  particular solution of the P1-2 equation (\ref{PI2asym}) described in the introduction.
\end{conjecture}
%{\color{red}
%In the previous paper Dubrovin-Grava-Klein for generalized NLS,  the constant $\tilde{k}$ is equal to 
%\[
%k=6\dfrac{\tilde{k}}{(nu_c^{n-1})^3}
%\]
%}
%{\color{red}
%Second formula
%\begin{equation}
%\label{KP12}
%u(x,y,t;\epsilon)\simeq u_c+\dfrac{1}{nu_c^{n-1}} 6^{\frac{3}{7}}\left(\dfrac{\epsilon^2}{k^2}\right)^{\frac{1}{7}}U\left(\frac{X}{(6^2 k\epsilon^6)^{1/7}},\frac{T}{(6^6 k^3\epsilon^4)^{1/7}}\right)+\bar{y}(F_y-F_\xi\dfrac{G_{\xi\xi y}}{G_{\xi\xi\xi}}),
%\end{equation}
%}
In the particular case in which the initial data is an even function 
of $y$, namely $u_0(x,y)=u_0(x,-y)$ one can easily check that this property is preserved by the KP equation, namely
$u(x,y,t,\epsilon)=u(x,-y,t,\epsilon)$ and therefore all the odd derivatives with respect to $y$ of the function $F$ vanish.
The formula (\ref{KP12}) can be simplified to the form
\begin{equation}
\label{KP12e}
u(x,y,t;\epsilon)\simeq u_c+\dfrac{6\left(\dfrac{\epsilon^2}{\kappa^2}\right)^{\frac{1}{7}}}{nu_c^{n-1}}U\left(\frac{\bar{x}-u^n_c\bar{t}-\dfrac{t_c}{2}G^c_{yy}\bar{y}^2}{(\kappa\epsilon^6)^{1/7}},\frac{\bar{t}-\frac{t^2_c}{2}\bar{y}^2G_{\xi yy}^c}{(\kappa^3\epsilon^4)^{1/7}}\right)+O(\epsilon^{\frac{4}{7}}),
\end{equation}
where $\bar{x}=x-x_c$, $\bar{t}=t-t_c$ and $\bar{y}=y-y_c$  and $G=F^n$.

\section{Numerical approaches}\label{sec2}
The numerical task in this paper is to solve the generalized KP 
equation (\ref{KP}) in \emph{evolutionary form},
\begin{equation}
    u_{t}+u^{n}u_{x}+\epsilon^{2}u_{xxx}=\sigma 
    \partial_{x}^{-1}u_{yy},\quad \sigma=\pm1
    \label{gKPev},
\end{equation}
and the generalized  dKP equation (equation (\ref{gKPev}) after 
formally putting $\epsilon=0$) after the transformation 
(\ref{implicit0}), i.e., (\ref{eqF}) in evolutionary form,
\begin{equation}
    F_{t}=\sigma\left((1+tnF^{n-1}F_{\xi})\partial_{\xi}^{-1}F_{yy}-tnF^{n-1}F_{y}^{2}\right)
    \label{eqFev};
\end{equation}
the nonlocal operators $\partial_{x}^{-1}$ and $\partial_{\xi}^{-1}$ are defined 
in Fourier space by their respective singular Fourier symbols $-i/k_{x}$ and 
$-i/k_{\xi}$ respectively where $k_{x}$ and $k_{\xi}$ are the Fourier 
variables dual to $x$ and $\xi$ respectively. To avoid problems with 
these singular Fourier symbols, we will  always consider 
initial data with $\partial_{x}^{-1}u(x,y,0)\in \mathcal{S}(R^{2})$, 
i.e., initial data which are the $x$-derivative of a function in the 
Schwartz space of rapidly decreasing smooth functions.

As discussed in \cite{KR11}, for such initial data it is convenient 
to use Fourier methods. Denoting with  $\hat{u}$ the 2-dimensional 
Fourier transform of $u$, equations (\ref{gKPev}) and (\ref{eqFev}) can be written
in the form
\begin{equation}
    \hat{u}_{t}=\mathcal{L}\hat{u}+\mathcal{N}(\hat{u}),
    \label{uhat}
\end{equation}
where $\mathcal{L}$ is a linear, {\it diagonal} operator, which is 
$ik_{y}^{2}/k_{\xi}$ for (\ref{eqFev}), and 
$ik_{y}^{2}/k_{x}+i\epsilon^{2} k_{x}^{3}$ 
for (\ref{gKPev}), and $\mathcal{N}(\hat{u})$ is a nonlinear term.
The idea of an \emph{exponential time differencing} (ETD) scheme is to treat the 
linear part of (\ref{uhat}) exactly. We use the fourth order ETD 
method by Cox and Matthews \cite{CM02}, but other schemes offer 
a very similar performance as discussed in \cite{KR11}. 

The Fourier 
transform will be approximated in standard manner by discrete Fourier 
transforms. This means the solution will be treated as essentially 
periodic on the domain $L_{x}[-\pi,\pi]\times L_{y}[-\pi,\pi]$, where 
$L_{x}$, $L_{y}$ are chosen large enough that the function and the 
discrete Fourier transform decrease to machine precision (here 
$10^{-16}$) if possible. We are working here on serial computers and 
can access a resolution of $N_{x}N_{y}=2^{15}$. Note that the nonlocal operators in (\ref{gKPev}) and (\ref{eqFev}) 
imply that solutions for this equation for generic initial data in 
$\mathcal{S}(\mathbb{R}^{2})$ will not stay in this space, but will 
develop tails with an algebraic fall off towards infinity, see the 
discussion in \cite{KSM} and references therein. The resulting loss 
of regularity 
%of the analytically as periodic continued function
 at the domain boundaries leads to a slower decrease of the Fourier 
coefficients than for an exponentially decreasing function. Thus one is 
forced to use higher resolution or larger domains for KP than for 
KdV solutions where the solution for initial data in 
$\mathcal{S}(\mathbb{R})$ stays in this space. As mentioned, 
the nonlocal terms 
in (\ref{gKPev}) and (\ref{eqFev}) correspond to singular Fourier 
symbols. To compute their action numerically, we regularize these 
symbols by adding some constant of the order of machine precision. In 
addition we use for (\ref{eqFev}) Krasny filtering \cite{Kra}, i.e., Fourier 
coefficients with a modulus smaller than $10^{-10}$ are put equal to 
zero.  This is necessary for (\ref{eqFev}) since the nonlocality 
there appears not only in the linear part that is treated exactly in 
ETD schemes.

The decrease of 
the Fourier coefficients allows in any case to control the spatial 
resolution in the numerical solution. The accuracy in time is 
controlled via the $L^{2}$ norms of the solutions $u$ and $F$ which 
are both exactly conserved for solutions of (\ref{gKPev}) and 
(\ref{eqFev}). Due to unavoidable numerical errors, the $L^{2}$ norm 
of the numerical solutions will depend in general on time. 
We use the relative computed $L^{2}$ norm denoted by 
$\Delta_{2}$ as an indicator of the numerical accuracy of the 
solution. As discussed in \cite{KR11} this quantity is a valid 
criterion for sufficient spatial resolution and overestimates the 
numerical accuracy by one to two orders of magnitude. 

Since the focus of this paper is on the critical behavior of the 
solution to generalized  KP equations near the critical points (\ref{GC}) of the solutions 
to the corresponding dispersionless equations, it is crucial to 
obtain these critical points and the solution at this point with high 
accuracy. To this end we run the code for (\ref{eqFev}) with a 
resolution of $N_{x}=2^{9}$ and $N_{y}=2^{11}$  for the data 
symmetric with respect to $y\mapsto -y$ and with $N_{x}=N_{y}=2^{10}$ 
for the non symmetric data with 
$N_{t}=1000$ time steps to some estimated break up time. Once the 
quantity $\Delta$ in (\ref{GC}) becomes negative, the code is 
stopped. Then it  is 
restarted with the same parameters between the last time $t_{i}$ with positive 
$\Delta$ and the first time $t_{e}$ with negative $\Delta$ with the 
computed solution at $t_{i}$ as initial data (this means the time 
resolution for this step has been increased by a factor 1000). If 
needed, this procedure is iterated. Note that the Fourier 
coefficients in all studied examples for (\ref{eqFev}) decrease to the order of the 
Krasny filter during the whole computation, and the computed $L^{2}$ 
norm of the solution to (\ref{eqFev}) is conserved to better than 
$10^{-12}$. 

It turns out that the precision obtained in this way for the critical 
time $t_{c}$ is sufficient to assure that equations (\ref{GC}) can be 
satisfied to better than plotting accuracy ($10^{-3}$). If much 
higher precision were needed, an interpolation of the solution on the 
recorded time steps could be used. To identify the critical values 
$x_{c}$ and $y_{c}$ via (\ref{GC}), an interpolation in the spatial 
coordinates is, however, needed (the interpolation is done by using 
the exact representation of the truncated 
Fourier series). To solve the equations $G_{\xi\xi}=G_{\xi y}=0$, 
we identify the coordinates of the minimum of $\Delta$ at the found 
critical time. These values are used as the starting values for an 
iterative solution  via the algorithm 
\cite{fminsearch} distributed in Matlab as \emph{fminsearch}. The 
equations are solved to machine precision. At the  identified 
critical point, the derivatives entering formula (\ref{KP12}) are 
computed. 

The PI2 transcendent appearing in the asymptotic formula (\ref{KP12}) 
is computed as detailed in \cite{GKK} to essentially machine 
precision. 

\section{Solutions to the Kadomtsev--Petviashvili equations near the 
critical points}
\label{sec:KP}
In this section we study solutions to the KP I and II equations for 
$\epsilon=0.01$ near 
a critical point of the solutions to the  corresponding 
dispersionless equations for the same initial data. The solutions are 
compared to the asymptotic formula (\ref{KP12}). 
Throughout the paper, we consider the initial data
\begin{equation}
    u_{0}(x,y) = -6\partial_{x}\mbox{sech}^{2}(x^2+y^2),
    \label{u0sym}
\end{equation}
which are symmetric with respect to $y\mapsto -y$, and 
\begin{equation}
    u_{0}(x,y) = 6\partial_{x}\exp(-x^{2}-5y^{2}-3xy)
    \label{u0asym},
\end{equation}
which do not have such a symmetry. 
Solutions to the dKP equation for these initial data have been 
discussed in detail in \cite{GKE}, where also figures can be found. 
Therefore we will concentrate here 
on the related KP solutions. 

\subsection{Symmetric initial data}
In Fig.~\ref{KP1e4sym} we show the solution to the KP I and II equation with 
$\epsilon=0.01$ for the 
symmetric initial data (\ref{u0sym}) at the critical time $t_{c}\approx 0.222$. 
We use $N_{x}=2^{13}$ and $N_{y}=2^{10}$ Fourier modes and 
$N_{t}=1000$ time steps. Note that while a solution to the dKP I equation (\ref{dKP}) 
gives under the simultaneous transformation $x\mapsto -x$, 
$u\mapsto -u$ a solution to the dKP II equation, this is not the case 
for the full KP equation. As can be seen in Fig.~\ref{KP1e4sym}, the KP I 
solution for localized initial data develops a tail with algebraic 
decrease to the right, 
whereas such a tail goes to the left for KP II solutions. Due to the 
imposed periodicity, these tails reenter the computational 
domain on the opposing side. Since we are interested here in the 
formation of dispersive shocks that are not affected by this 
reentering of the tails, we can use the shown domain size (the 
Fourier coefficients for both cases decrease to at least $10^{-7}$).  
The numerically computed $L^{2}$ norm of the solution is conserved to 
better than $10^{-10}$. The oscillations are hardly visible on these plots which is why we show only 
close-ups in the following. 
\begin{figure}
    \includegraphics[width=0.49\textwidth]{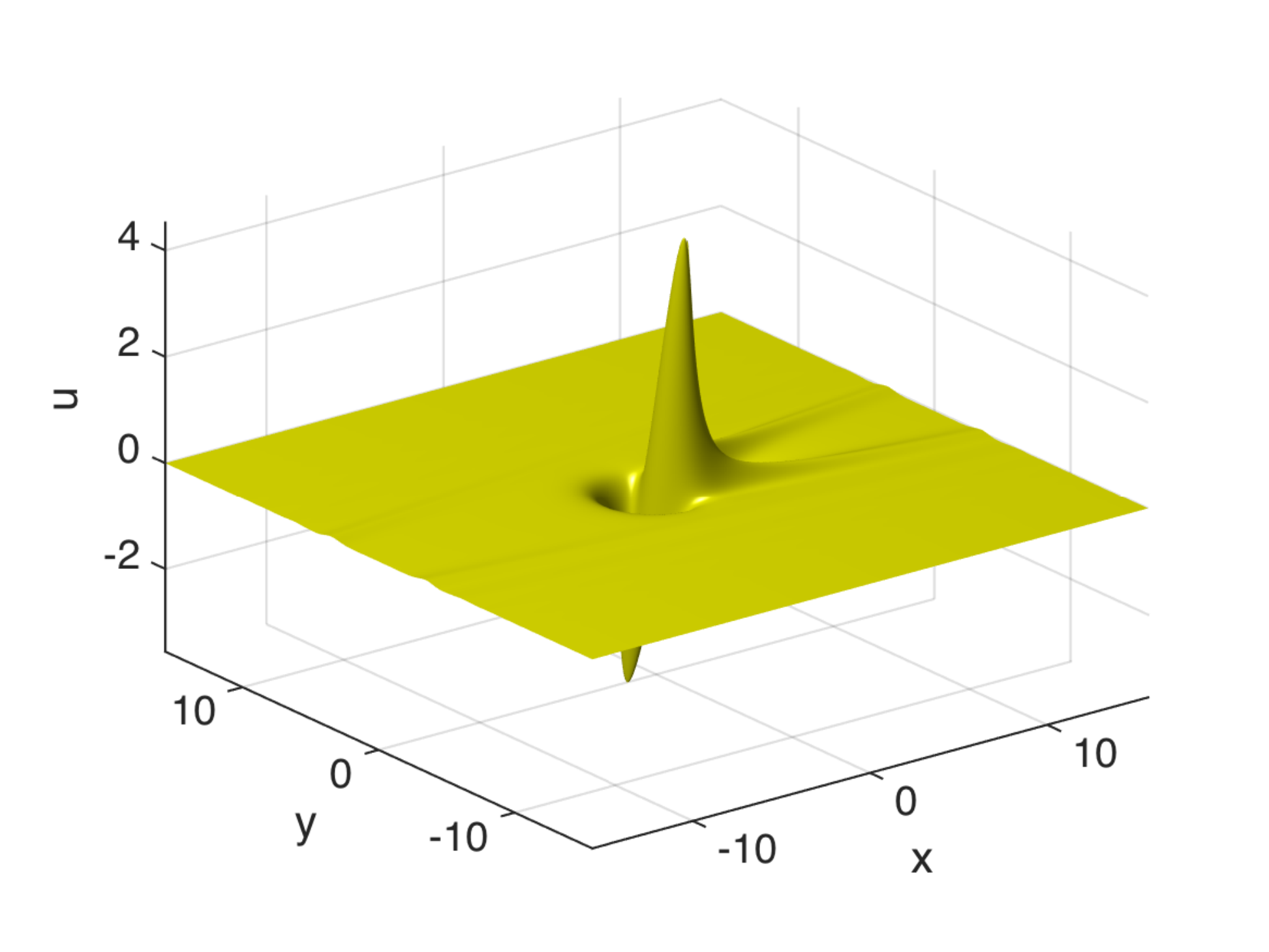}
 \includegraphics[width=0.49\textwidth]{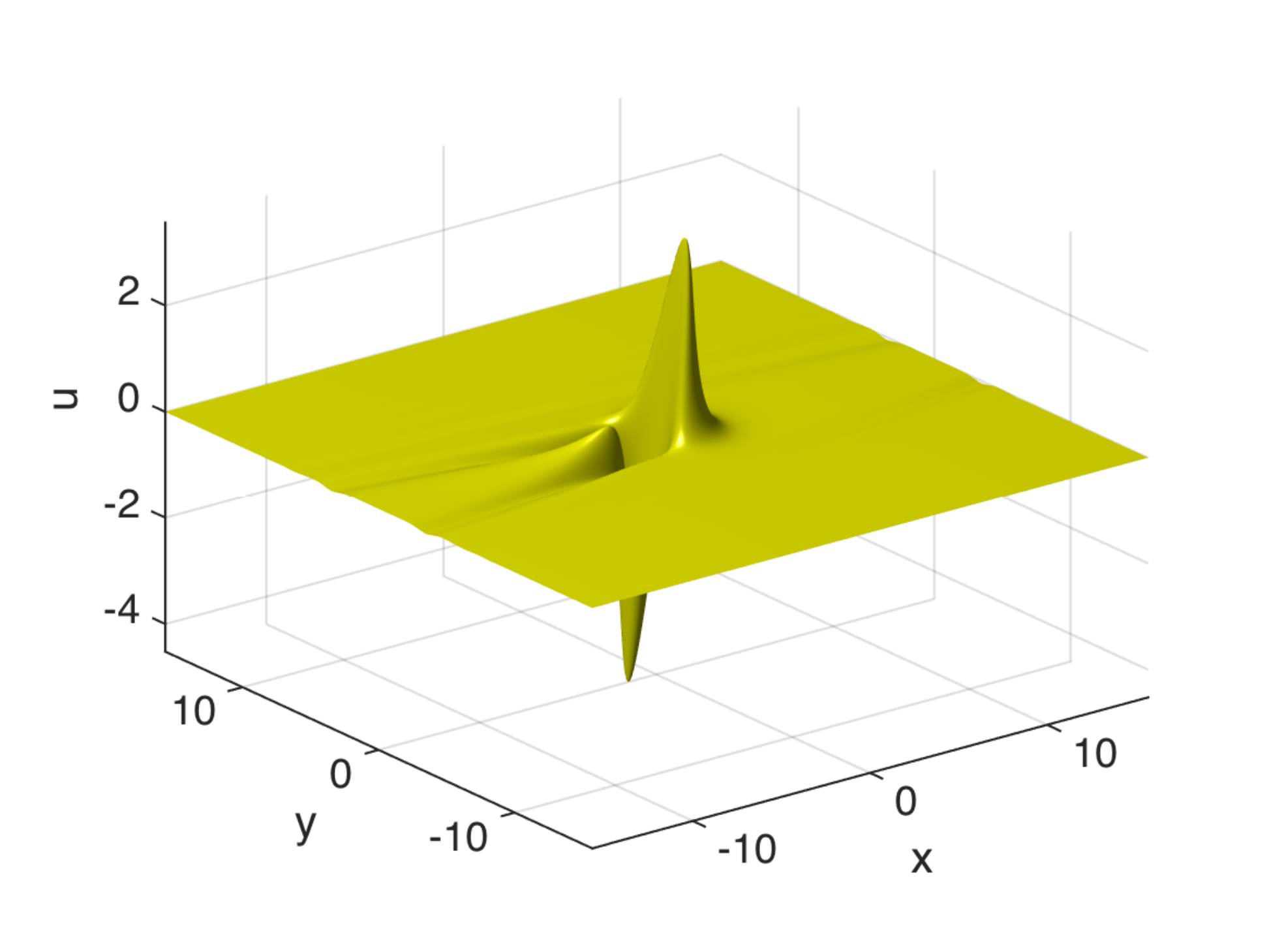}   
 \caption{Solutions to the KP equation  with $\epsilon=0.01$ for the initial data 
 (\ref{u0sym}) at the critical time $t_{c}\approx 0.222$; on the left 
 for KP I, on the right for KP II.   }
 \label{KP1e4sym}
\end{figure}

The mass transfer to infinity via these 
tails implies stronger gradients on the respective side and leads to 
a break-up exactly on this side. Consequently also 
the first oscillation of the KP 
solutions are observed on the side of the tails. For dKP I the 
first critical point is $x_{c}\approx 1.79$ and $y_{c}=0$. First 
oscillations appear around the critical time near this point as can 
be seen in Fig.~\ref{KPIc}. Near the critical point,  the 
KP I solution is well approximated by the asymptotic solution 
(\ref{KP12}) in terms of the PI2 transcendent. The approximation is 
local (for small $|x-x_{c}|$ and small $|y-y_{c}|$), but it can be 
seen that even the first oscillation will be approximately captured 
by small enough $\epsilon$. Also the $y$-dependence is well 
reproduced, but gets worse for larger values of $|y-y_{c}|$.
\begin{figure}
    \includegraphics[width=0.49\textwidth]{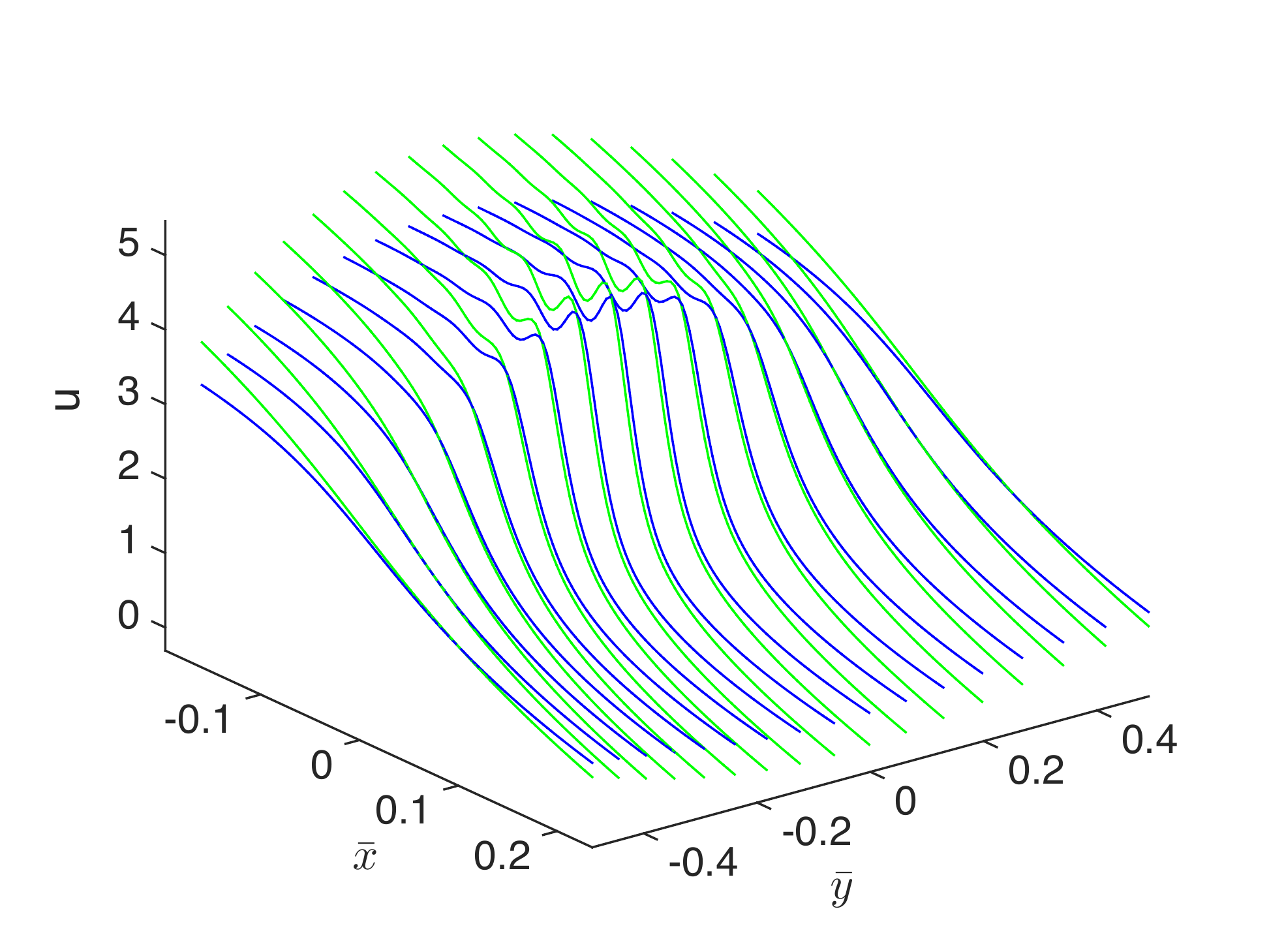}
    \includegraphics[width=0.49\textwidth]{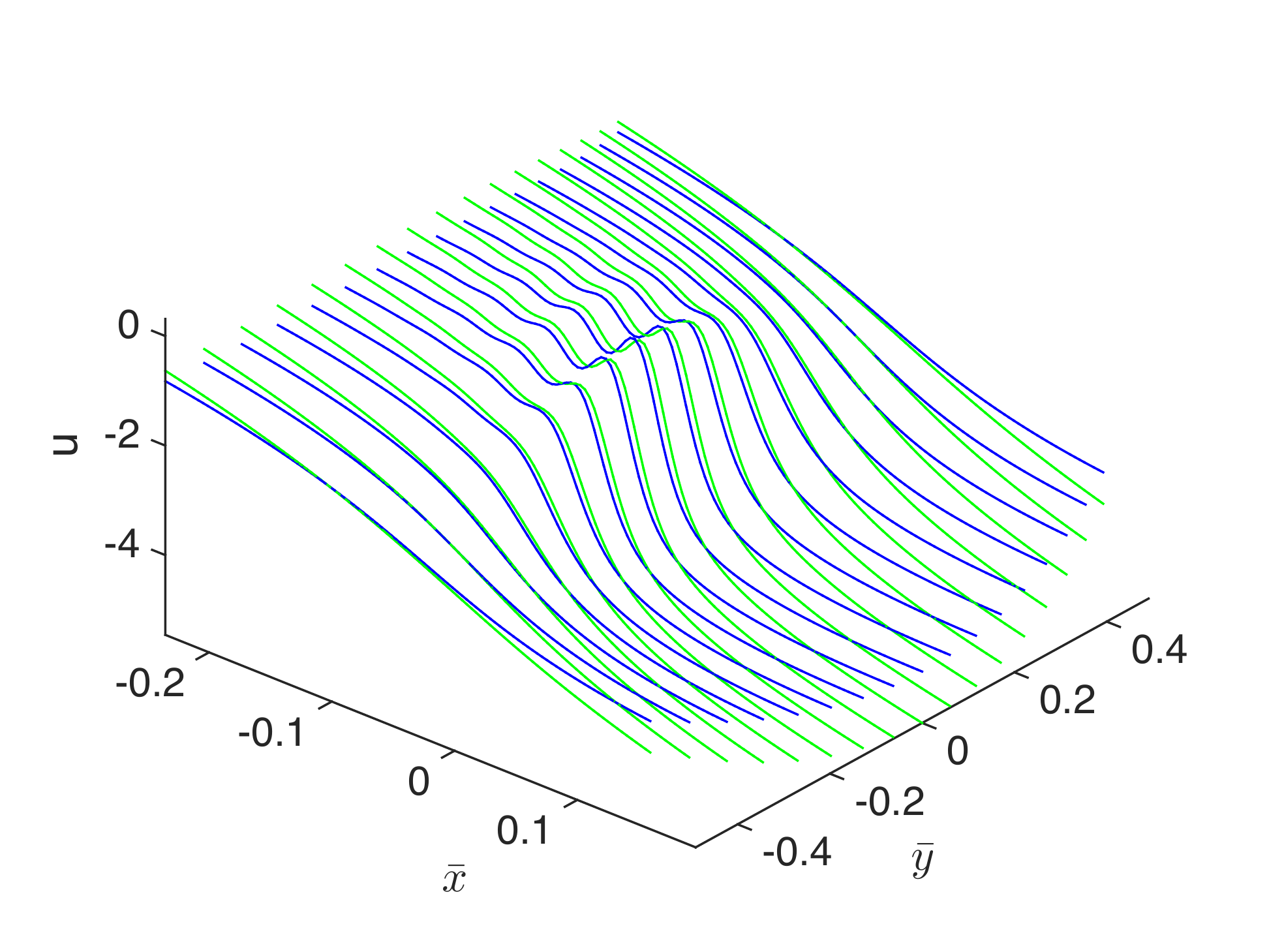}
 \caption{Solutions to the KP  equations  with $\epsilon=0.01$ for the symmetric initial data 
 (\ref{u0sym}) at the critical time $t_{c}\approx 0.222$ near the 
 critical points $x_{c}\approx \pm 1.79$ and $y_{c}=0$ in blue and the 
 corresponding PI2 asymptotic solution (\ref{KP12}) in green; on the left 
 the KP I solution, on the right the KP II solution.   }
 \label{KPIc}
\end{figure}

The asymptotic description (\ref{KP12}) via the PI2 transcendent   of the KP 
solutions is local near the critical point $(x_c,y_c,t_c)$. For the spatial 
dependence, this can be seen in Fig.~\ref{KPIc} at the critical 
time. However the description is also valid for small 
$|t-t_{c}|$ as is visible in Fig.~\ref{KPct} for times before and after 
the critical time. It can be seen that the PI2 asymptotics is shifted 
slightly to the left for KP I as before and moves with higher speed 
than the KP I solution to right. Thus the approximation becomes close 
the inflection point better with time, but the oscillations are less 
well reproduced. Note, however, that there appears to be the same 
number of oscillations of the KP and PI2 solution. 
\begin{figure}
    \includegraphics[width=0.32\textwidth]{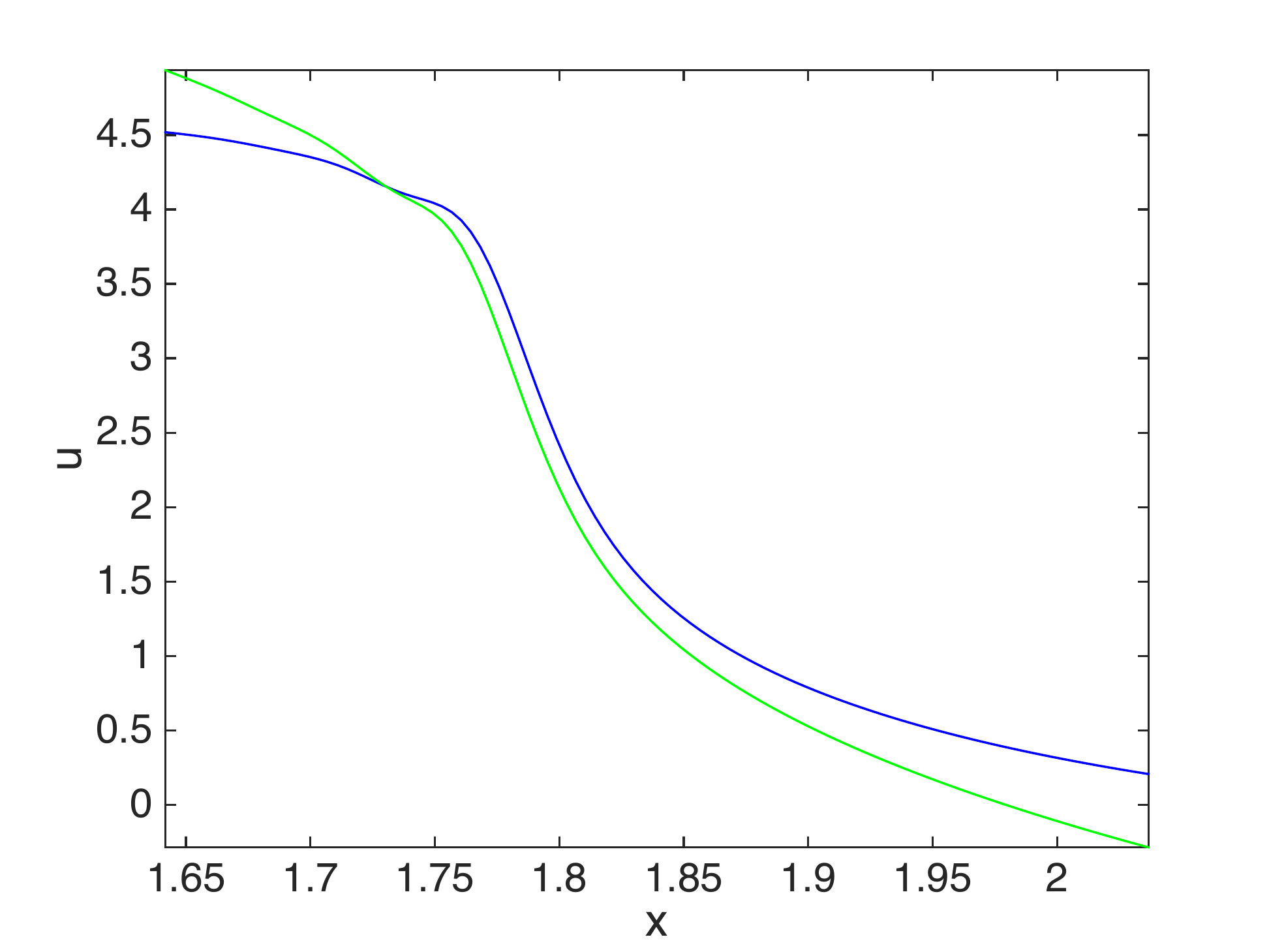}
 \includegraphics[width=0.32\textwidth]{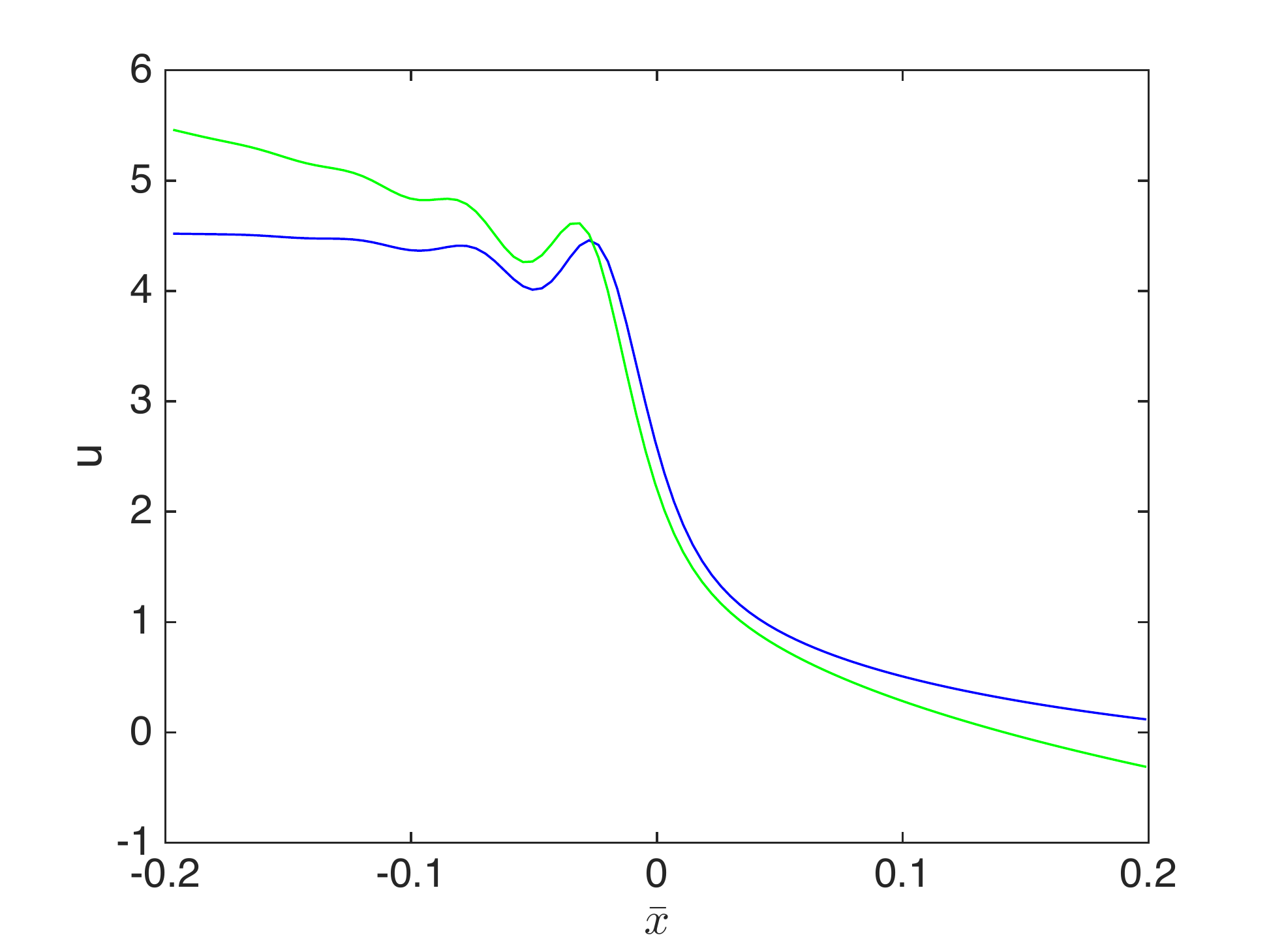}   
 \includegraphics[width=0.32\textwidth]{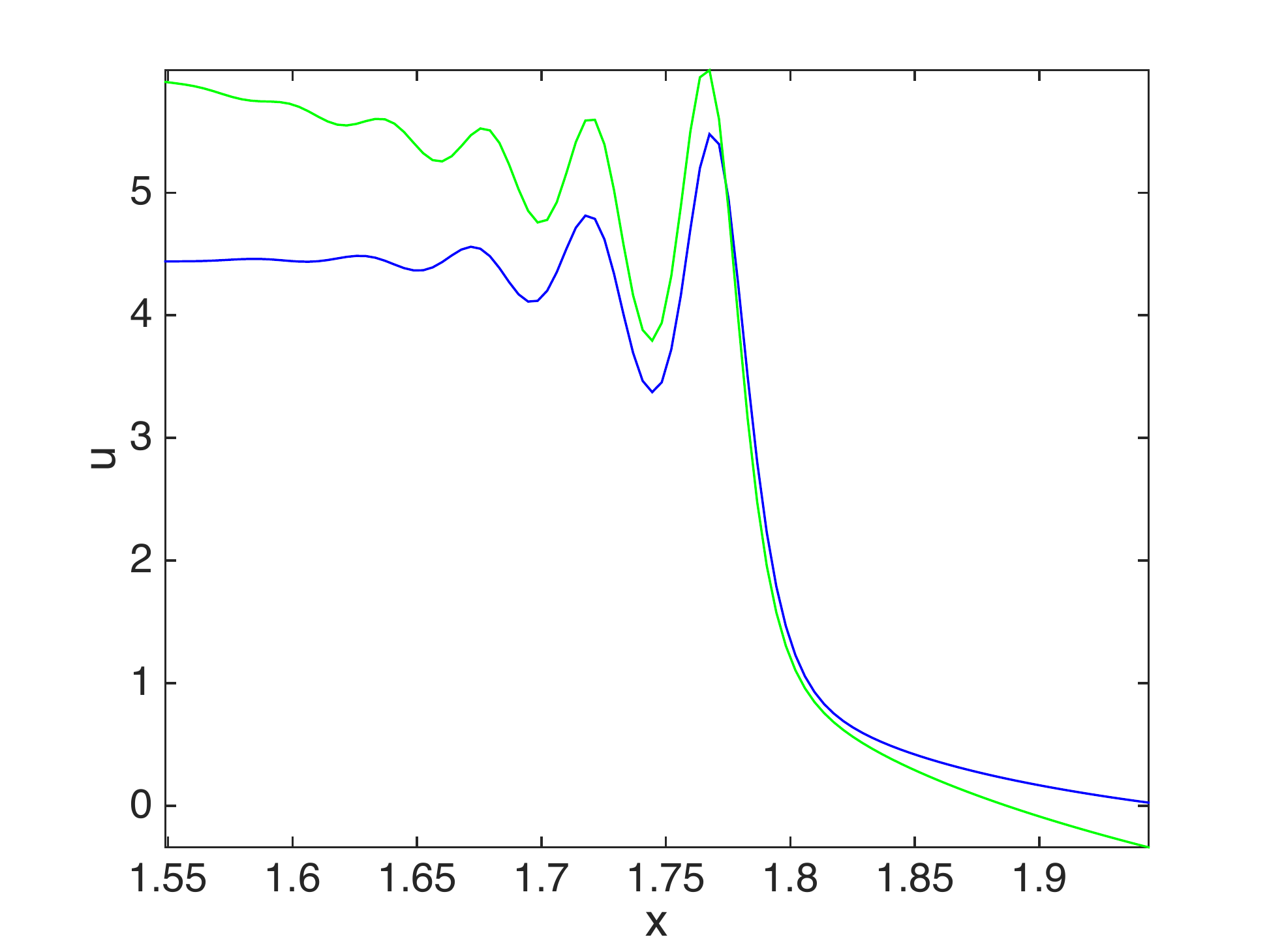}   \\ \includegraphics[width=0.32\textwidth]{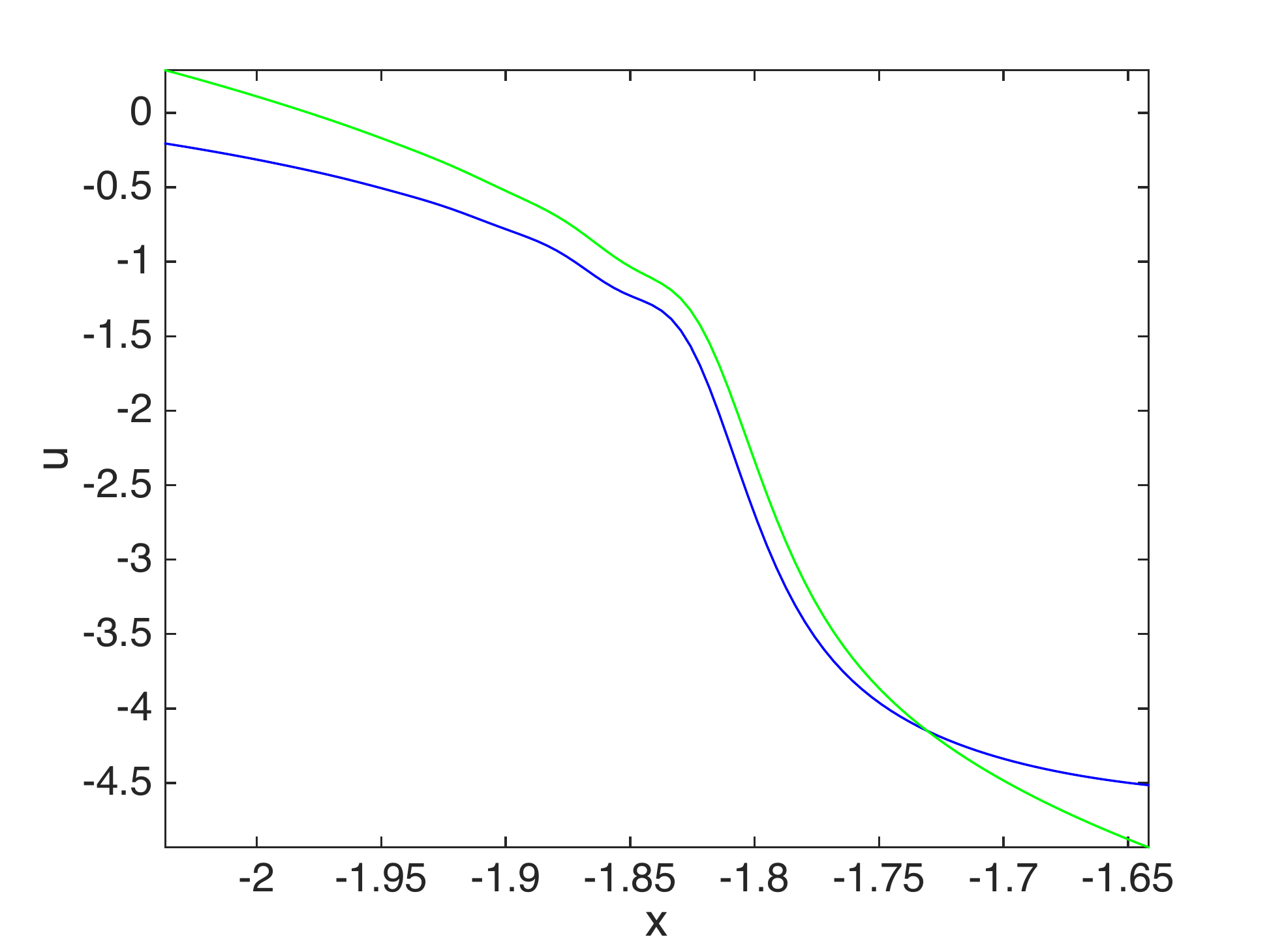}
 \includegraphics[width=0.32\textwidth]{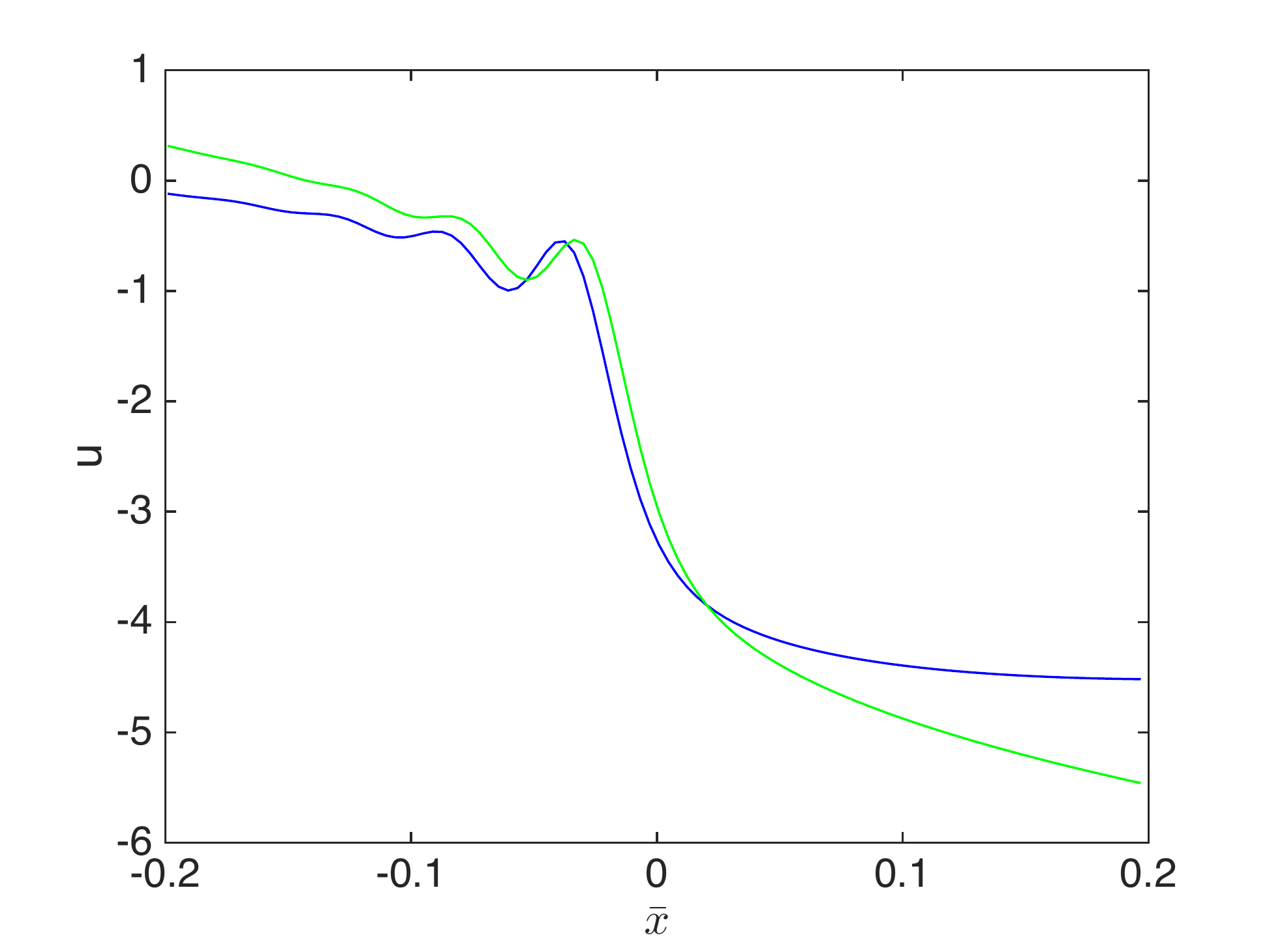}   
 \includegraphics[width=0.32\textwidth]{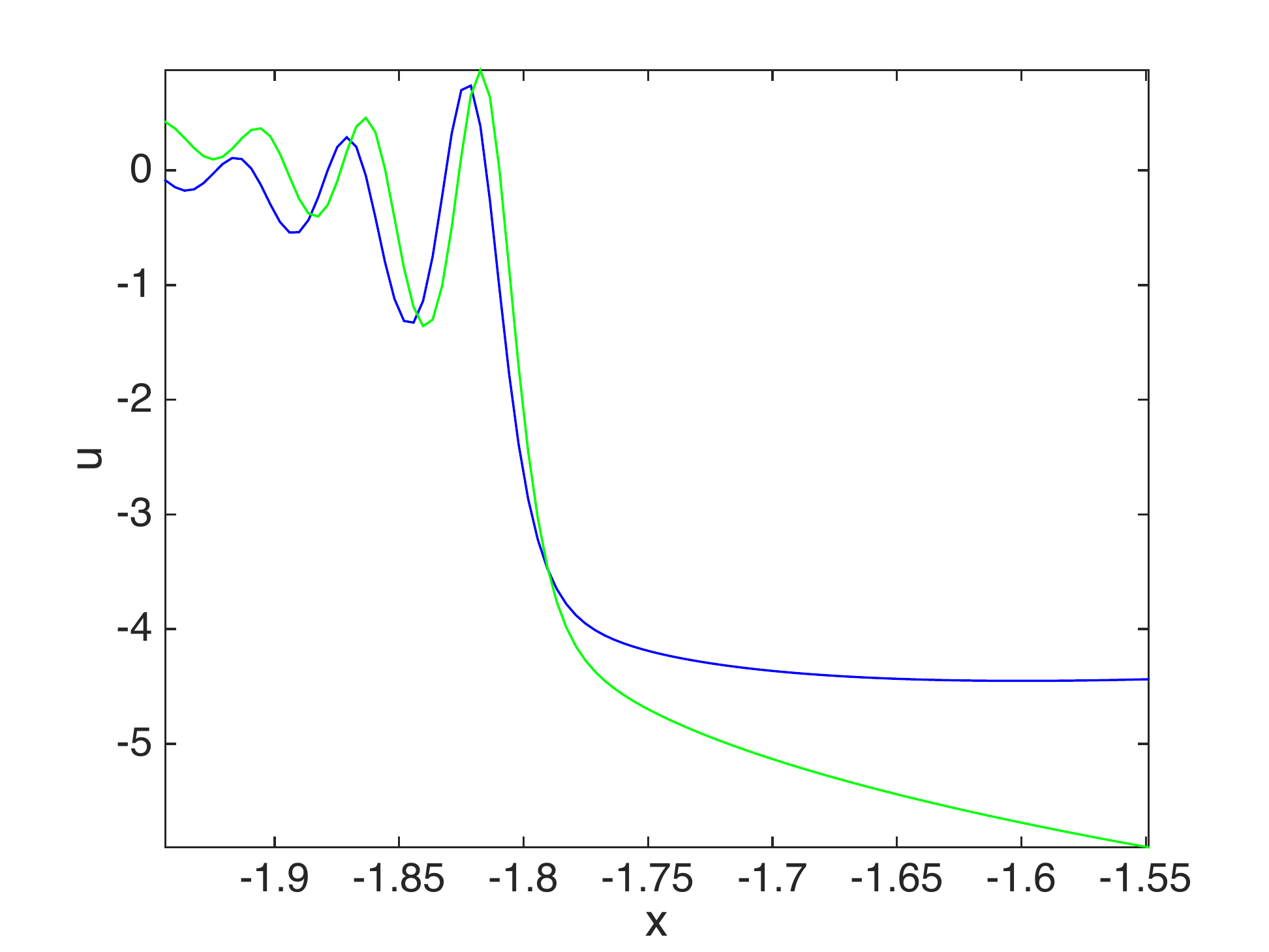}   
 \caption{Solutions to the KP I  equation in the upper row and KP~II 
 equation in the lower row  with $\epsilon=0.01$ for the symmetric initial data 
 (\ref{u0sym}) on the left for $t=0.204<t_{c}$, in the center for 
 $t=t_{c}$ and on the right for 
 $t=0.24>t_{c}$ on the $y$-axis near the 
 critical point $x_{c}\approx -1.79$ in blue and the 
 corresponding  asymptotic solution (\ref{KP12}) in terms of the PI2 transcendent  in green.   }
 \label{KPct}
\end{figure}

The oscillatory zone of the corresponding KP II solution can be seen 
in Fig.~\ref{KPIc} on the right. It appears near $x_{c}\approx-1.79$, and due to 
the symmetry between dKP I and dKP II solutions, it is again 
decreasing with $x$. Note, however, that the asymptotic solution 
(\ref{KP12}) is slightly shifted towards more positive values of $x$ 
with respect to the KP II solution, whereas it is slightly shifted to 
more negative $x$-values for KP I. The quality of the approximation 
is the same in both cases.  The time dependence of the asymptotic 
description can be seen in Fig.~\ref{KPct} on the $y$-axis.
The PI2 asymptotics is slightly shifted to the 
right for $t\leq t_{c}$ and travels with somewhat higher speed than the KP II
dispersive shock front to the left. Thus the approximation of the shock 
front becomes again better with time, and this time also the oscillations 
are better approximated.  
%\begin{figure}
%    \includegraphics[width=0.32\textwidth]{KPII1e46secht0204y0.pdf}
% \includegraphics[width=0.32\textwidth]{KPII1e46sechcy0.pdf}   
% \includegraphics[width=0.32\textwidth]{KPII1e46secht024y0.pdf}   
% \caption{Solutions to the KP II  equations with $\epsilon=0.01$ for the initial data 
% (\ref{u0sym}), on the left for $t=0.0204<t_{c}$, in the center for 
% $t=t_{c}$,  and on the right for 
% $t=0.24>t_{c}$ on the $y$-axis near the 
% critical points $x_{c}\approx -1.79$ in blue and the 
% corresponding asymptotic solution (\ref{KP12}) in green.   }
% \label{KPIIct}
%\end{figure}

\subsection{Second critical point for symmetric initial data}
The initial data (\ref{u0sym}) are odd functions in $x$. For the 
$1+1$ dimensional Hopf equation, a possible break up would occur at 
the same time at the points  $\pm x_{c}$. For dKP solutions for such 
initial data, this is no longer true. But the method detailed in 
\cite{GKE} to numerically integrate the dKP equation also allows to 
reach the second breaking of the dKP solution for the data 
(\ref{u0sym}). It occurs at the time $t_{c}\approx 0.3001$ at 
$x_{c}\approx -2.033$ for KP I and $-x_{c}$ for KP II. The  behaviour of the solution of 
 KP I  near the critical point $(x_c,t_c,y_c)$ is shown in 
 Fig.~\ref{KPc2} in the upper row. It can be seen that the 
PI2 asymptotic description (\ref{KP12}) is slightly worse than at the first breaking, 
it is again somewhat shifted towards the left.
\begin{figure}
    \includegraphics[width=0.49\textwidth]{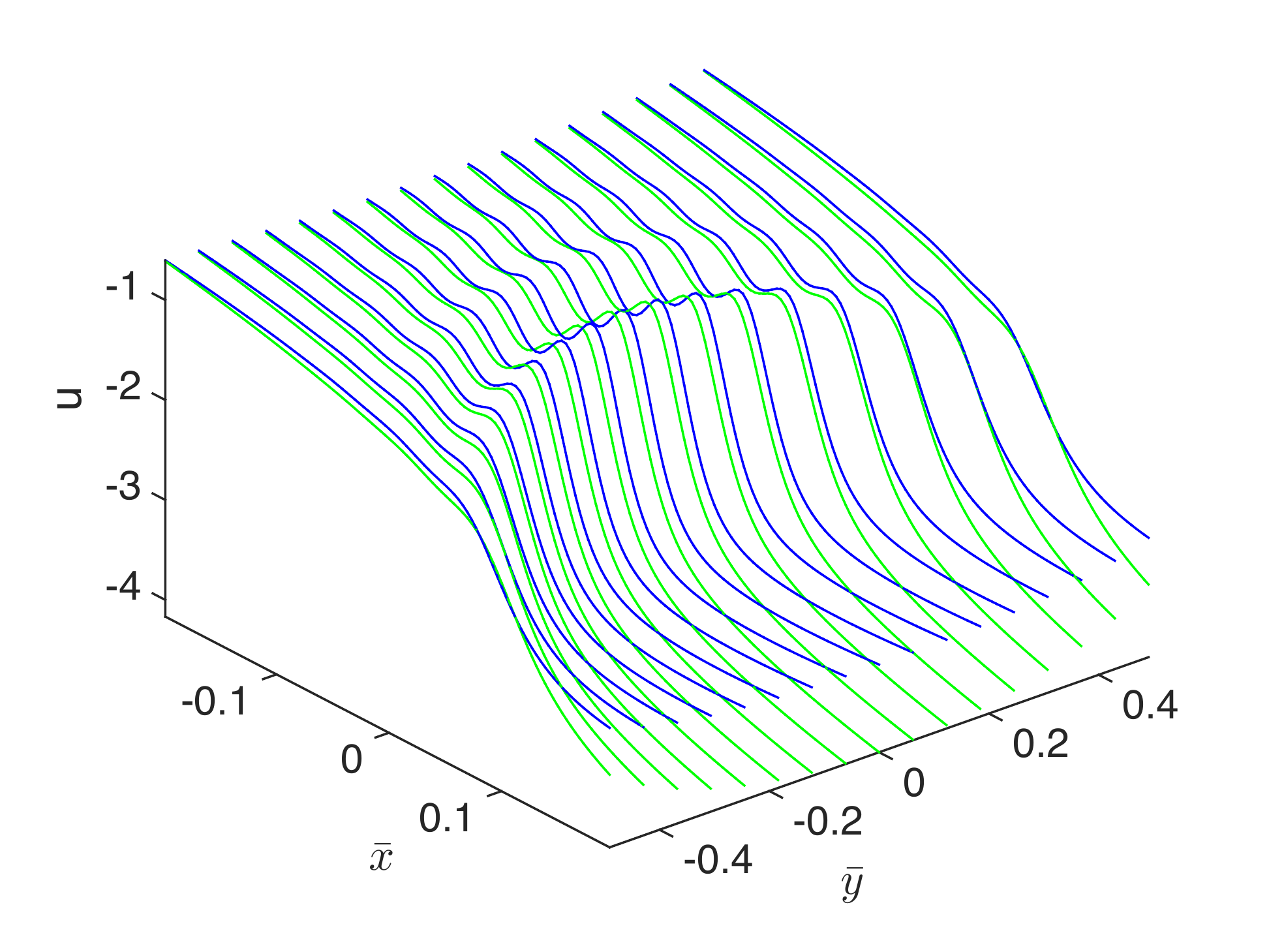}
 \includegraphics[width=0.49\textwidth]{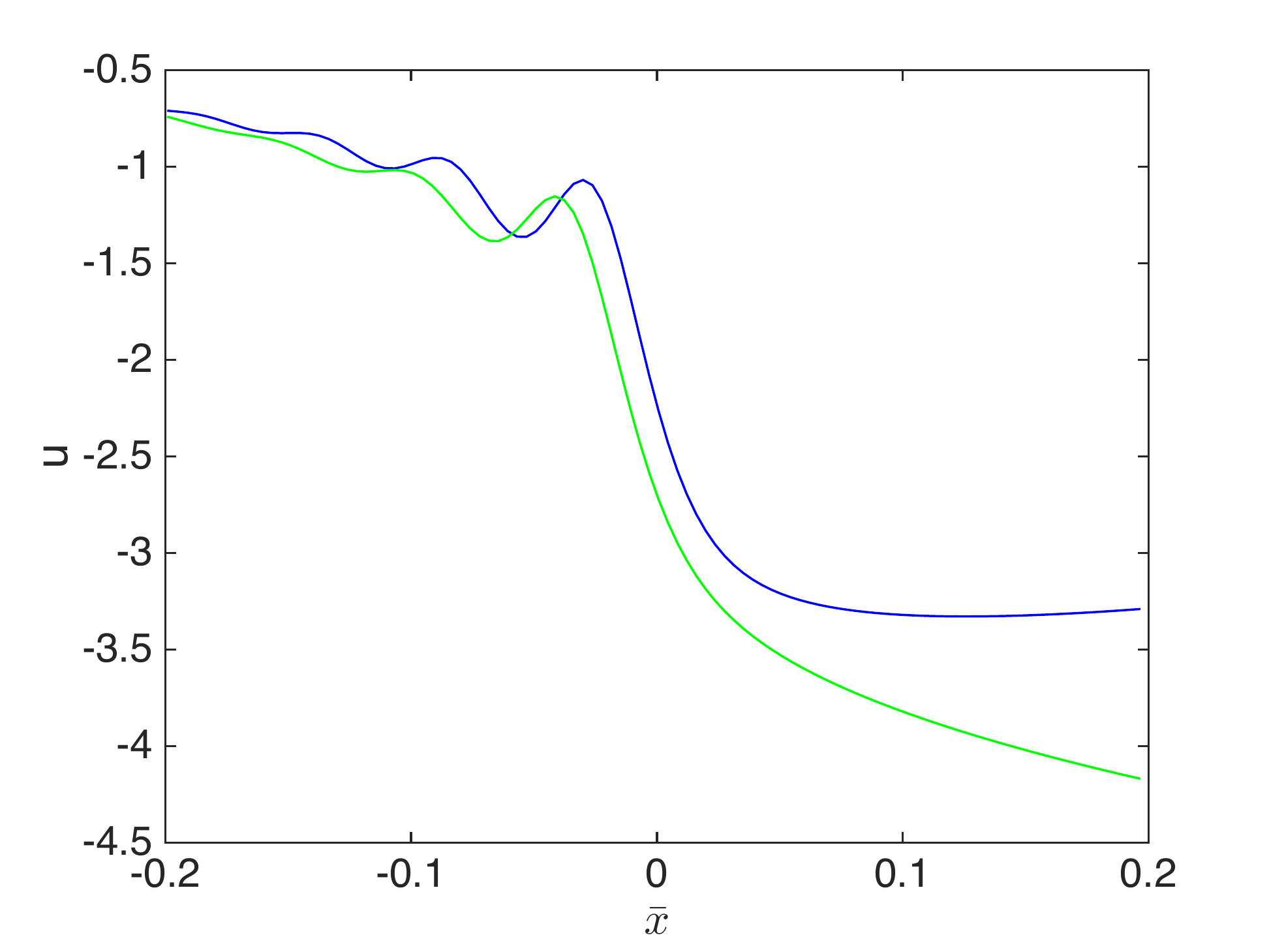}   
    \includegraphics[width=0.49\textwidth]{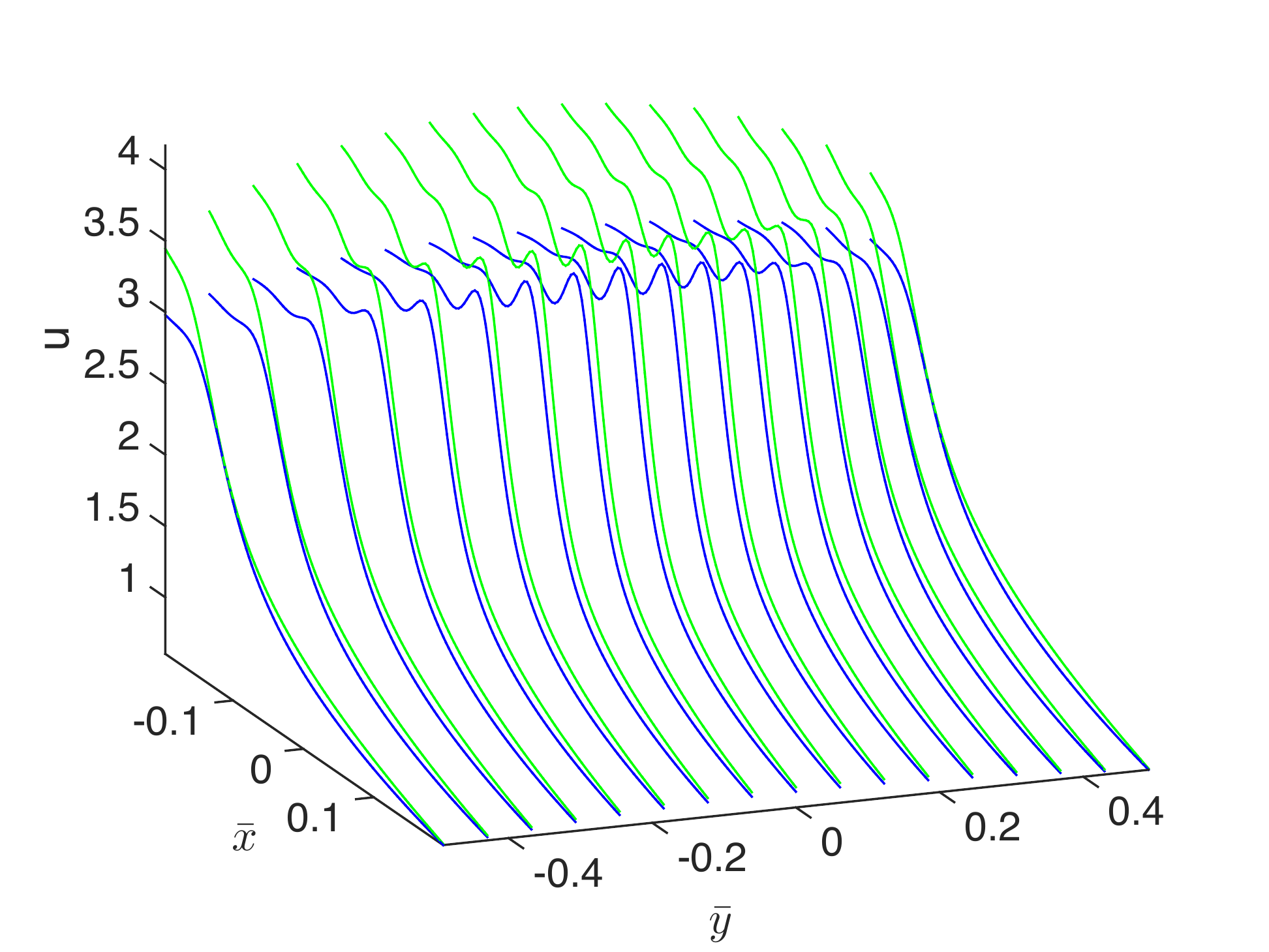}
 \includegraphics[width=0.49\textwidth]{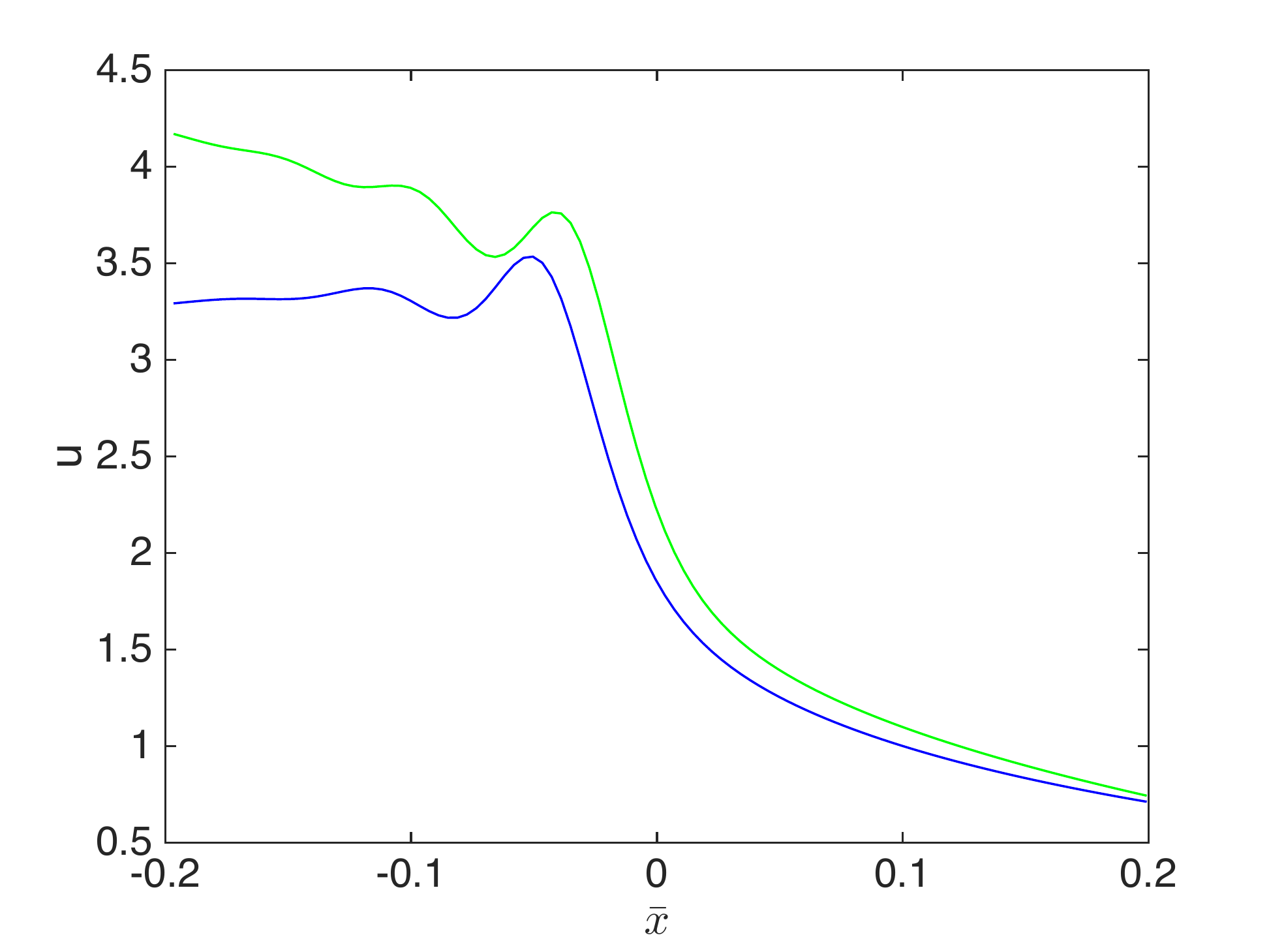}   
 \caption{Solutions to the KP  equation with $\epsilon=0.01$ for the symmetric initial data 
 (\ref{u0sym}) at the critical time $t\approx 0.3001$ of the second 
 break-up  in blue and the 
 corresponding asymptotic solution (\ref{KP12}) in green; on the left 
 in the vicinity of the critical point, on the right on the $y$-axis, 
 in the upper row for KP I, in the lower row for KP II.   }
 \label{KPc2}
\end{figure}

The same situation for KP II can be see in the lower row of Fig.~\ref{KPc2}. Here 
the asymptotic solution is as for the first breaking in KP II shifted 
towards the right. 

\subsection{Nonsymmetric initial data}
Both the numerical approach to dKP in \cite{GKE} and the asymptotic 
formula (\ref{KP12}) are  applicable to non symmetric initial data as for example in  (\ref{u0asym}). For such data, we could not 
reach the second breaking, but the first break-up is well resolved. 
The KP I solution for these data can be seen in the upper row of 
Fig.~\ref{KPasymc} at the critical time 
$t_{c}\approx 0.0855$ in the vicinity of the critical point 
$x_{c}\approx0.1045$ and $y_{c}\approx-0.2566$ on the left. Visibly the 
PI2 asymptotic description (\ref{KP12}) captures well the onset of the oscillations, also 
in dependence of $y-y_{c}$. The corresponding situation for KP II can 
be seen in upper row of Fig.~\ref{KPasymc} on the right. 
%The break-up there occurs  at the same critical time as for KP I at $-x_{c}$, $-y_{c}$. 
The  PI2 asymptotic solution (\ref{KP12}) matches also well with the KP II numerical  solution. 
\begin{figure}
    \includegraphics[width=0.49\textwidth]{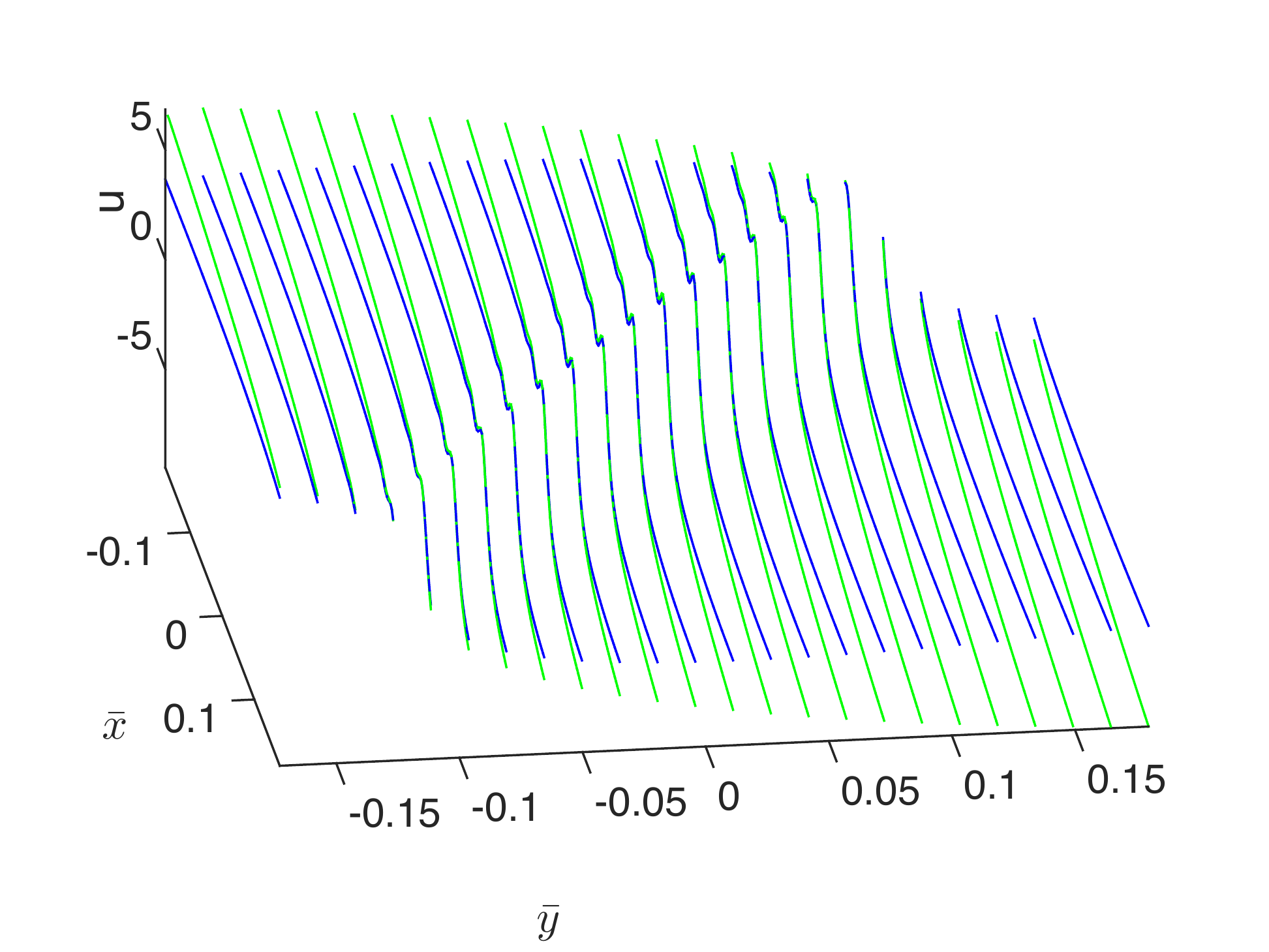}
 \includegraphics[width=0.49\textwidth]{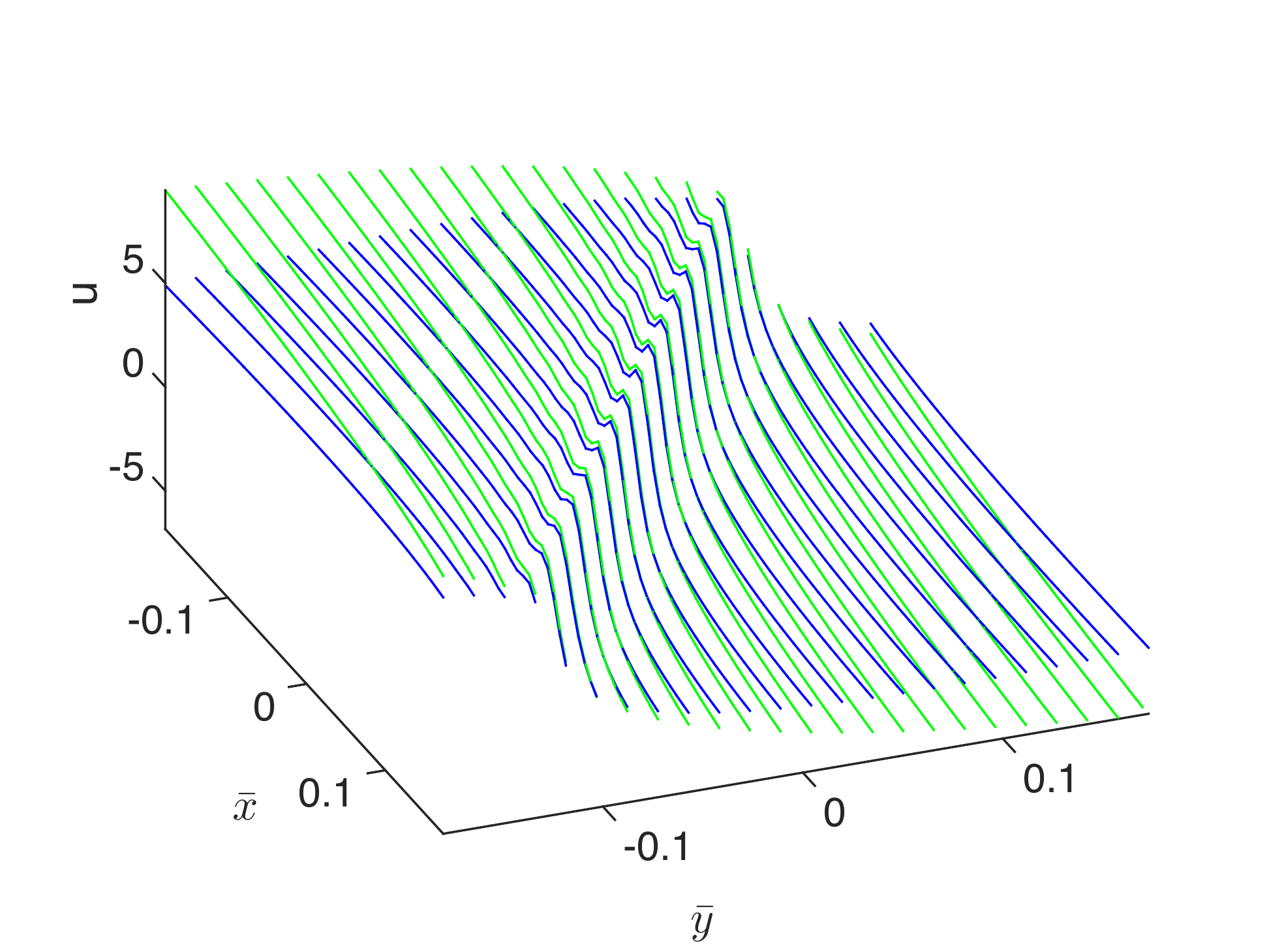}   
    \includegraphics[width=0.49\textwidth]{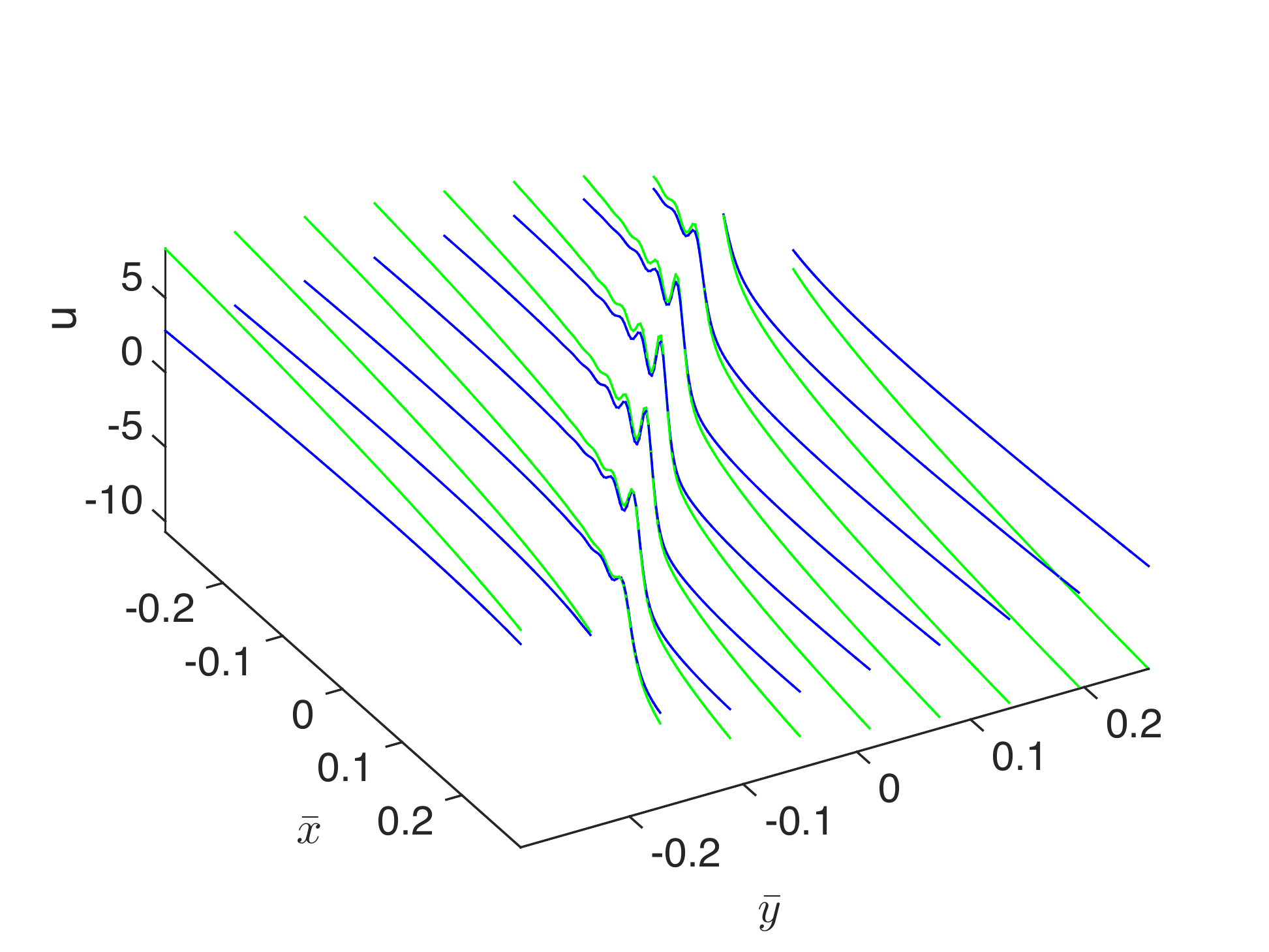}
 \includegraphics[width=0.49\textwidth]{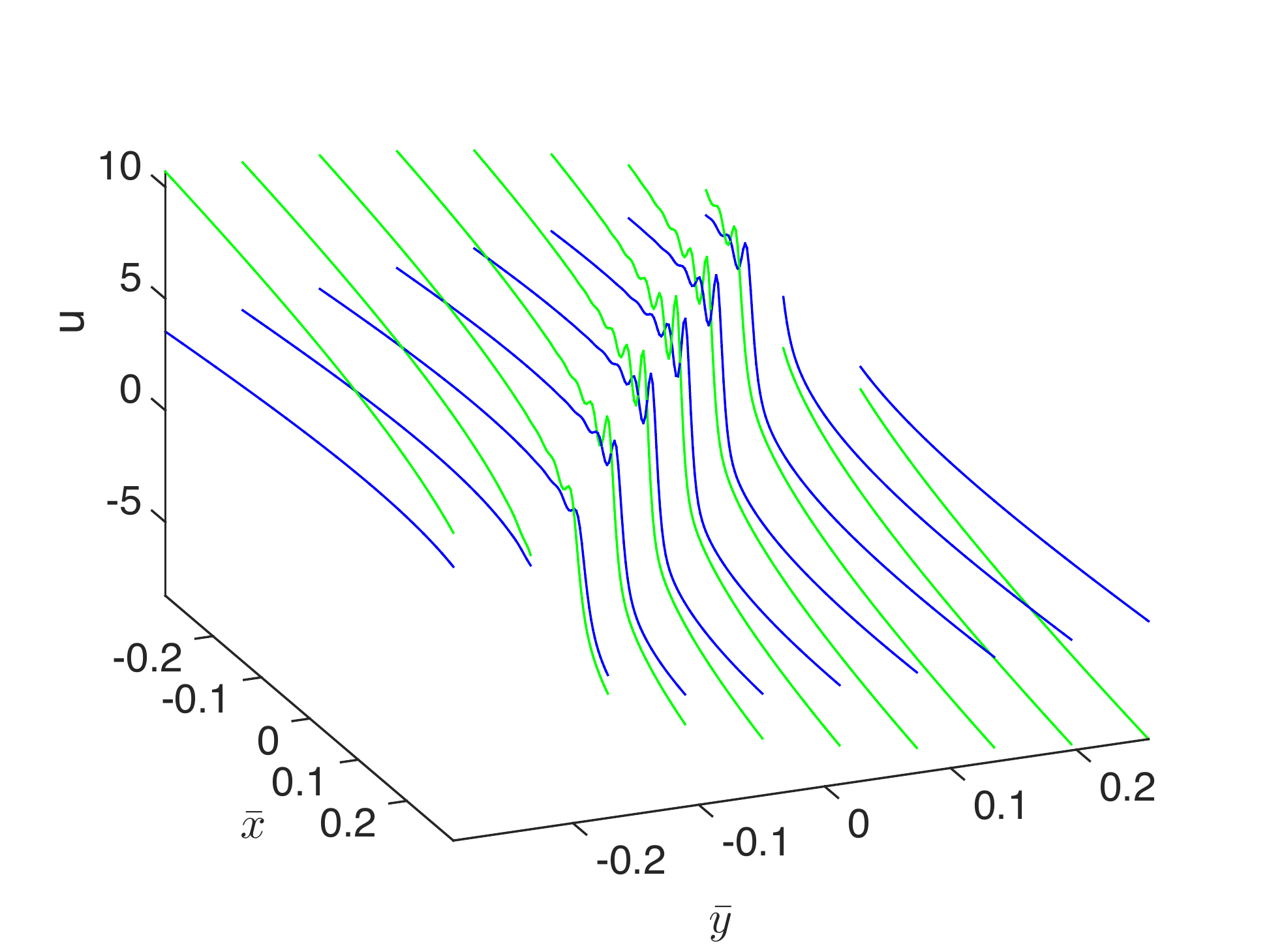}   
 \caption{Solutions to the KP equations with $\epsilon=0.01$ for the asymmetric  initial data 
 (\ref{u0asym}) in the vicinity of the critical point
 in blue and the 
 corresponding PI2 asymptotic solution (\ref{KP12}) in green; on the left 
KP I, on the right KP II, in the upper row at the critical time 
$t\approx 0.0855$, in the lower row at the time $t=0.09>t_{c}$.   }
 \label{KPasymc}
\end{figure}

It can be seen in the upper row of Fig.~\ref{KPasymc} that the situation is not the 
same for KP I and KP II, the asymptotic solution (\ref{KP12}) 
approximates the critical behavior better for KP I than for KP II. 
This persists for larger values of $t$ as can be seen in the lower 
row of 
Fig.~\ref{KPasymc}. Still it is remarkable that the asymptotic 
solution also catches at least qualitatively the behavior of the 
KP solution in larger distance from the critical points, especially 
the $y$-dependence.

\section{Solutions to the generalized Kadomtsev--Petviashvili equations near the 
critical point}
\label{sec:gKP}
Generalized KP equations allow to study the competing influence of 
dispersion and nonlinearity in a $2+1$ dimensional model. It is known 
that solutions to  generalized KP I equations with $n\geq 4/3$ can have blow-up. 
There is no theoretical description of this blow-up, it is just known 
that the $L^{2}$ norm of $u_{y}$ can explode in finite time, see 
\cite{Saut93} and \cite{Liu}. There are no theoretical predictions for generalized KP II so far. 
It is just expected that the defocusing effect in  generalized KP II should 
make blow-up less likely than in  generalized KP I situations 
(there is an analogy between generalized KP and nonlinear 
Schr\"odinger (NLS) equations in this respect). 
Numerical studies in \cite{KS12} and in more detail in 
\cite{KP14} have indicated that a self-similar blow-up can be 
expected for generalized KP I for $n\geq 4/3$ (see also the following section). 
No blow-up was found for generalized KP II for $n<3$. Therefore we consider here 
only the case $n=3$ where also for generalized KP II blow-up was observed. In 
addition an odd exponent $n$ has the advantage that the breaking is 
similar to the case $n=1$ studied in the previous section. We always 
use $N_{x}=N_{y}=2^{12}$ Fourier modes and $N_{t}=2000$ time steps. 

We consider the same initial data as in the KP case. It is 
expected that the balance between dispersion and nonlinearity is here 
tilted towards the nonlinearity. This means that dispersive shocks 
will be less pronounced (less oscillations confined to smaller 
domains), and that the break-up will happen at earlier times. In fact 
we find for the symmetric initial data (\ref{u0sym})that a first break-up  in the generalized dKP solution occurs 
at $t_{c}\approx0.0059$ (we did not study a second breaking) 
and $x_{c}\approx1.33$ and $y_{c}=0$. The 
corresponding solution to the generalized KP I equation can be seen for the 
critical time in the vicinity of the critical point in 
Fig.~\ref{gKPIc} in the upper row. Due to the reduced dispersive effects, 
the PI2  asymptotic description (\ref{KP12}) is almost more accurate than for KP 
I, but in a smaller domain. 
\begin{figure}
    \includegraphics[width=0.49\textwidth]{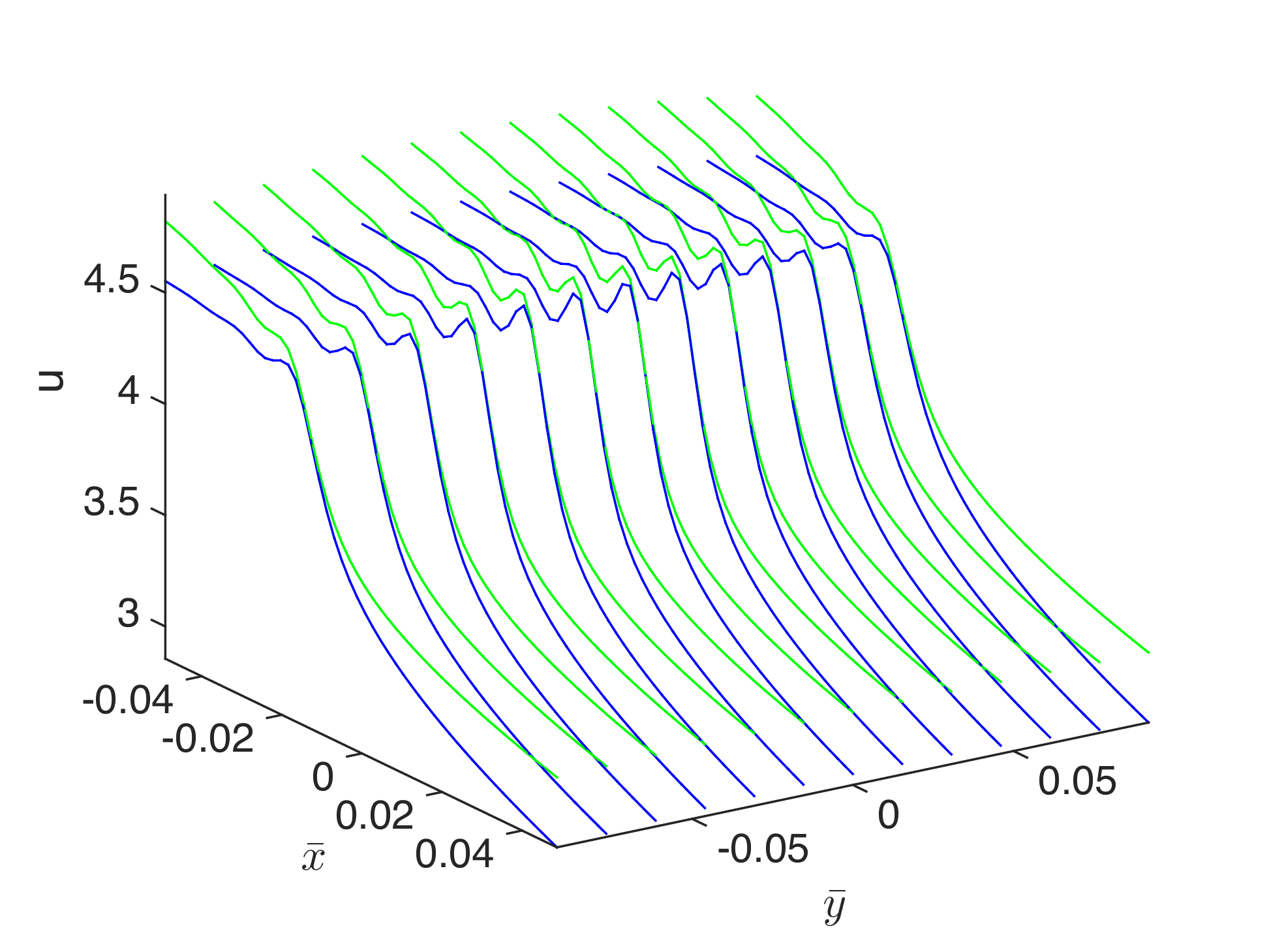}
 \includegraphics[width=0.49\textwidth]{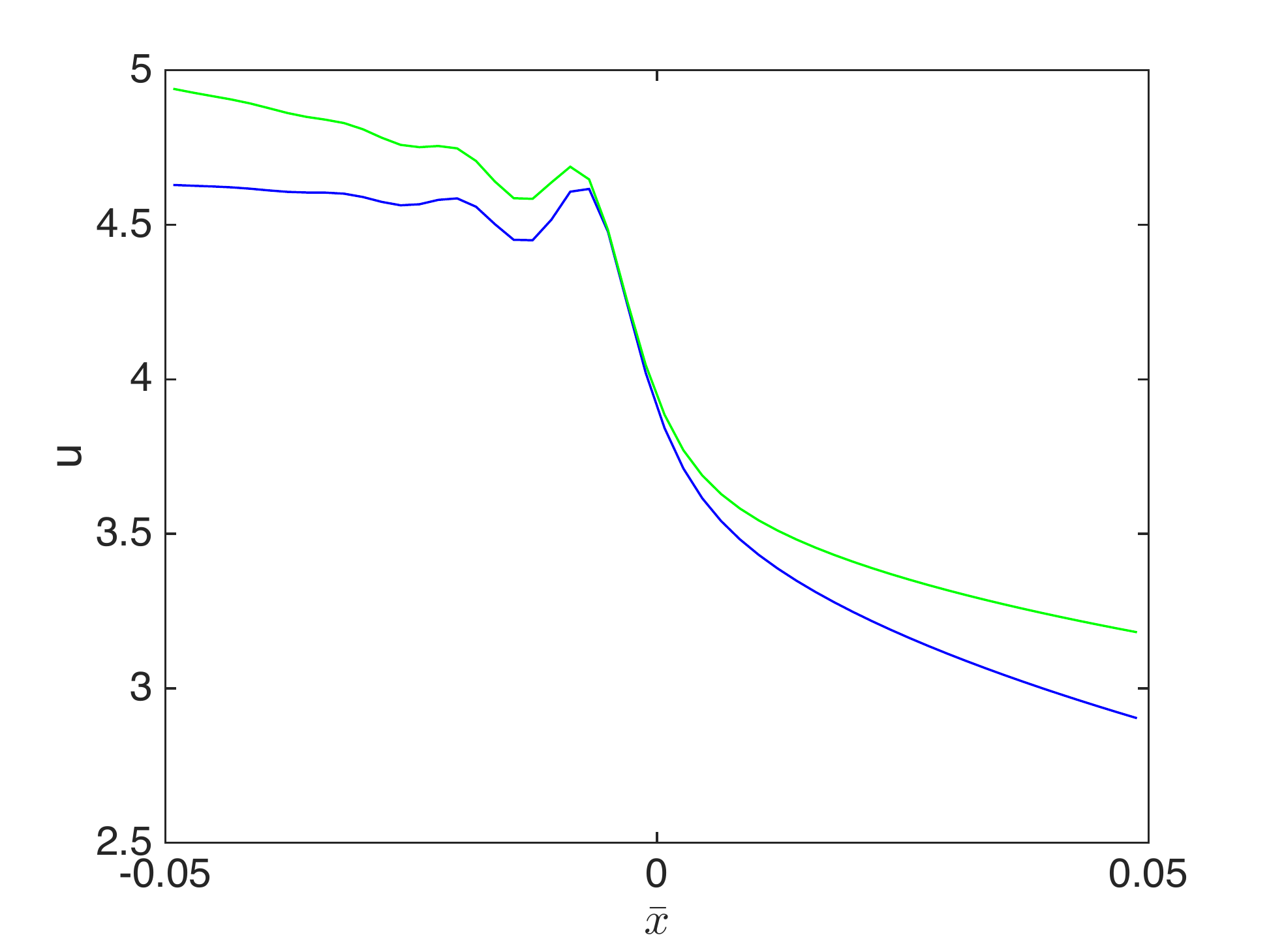}   
    \includegraphics[width=0.49\textwidth]{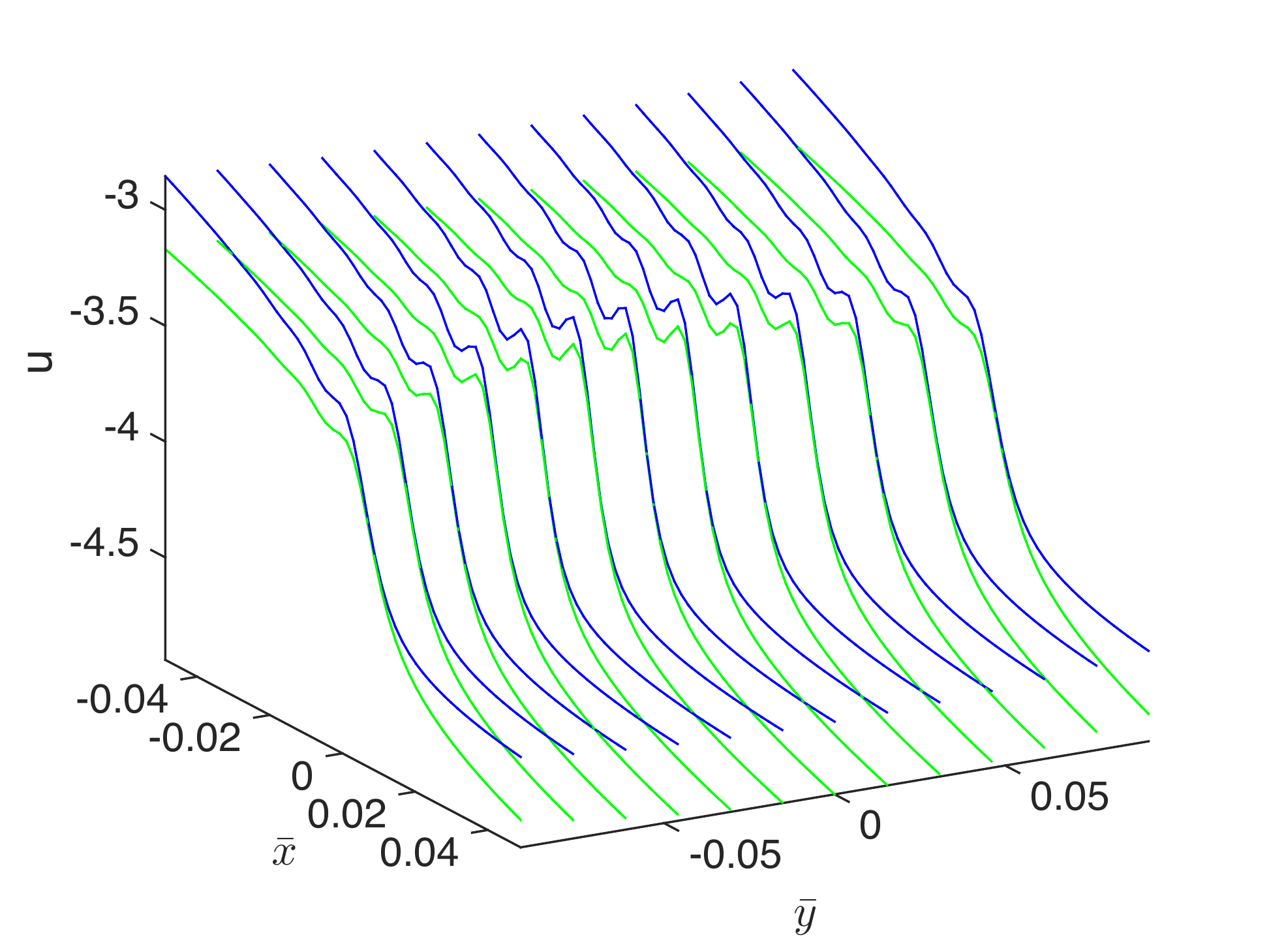}
 \includegraphics[width=0.49\textwidth]{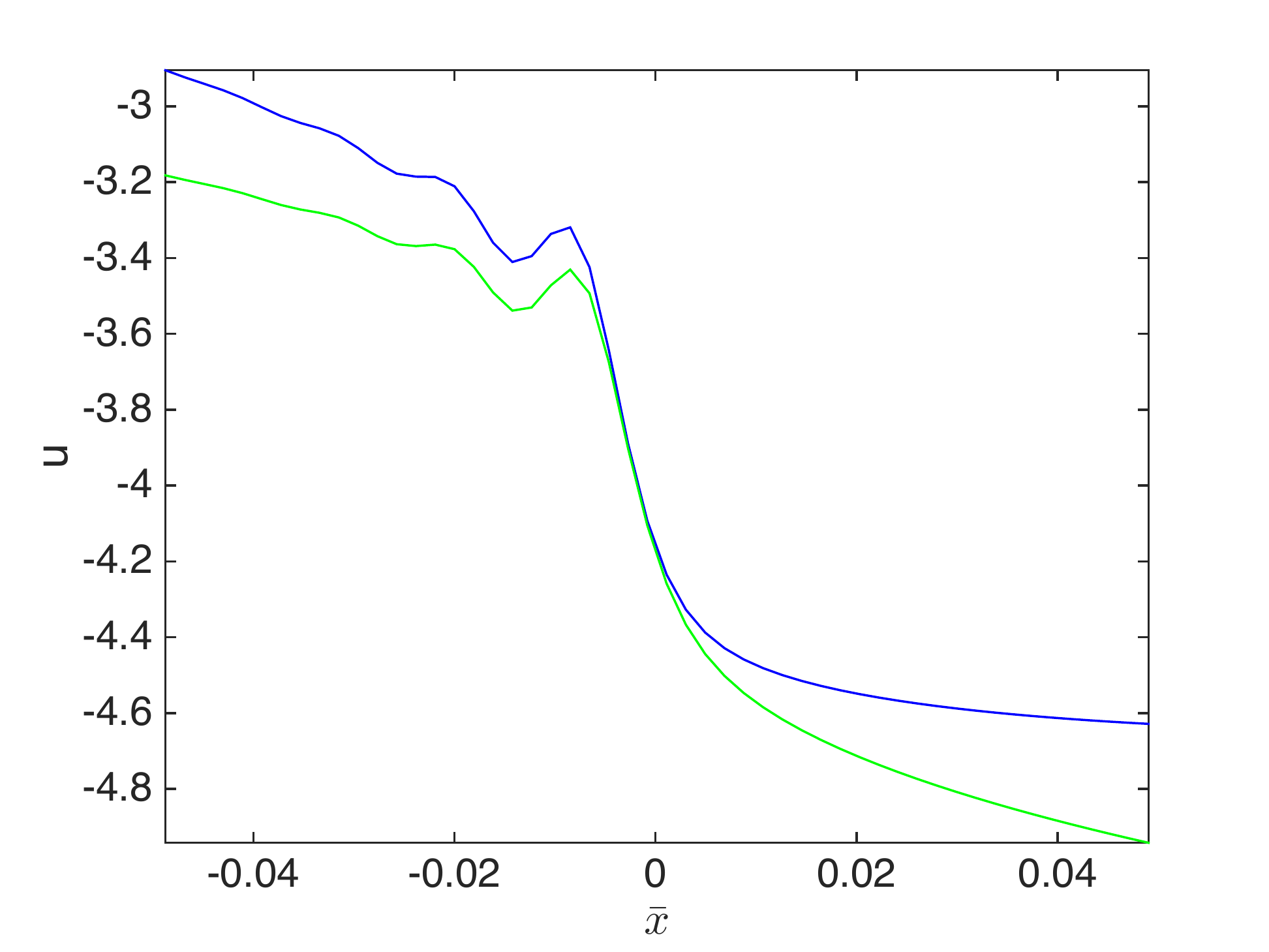}   
 \caption{Solutions to the generalized KP equations with $n=3$ and $\epsilon=0.01$ for the symmetric  initial data 
 (\ref{u0sym}) at the critical time $t_{c}\approx 0.0059$ near the 
 critical points $x_{c}$ and $y_{c}=0$ in blue and the 
 corresponding PI2 asymptotic solution (\ref{KP12}) in green; on the left 
 in dependence on $x$ and $y$, on the right on the $y$-axis, in the 
 upper row for KP I, in the lower row for KP II.   }
 \label{gKPIc}
\end{figure}

In Fig.~\ref{gKPIc} in the lower row 
we show the corresponding solution for the generalized KP 
II equation. The PI2 asymptotic description (\ref{KP12}) again matches 
the generalized KP II solution well in the vicinity of the critical point and 
catches qualitatively the first oscillation. The domain of 
applicability of the approximation is more confined than for KP II.

The solutions to the generalized KP I equation with $\epsilon=0.01$ and $n=3$ break 
for the asymmetric initial data (\ref{u0asym}) at $t_{c}\approx0.0028$ 
at $x_{c}\approx 0.2074$ and $y_{c}\approx 
   -0.0085$. The resulting solution at the critical time 
   can be seen in the vicinity of the critical point in 
   Fig.~\ref{gKPasymc} on the left. The PI2 asymptotic solution 
   (\ref{KP12}) describes the breaking well. The corresponding generalized dKP 
   II solution breaks at the same $t_{c}$ at $-x_{c}$, $y_{c}$. It 
   can be seen in Fig.~\ref{gKPasymc} on the right. 
\begin{figure}
    \includegraphics[width=0.49\textwidth]{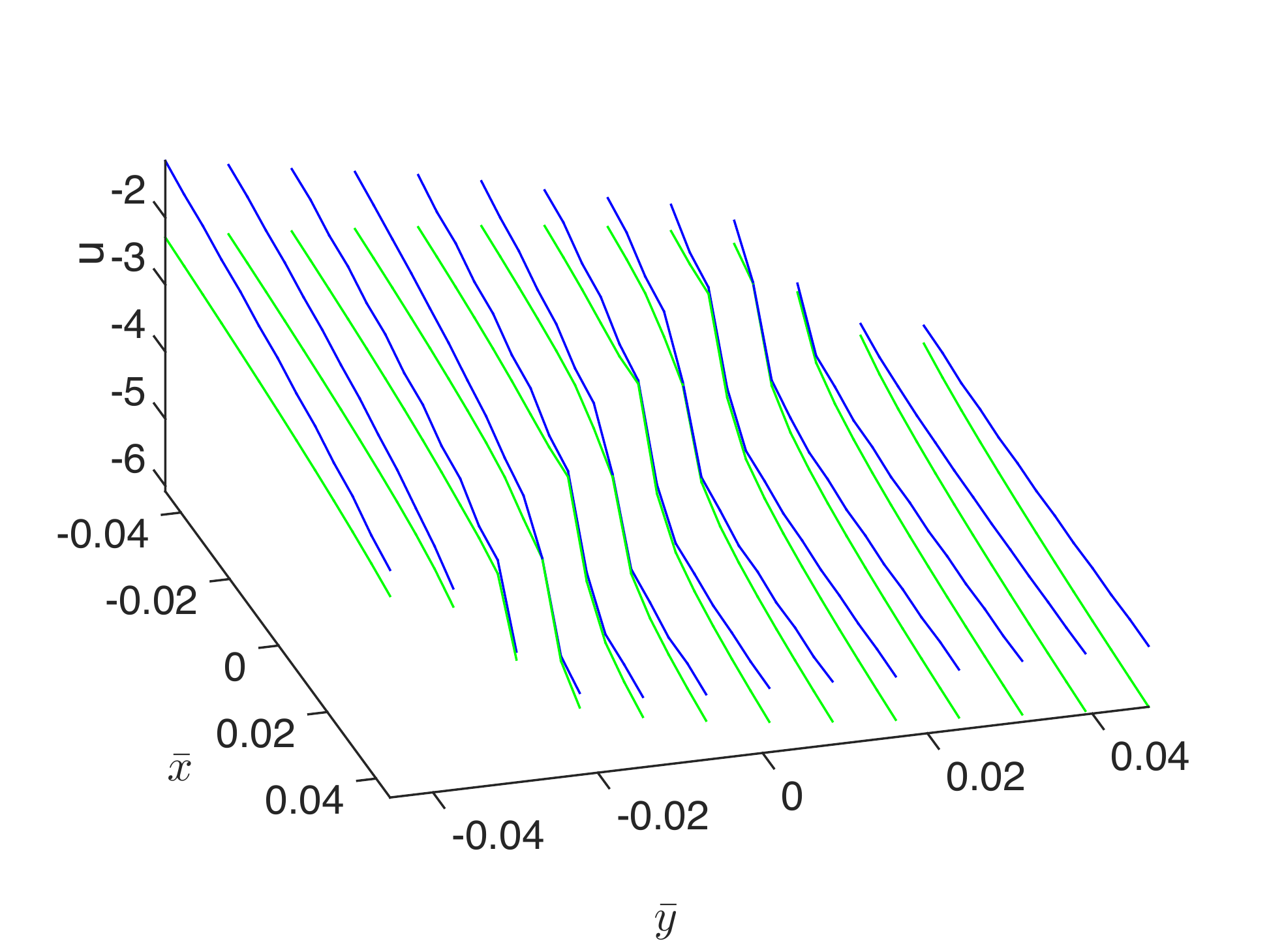}
 \includegraphics[width=0.49\textwidth]{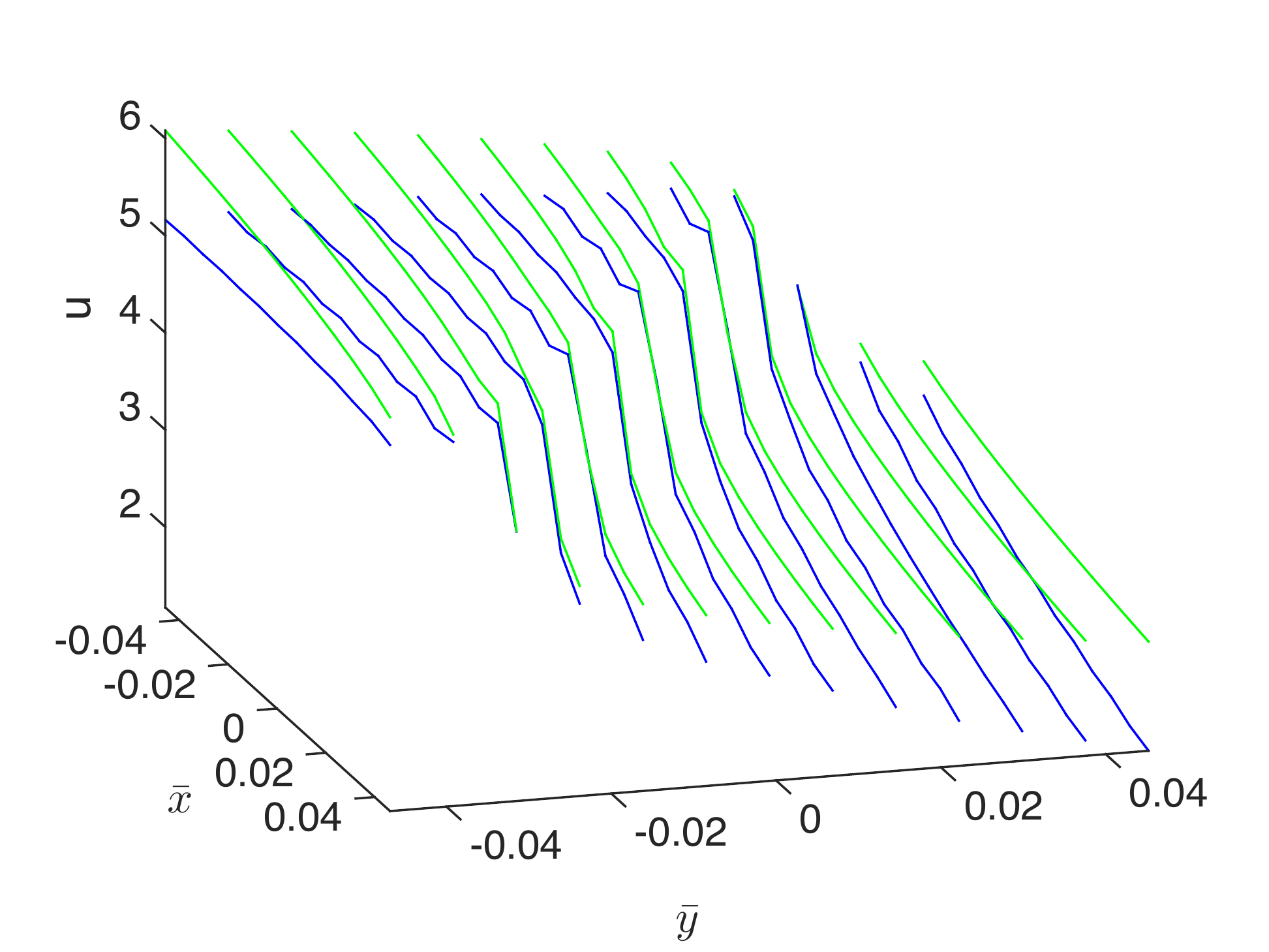}   
 \caption{Solutions to the generalized KP equations with $n=3$ and $\epsilon=0.01$ for the  asymmetric initial data 
 (\ref{u0asym}) at the critical time $t\approx 0.0028$ in the vicinity of the critical point
 in blue and the 
 corresponding PI2 asymptotic solution (\ref{KP12}) in green; on the left 
generalized KP I, on the right generalized KP II.   }
 \label{gKPasymc}
\end{figure}

\section{Outlook: Longtime behavior}
\label{outlook}
In this section we study the solutions of the previous sections for 
times $t\gg t_{c}$.  The solution of the  KP equations develops in 
the $(x,y)$ plane   a region of fast oscillations after the critical 
time $t_c$, the dispersive shock waves. 

The  generalized KP equation with $n\geq 4/3$ is, however, critical or
supercritical and its solution may develop   blow-up in finite time. In \cite{Saut93} and \cite{Liu}, it was shown 
that in generalized KP I solutions, a blow-up of the $L^{2}$ norm of $u_{y}$ can 
be observed for certain initial data. There are no rigorous analytic results 
on blow-up in the generalized KP II solutions so far. The first numerical studies of the generalized KP 
solutions appear to have been presented in \cite{WAS}, a more 
detailed study was performed in \cite{KP14}. The numerical results in 
\cite{KP14} led to the conjecture that there is an $L^{\infty}$ 
blow-up with a self similar blow-up profile for $n\geq 4/3$ for generalized KP I 
and for $n\geq 3$ for generalized KP II. In \cite{KP15}, blow-up in dispersive 
shocks for solutions to generalized KdV equations were studied 
numerically in the limit of small dispersion. These results are 
compared with the blow-up found below in solutions to generalized KP equations in 
the small dispersion limit. 

Since we did not have access to parallel computers for this work, we 
lack the necessary resolution to address these questions with the 
needed resolution (we use here $N_{x}=2^{13}$ and $N_{y}=2^{12}$ with 
$N_{t}=5000$ up to $N_{t}=20000$). The figures shown below have thus to be seen as 
indicative with the goal to outline directions of future research. 

\subsection{Solution to the KP equations for symmetric initial data}
In Fig.~\ref{KPct04a} we show  the dispersive shock waves that are formed in the KP I and KP II solutions   for the initial data (\ref{u0sym})  at $t=0.4>t_c$ i.e. well 
after both break-ups.
In Fig.~\ref{KPct04} we  zoom in  the oscillatory zones in the KP I  solution.
 The oscillations appear near the two critical 
points discussed in section \ref{sec:KP}.  It can be seen that the oscillations follow a 
parabolic pattern initially as suggested by the asymptotic formula 
(\ref{KP12}). But the focusing effect of the KP I equation appears to 
lead to the formation of a cusp. This is presumably a real effect and 
not related to a lack of resolution since a similar behavior was seen 
in \cite{KR13} where considerably higher resolution could be used. 
The analogy between KP and NLS suggests that this cusp  could be 
related to higher genus solutions appearing in the asymptotic 
description of the oscillations.  For times closer 
to the critical time, the PI2 asymptotic (\ref{KP12}) in the lower row 
of Fig.~\ref{KPct04} gives a 
qualitative approximation to the oscillations both for the parabolic 
shape and the number of oscillations. For larger times, the PI2 
asymptotic has considerably more oscillations. 

\begin{figure}
   \includegraphics[width=0.51\textwidth]{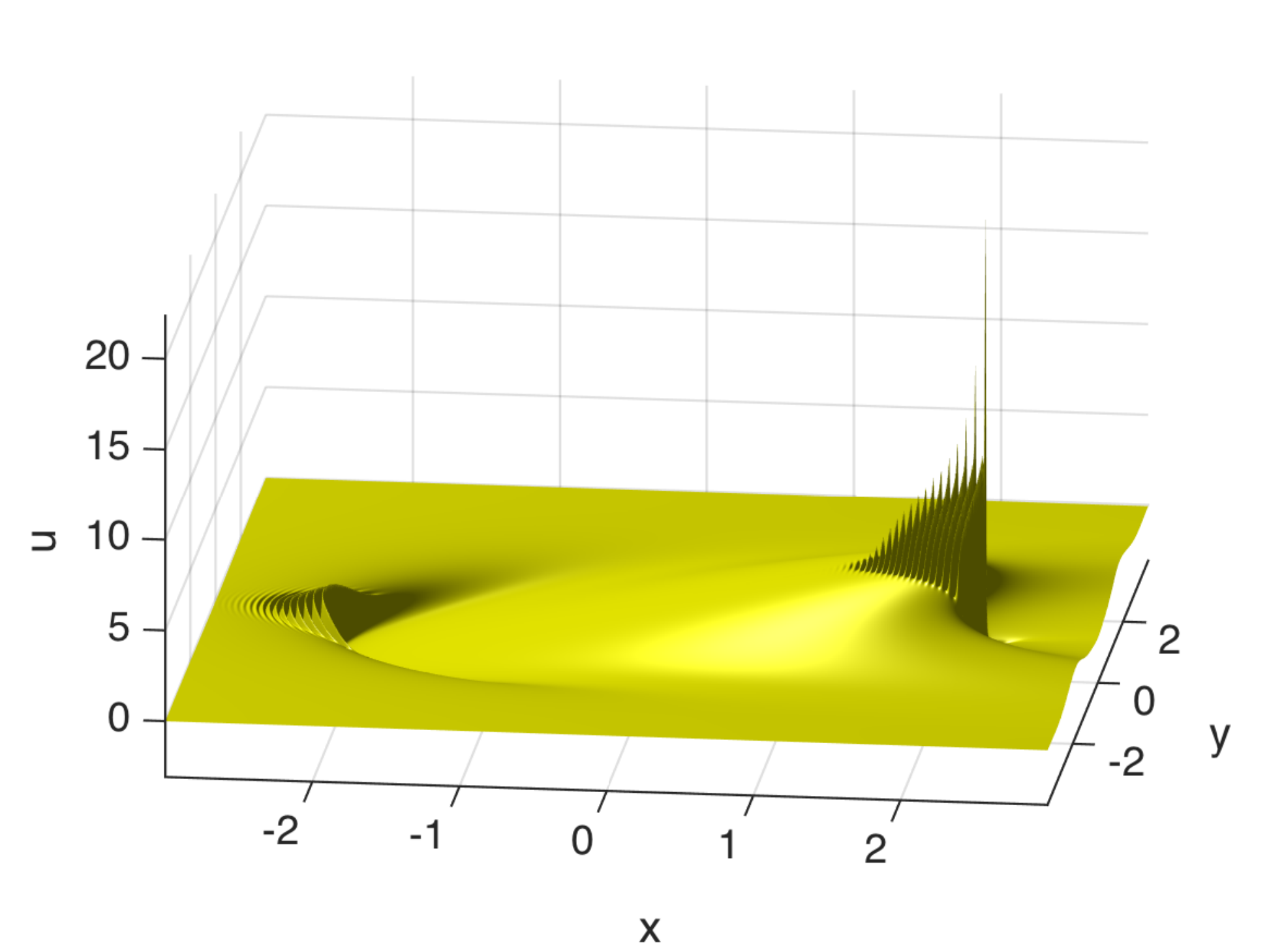}
 \includegraphics[width=0.51\textwidth]{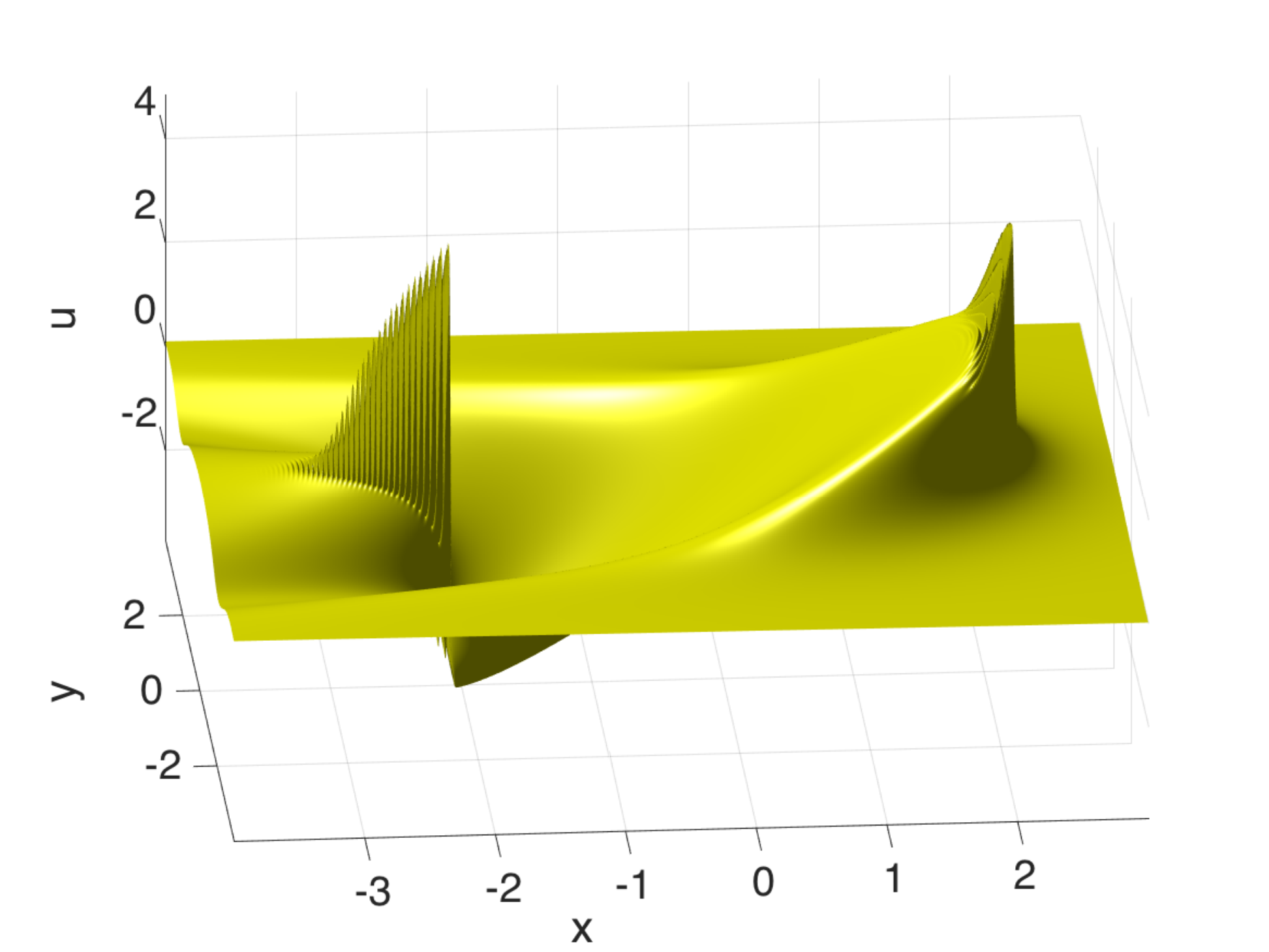}   
 \caption{Solution to the KP I  (left) and KPII (right)  equations with $\epsilon=0.01$ for the symmetric  initial data 
 (\ref{u0sym}) for $t=0.4>t_{c}\simeq 0.22$.  The KP I focusing effect can be observed from   the different scales in the $u$-axis of the two plots.  A zoom in   of these plots is shown in Fig.~\ref{KPct04} and Fig.\ref{KPIIct04} respectively.}
 \label{KPct04a}
\end{figure}

\begin{figure}
   \includegraphics[width=0.49\textwidth]{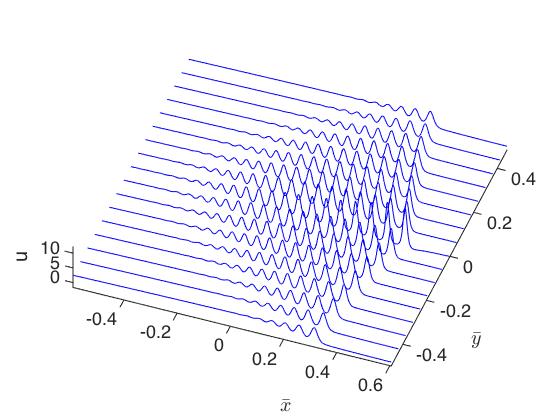}
 \includegraphics[width=0.49\textwidth]{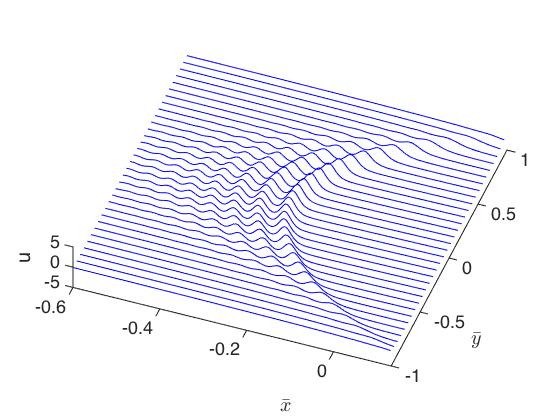}   
   \includegraphics[width=0.49\textwidth]{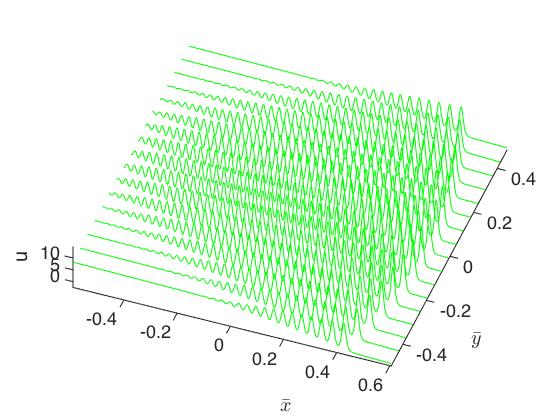}
  \includegraphics[width=0.49\textwidth]{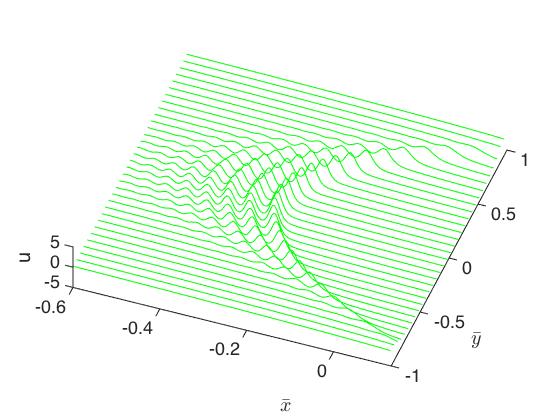}   
 \caption{Solution to the KP I equations with $\epsilon=0.01$ for the symmetric  initial data 
 (\ref{u0sym}) for $t=0.4>t_{c}\simeq 0.22$; in the upper row on the left the dispersive shock  wave
 appearing first, on the right the dispersive shock wave  appearing at a later time ($t_c\simeq 0.3$), in 
 the lower row the corresponding PI2 asymptotic solutions (\ref{KP12}).   }
 \label{KPct04}
\end{figure}

The corresponding KP II solution can be seen in Fig.~\ref{KPIIct04}. 
Due to the defocusing character of KP II, the oscillations  have 
a parabolic pattern at least in the range of times we considered. It is unclear 
whether some sort of cusp will appear at later times in the envelope of the oscillatory profile.
 In the lower row of Fig.~\ref{KPIIct04} we show the corresponding  PI2
asymptotic solution (\ref{KP12}) which shows a similar behavior as 
for KP I. 
\begin{figure}
    \includegraphics[width=0.49\textwidth]{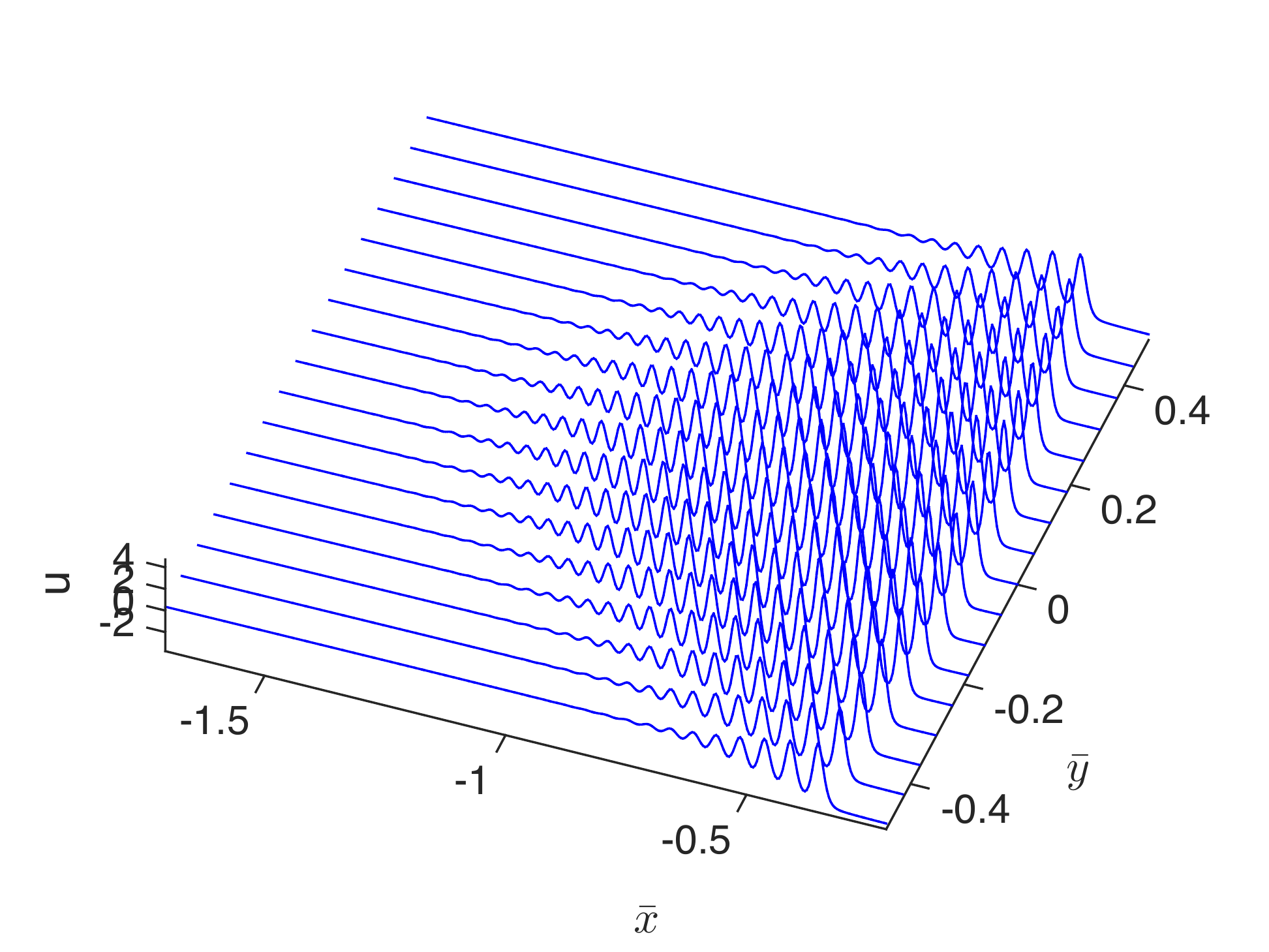}
 \includegraphics[width=0.49\textwidth]{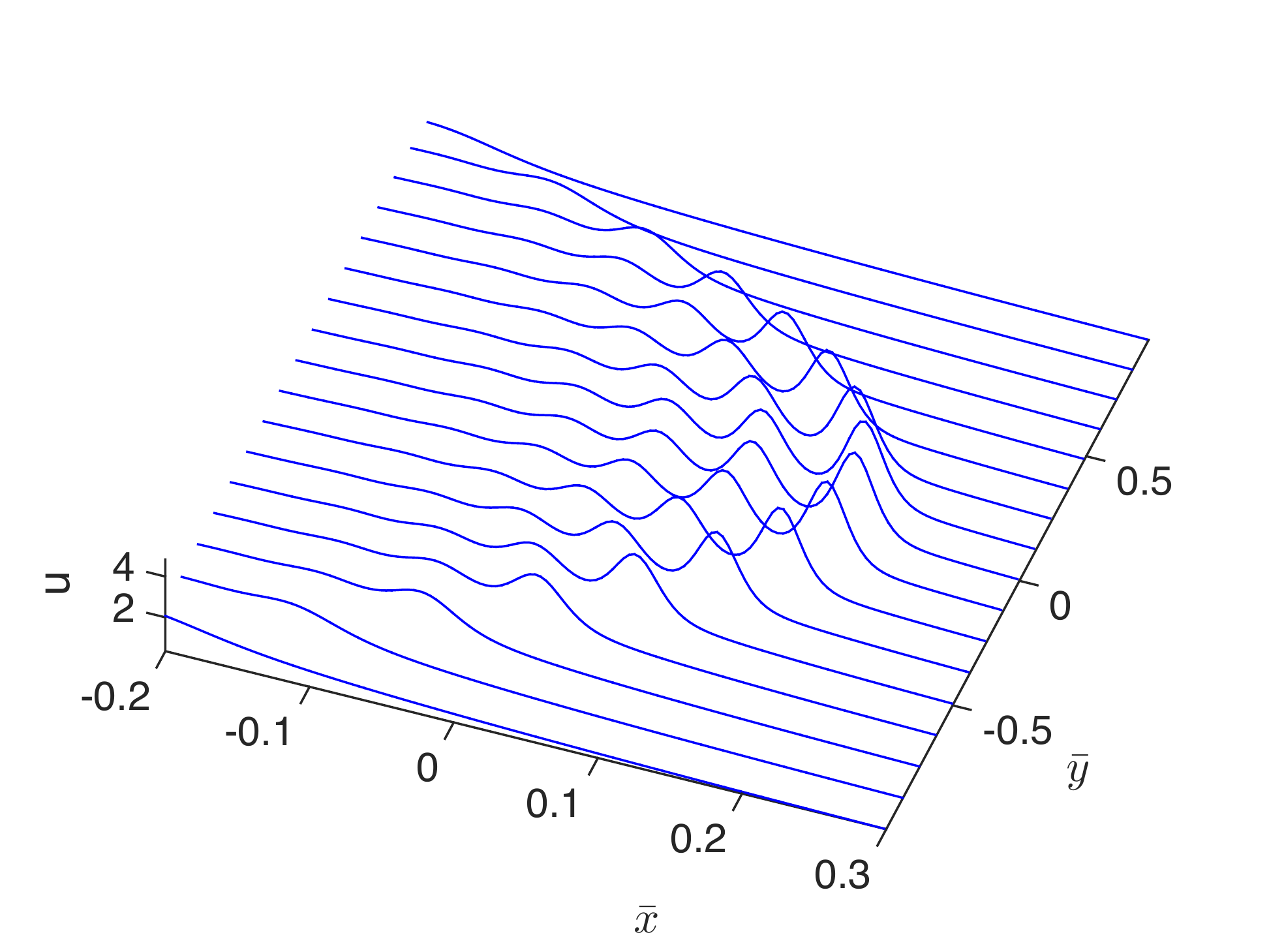}   
    \includegraphics[width=0.49\textwidth]{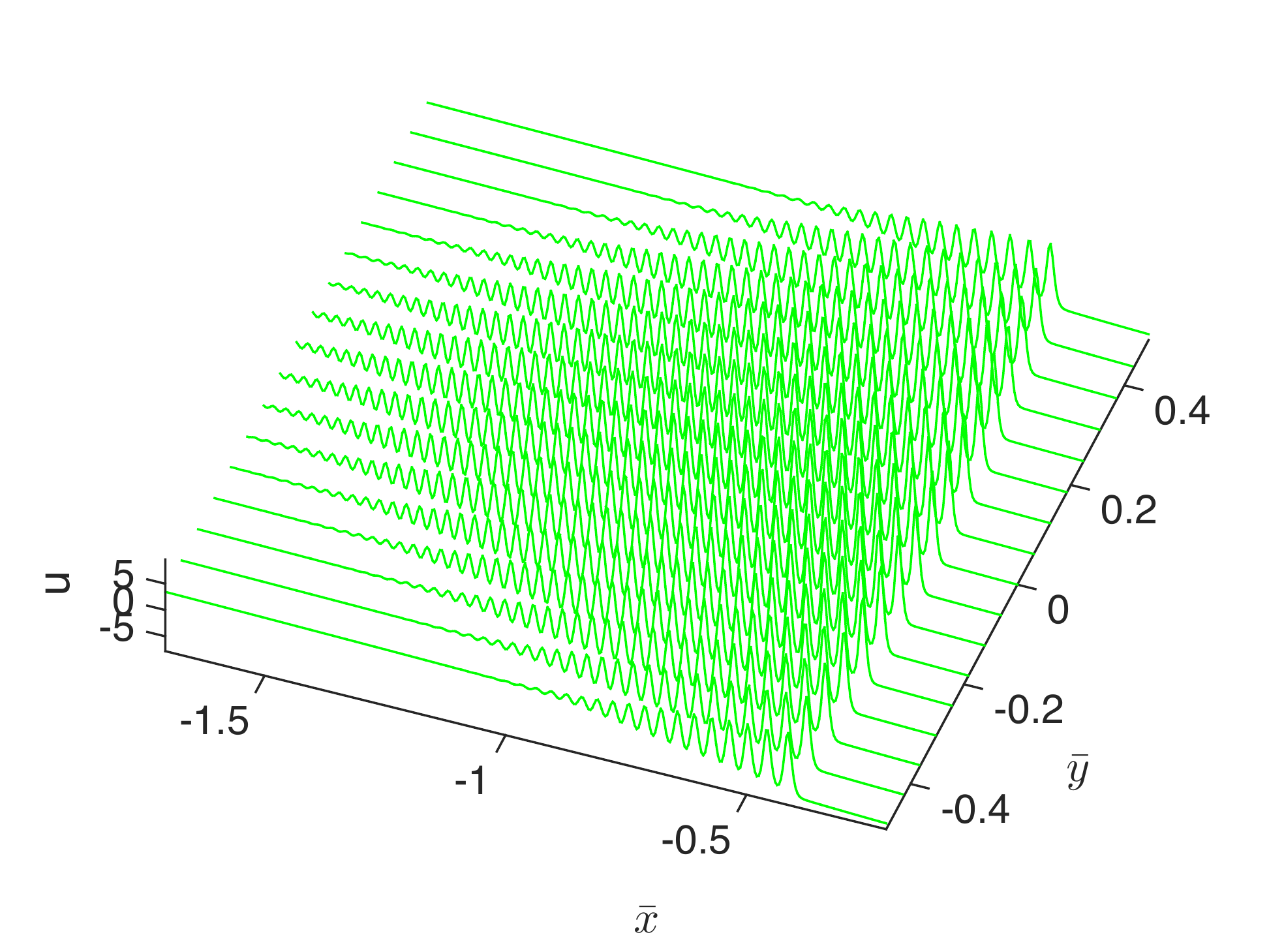}
 \includegraphics[width=0.49\textwidth]{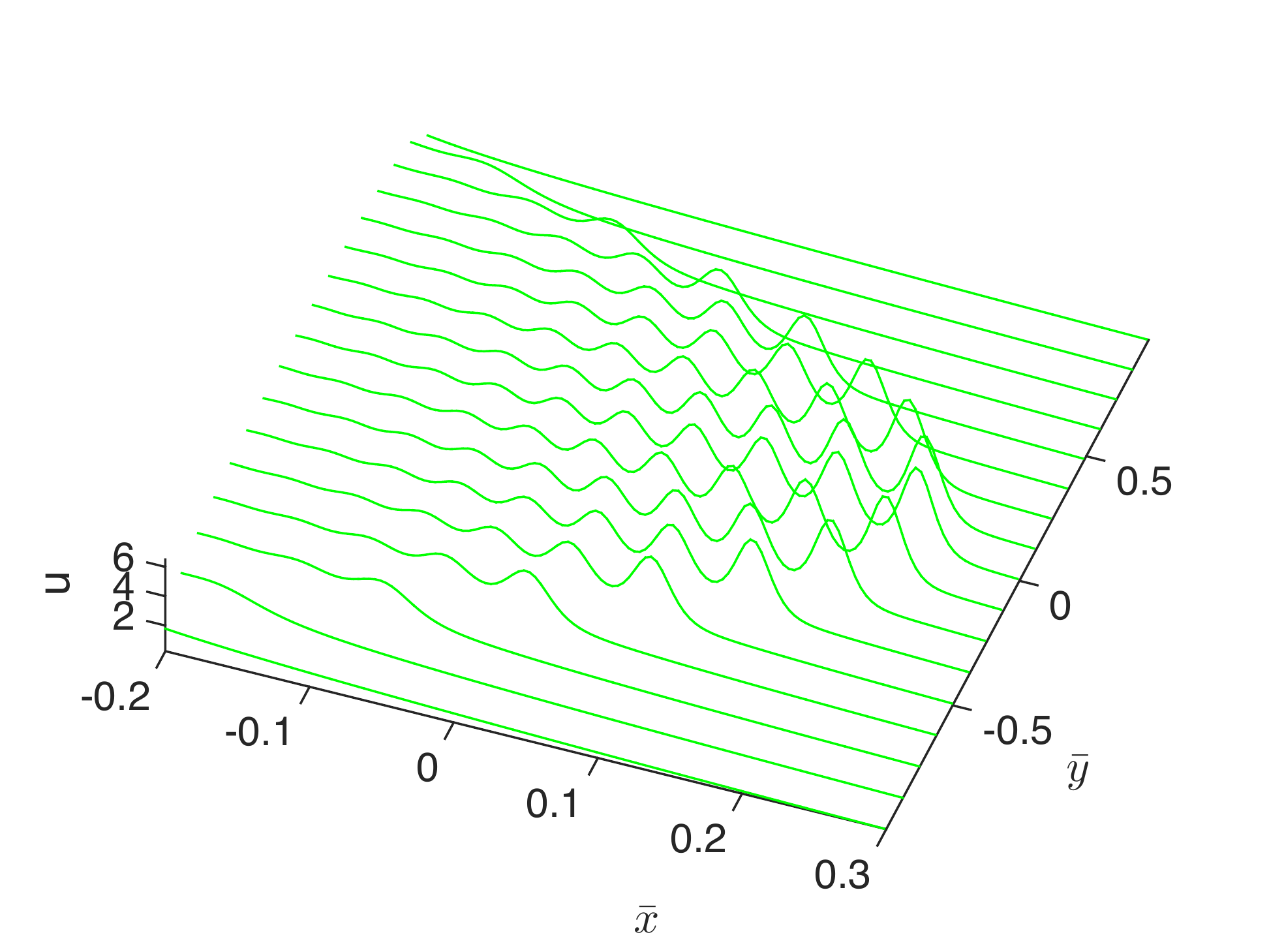}   
 \caption{Solution to the KP II  equation with $\epsilon=0.01$ for the symmetric  initial data 
 (\ref{u0sym}) for $t=0.4>t_{c}\simeq 0.22$;  in the upper row on the left the dispersive shock  wave 
 appearing first, on the right the dispersive  shock wave  appearing at a later time ($t_c\simeq 0.3$); 
 in the lower row the corresponding  PI2 asymptotic solutions (\ref{KP12}.  }
 \label{KPIIct04}
\end{figure}

The asymptotic description of these oscillation in the small 
dispersion limit will be the subject of further research. The task is  to derive and solve the 
Whitham equations which asymptotically describe the oscillations  and to find the needed phase 
information in the asymptotic solutions.

\subsection{Solution to generalized KP equations for symmetric initial data}
In \cite{KP14} blow-up in generalized KP solutions was studied numerically. It 
was conjectured that a self similar $L^{\infty}$ blow-up is observed.

We consider again the symmetric initial data (\ref{u0sym}) for generalized KP I 
with $\epsilon=0.01$ and let the code run with $N_{t}=20000$ and 
$N_{x}=2^{13}$, $N_{y}=2^{12}$. High resolution in $x$ is needed to 
resolve the dispersive shock as before, but an even higher 
resolution in $y$ would be necessary since the blow-up in $y$ 
according to \cite{KP14} is stronger in $y$ than in $x$ direction. 
Thus on the used computers, just 
indication for future work on larger computers can be obtained. 
For KP, the relative mass conservation drops below $10^{-3}$ at 
$t\approx 0.0085$, but the Fourier coefficients deteriorate even earlier. 
We show the generalized KP I solution for 
$t=0.00695$, i.e., at a time where the solution is still resolved,  on the 
left in Fig.~\ref{gKPbu}. It seems that the cusped structure in the 
oscillatory zone in Fig.~\ref{KPct04} blows up in the gKP I solution. 
\begin{figure}
    \includegraphics[width=0.49\textwidth]{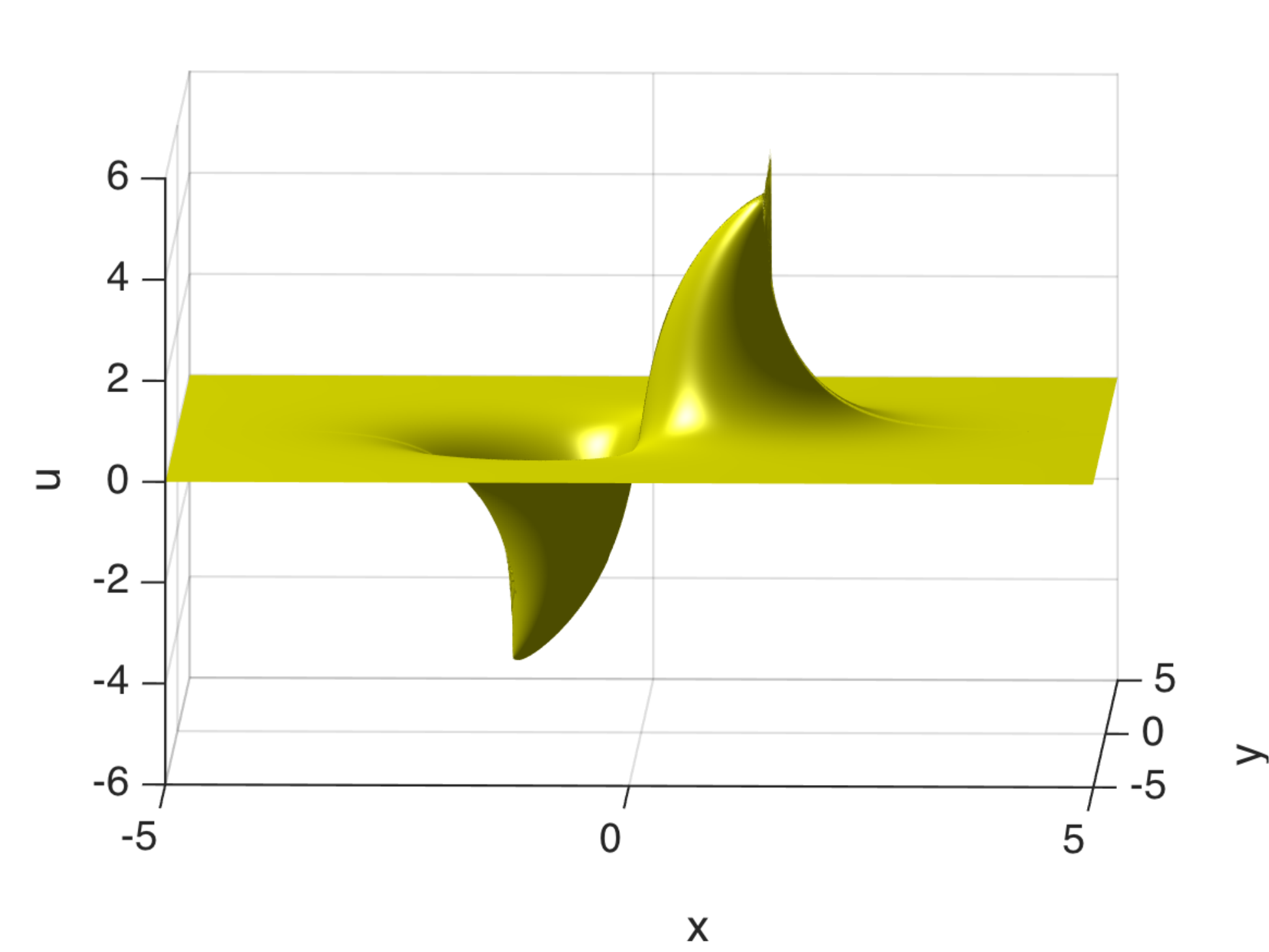}
 \includegraphics[width=0.49\textwidth]{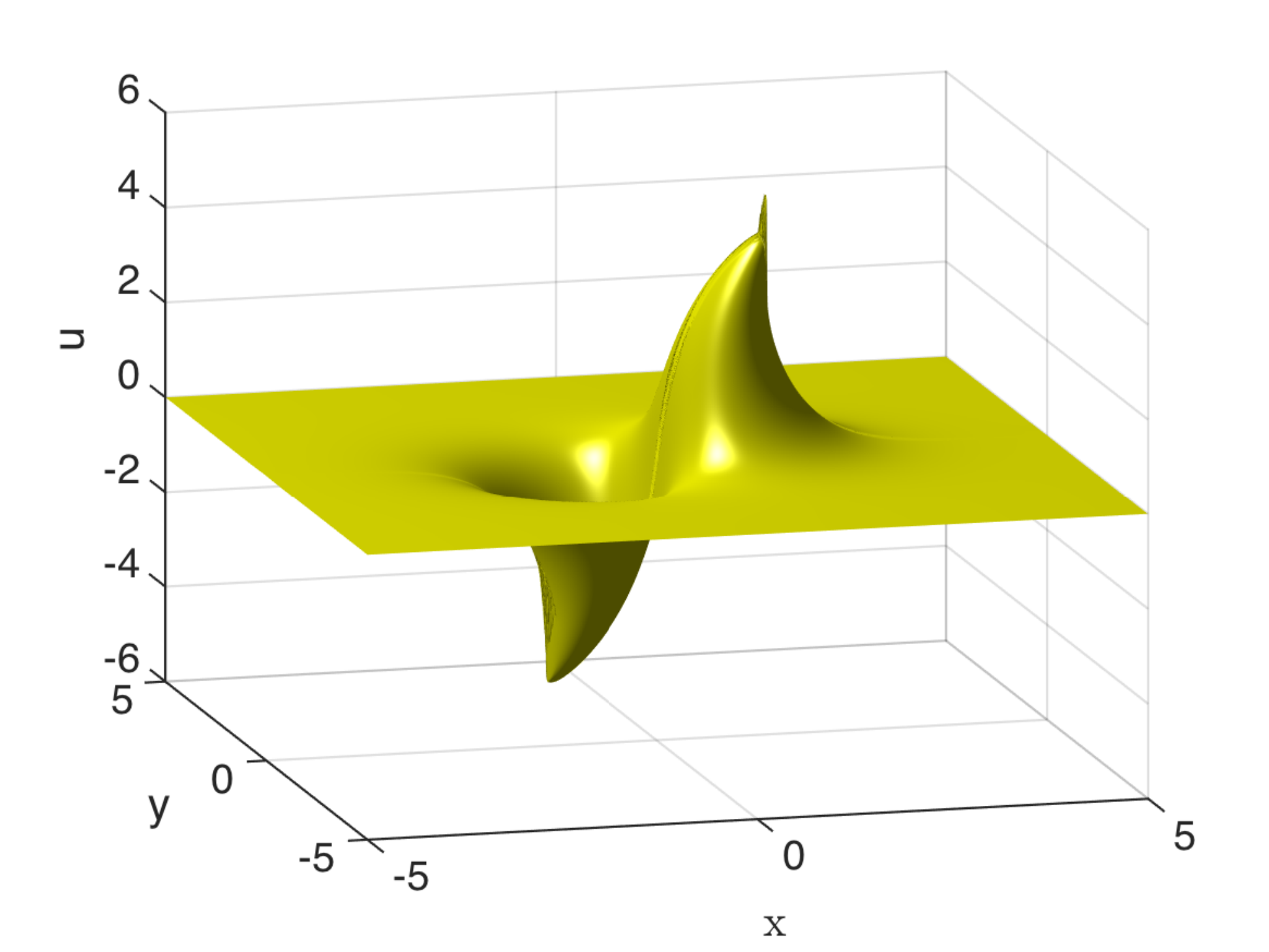}   
 \caption{Solution to generalized KP  equations with $\epsilon=0.01$ and $n=3$ 
 for the initial 
 data (\ref{u0sym}); on the left the generalized KP I solution for $t=0.00695$
 $L^{\infty}$ norm of $u$, on the right the generalized KP II solution for 
 $t=0.0071$.   }
 \label{gKPbu}
\end{figure}

In Fig.~\ref{gKPnorm} we show the $L^{\infty}$ norm of the solution 
which appears to explode. The 
fuzzy structure of the $L^{\infty}$ norm indicates a lack of 
resolution. In fact we did not manage to fit the norms to the 
asymptotic formulae of \cite{KP14} which means we do not get close 
enough to the actual blow-up. 
\begin{figure}
    \includegraphics[width=0.49\textwidth]{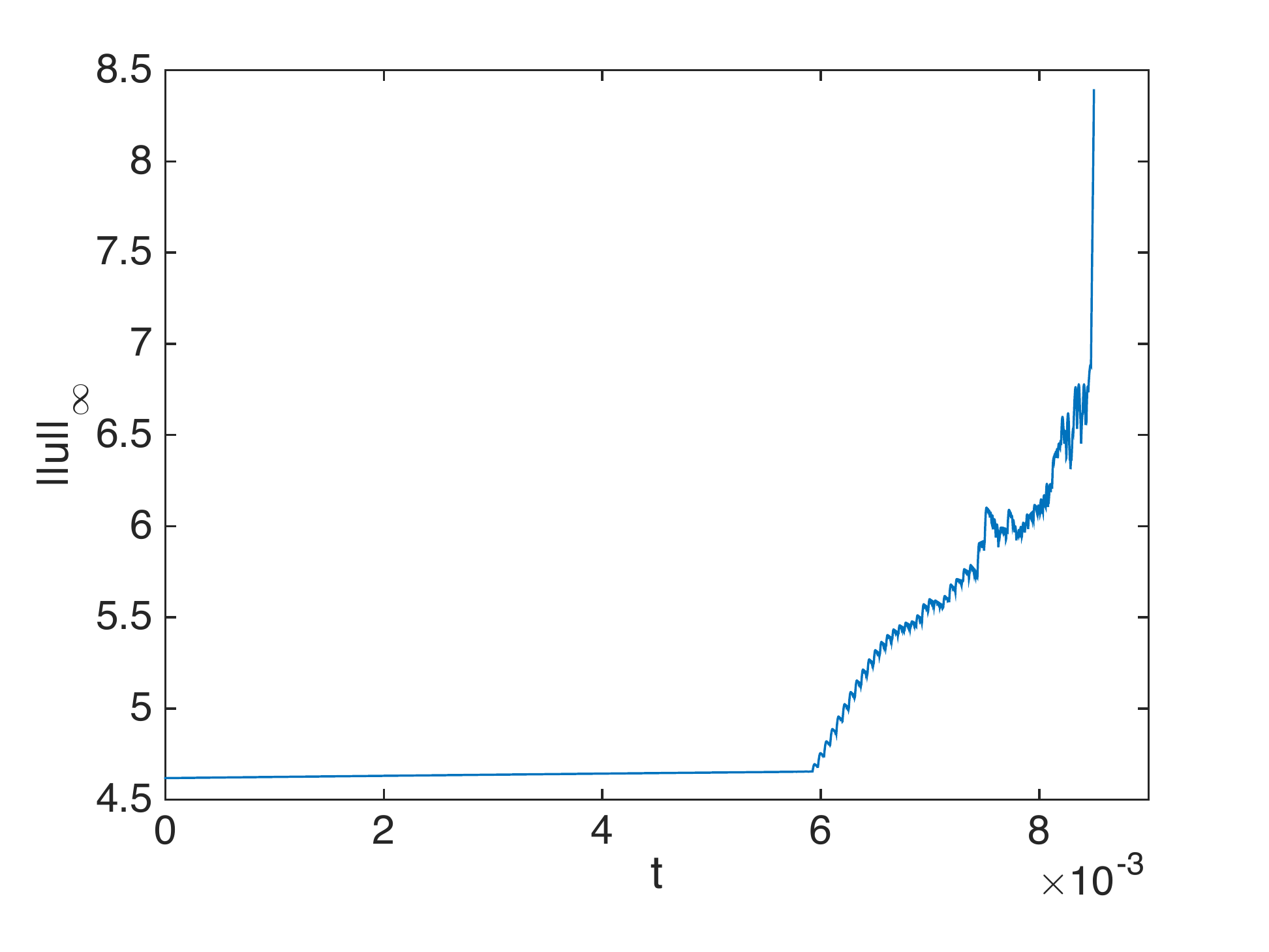}
    \includegraphics[width=0.49\textwidth]{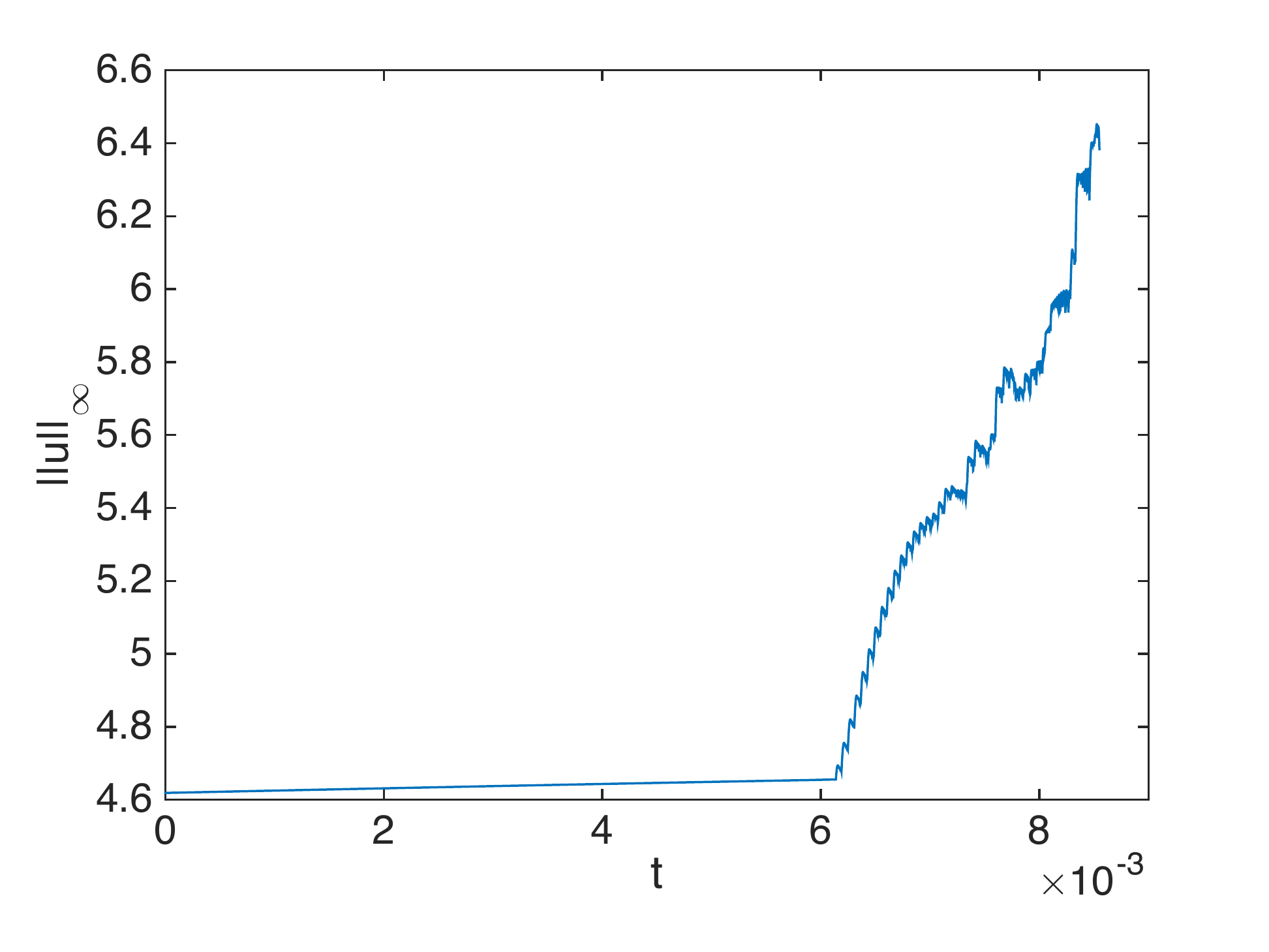}
 \caption{$L^{\infty}$ norm of the solution to generalized KP  
 equations with 
 $\epsilon=0.01$ and $n=3$ for the initial 
 data (\ref{u0sym}) in dependence of time; on the left for KP I, on 
 the right for KP II.   }
 \label{gKPnorm}
\end{figure}

The corresponding generalized KP II solution can be seen in Fig.~\ref{gKPbu} on 
the right. There appears to be some blow-up at a slightly later time which is 
in accordance with the defocusing character of generalized KP II. The blow-up is 
again in the region of positive values of $u$,  not on the 
side of the algebraic tails as for generalized KP I. The $L^{\infty}$ norm of the solution 
shown in  Fig.~\ref{gKPnorm}on the right  also indicates a blow-up. Again we do 
mot manage to fit the norms to the formulae of \cite{KP14} for lack of 
spatial resolution. 
It will be the subject of future work to study the details of the 
blow-up in the presence of a dispersive shock. For 
instance  for generalized KdV in \cite{KP15} the dependence of the blow-up 
time as a function of  $\epsilon$ was obtained.


\begin{thebibliography}{99}

\bibitem{Ablowitzetall}  M. J. Ablowitz, D. Bar Yaacov,  A.S. Fokas,  On the inverse scattering transform for the Kadomtsev--Petviashvili equation. {\em Stud. Appl. Math.} {\bf 69} (1983), no. 2, 135-143.

\bibitem{Ablowitz} M. J. Ablowitz, Al. Demirci, Yi-Ping Ma, 
Dispersive shock waves in the Kadomtsev--Petviashvili and Two Dimensional Benjamin--Ono equations.
arXiv:1507.08207.

\bibitem{Alinhac} S.~Alinhac, \emph{Blowup for nonlinear hyperbolic equations}, Birkh{\"a}user, 1995.

\bibitem{BPP94}  M. Boiti,  F.  Pempinelli, A.Pogrebkov, Properties of solutions of the Kadomtsev--Petviashvili I equation. {\em J. Math. Phys.} {\bf 35} (1994), no. 9, 4683-4718.
\bibitem{Bourgain} J. Bourgain, On the Cauchy problem for the Kadomtsev--Petviashvili equation, {\em Geom. Funct. Anal.} {\bf 3} (1993), 315-341. 
%\bibitem{BMP}
%       E. Br\'ezin, E. Marinari, and G. Parisi,
%        A non-perturbative ambiguity free solution of a string model,
%        {\em Phys. Lett. B} {\bf 242} (1990), no. 1, 35--38.


\bibitem{CG1} T. Claeys and T. Grava, Universality of the break-up profile for the KdV equation in the small
dispersion limit using the Riemann-Hilbert approach, {\em Comm. Math. Phys.} {\bf 286} (2009),
979--1009.


\bibitem{CG4} T. Claeys and T. Grava The KdV hierarchy: Universality and a Painlev\'e transcendent,
{\em Internat. Math. Res. Notices} {\bf 2011} (2011),
doi:10.1093/imrn/rnr220.

\bibitem{Claeys0} T. Claeys,  Asymptotics for a special solution to the second member of the Painlev\'e  I hierarchy. {\em J. Phys. A} {\bf 43} (2010) 434012, 18 pp.
\bibitem{CV}  T. Claeys, M.  Vanlessen,  The existence of a real pole-free solution of the fourth order analogue of the Painlev\'e I equation. {\em Nonlinearity} {\bf 20} (2007), no. 5, 1163-1184. 

\bibitem{CM02}S.~Cox and P.~Matthews, Exponential time differencing 
for stiff systems, {\em Journal of Computational Physics} {\bf 176} (2002), pp. 
430-455.  
\bibitem{DVZ} P. Deift, S. Venakides, X. Zhou, New results in small dispersion KdV by an extension of the steepest descent method for Riemann-Hilbert problems. {\em Internat. Math. Res. Notices} 1997, no. 6, 286-299. 
\bibitem{Dryuma}V.S.~Dryuma, Analytic solutions of the two-dimensional Korteweg - de
Vries equation, {\em Pis'ma ZhETF} {\bf 19} (1974) 753-757.

\bibitem{Dub1}
        B.~Dubrovin,
        On Hamiltonian perturbations of hyperbolic systems of
            conservation laws, II, {\em Comm. Math. Phys.} {\bf 267} (2006) 117--139.
        \bibitem{Dub2}   B.~Dubrovin, On universality of critical behaviour. In: Hamiltonian PDEs. Geometry, Topology, and Mathematical Physics, {\em Amer. Math. Soc. Transl.} Ser. 2, {\bf 224} (2008) 59-109.
                  
\bibitem{DGK} B.~Dubrovin, T.~Grava and C.~Klein,  \emph{Numerical Study of breakup in generalized Korteweg--de Vries and Kawahara equations}, {\em SIAM J. Appl. Math.} {\bf 71} (2011) 983-1008.

        \bibitem{DGKM} B.~Dubrovin, T.~Grava, C.~Klein, A.~Moro,  On critical behaviour in systems of Hamiltonian partial differential equations. {\em J. Nonlinear Sci.} {\bf 25} (2015) 631-707. 
     
\bibitem{DMT} M.~Dunajski, L.~Mason, P.~Tod,  Einstein--Weyl geometry, the dKP equation and twistor theory. {\em J. Geom. Phys.} {\bf 37} (2001) 63-93.

    \bibitem{El}G. A.~El, A. M.~Kamchatnov, V. V.~Khodorovskii, E. S.~Annibale, and A.~Gammal, Two-dimensional supersonic nonlinear Schr\"odinger equation flow past an extended obstacle,
{\em Phys. Rev. E} {\bf 80} (2009) 046317.

    \bibitem{FokasAblowitz}  A.S.~Fokas, M.J.~Ablowitz,  On the inverse scattering of the time-dependent Schršdinger equation and the associated Kadomtsev--Petviashvili equation. {\em Stud. Appl. Math.} {\bf 69} (1983), no. 3, 211-228.
  
\bibitem{FS97} A.S.~Fokas and L.Y.~Sung, The Cauchy problem for the 
Kadomtsev--Petviashvili-I equation without the zero mass constraint, 
{\em Math. Proc. Camb. Phil. Soc.} (1999), 125, 113.  
    
    \bibitem{Ferapontov}      E.V.~Ferapontov, A.~Moro,   Dispersive deformations of hydrodynamic reductions of (2+1)D dispersionless integrable systems. {\em J. Phys. A} {\bf 42} (2009), no. 3, 035211, 15 pp. 
    
\bibitem{GKK} T.~Grava, A.~Kapaev, and C.~Klein, {On the tritronqu\'ee 
solutions of P$_I^2$}, {\em Constr. Approx.} {\bf 41} (2015) 425--466. %DOI 10.1007/s00365-015-9285-3 

\bibitem{GK12} T.~Grava and C.~Klein,  {Numerical study of the small 
dispersion limit of the Korteweg--de Vries equation and asymptotic 
solutions},  {\em Physica D} 10.1016/j.physd.2012.04.001 (2012).

\bibitem{GK07} T.~Grava and C.~Klein,  {Numerical solution of the small dispersion limit of 
Korteweg de Vries and Whitham equations},
{\em Comm. Pure Appl. Math.} {\bf 60(11)} (2007) 1623-1664.

\bibitem{GKE}   T.~Grava, C.~Klein and J.~Eggers, {Shock formation in 
     the dispersionless Kadomtsev--Petviashvili equation}, 
    arXiv:1505.06453 (2015)

\bibitem{GP}  A. V.~Gurevich and L. P.~Pitaevskii, Non stationary 
structure of collisionless shock waves, {\em JETP Letters},
{\bf 17}, 193-195 (1973) 

\bibitem{Hoefer} M. A.~Hoefer and B.~Ilan, Theory of two-dimensional oblique dispersive shock waves in
supersonic flow of a superfluid, {\em Phys. Rev. A} {\bf 80} (2009), 061601(R).

\bibitem{HMV} G.~Huang, V.A.~Makarov, and M.G.~Velarde, Two-dimensional 
solitons in Bose--Einstein condensates with a disk-shaped trap. {\em Phys.  Rev. A} {\bf 67} (2003) 23604-23616. 
  
\bibitem{IN} R.J.~I—rio Jr., W.V.L.~Nunes, On equations of KP-type, {\em Proc. Roy. Soc.} (1998), 725-743

\bibitem{Joh} R.S.~Johnson, The classical problem of water waves: a 
reservoir of integrable and nearly integrable equations. {\em J. Nonlinear 
Math. Phys.} {\bf 10} (2003) 72-92.   

\bibitem{JR} C.A.~Jones and P.H.~Roberts. Motions in a Bose 
condensate, IV: Axisymmetric solitary waves. {\em J. Phys. A Math. Gen.} 
{\bf 15} (1982) 2599-2619.   

\bibitem{KP70}
B.~B. Kadomtsev and V.~I. Petviashvili, {On the stability of solitary
  waves in weakly dispersive media}, {\em Sov. Phys. Dokl.} \textbf{15} (1970), 539.

\bibitem{KP15} C.~Klein and R.~Peter, {Numerical study of blow-up in solutions 
to generalized Korteweg--de Vries equations},  {\em Physica D}  304-305 
(2015), 52-78
% 52-78 DOI 10.1016/j.physd.2015.04.003.

\bibitem{KP14}  C.~Klein and R.~Peter, {Numerical study of blow-up in solutions 
to generalized Kadomtsev--Petviashvili equations}, 
{\em Discr. Cont.  Dyn.  Syst. B}  {\bf 19(6)} (2014) doi:10.3934/dcdsb.2014.19.1689

\bibitem{KR13} C.~Klein and K.~Roidot, {Numerical study of shock formation in the dispersionless 
Kadomtsev--Petviashvili equation and dispersive 
regularizations}, {\em Physica D} {\bf 265}  (2013) 1--25.

\bibitem{KS12}  C.~Klein and J.-C.~Saut,  {Numerical study of blow up and 
stability of solutions of generalized Kadomtsev--Petviashvili 
equations}, {\em J. Nonl. Sci.} {\bf 22 (5)} (2012) 763-811.

\bibitem{KR11} C.~Klein and K.~Roidot,  {Fourth order time-stepping for Kadomtsev--Petviashvili and 
Davey--Stewartson equations},  {\em SIAM J. Sci. Comput.} {\bf 33(6)} (2011) 3333-3356. 

\bibitem{KS07} C.~Klein and C.~Sparber,  {Numerical simulation of 
generalized KP type equations with small dispersion},  in   \emph{Recent Progress in Scientific Computing}, ed.  by W.-B.~Liu, Michael Ng and Zhong-Ci Shi, Science Press (Beijing)
(2007).


\bibitem{KSM} C.~Klein, C.~Sparber and P.~Markowich,  {Numerical study of  oscillatory regimes in the Kadomtsev--Petviashvili equation},  {\em J. Nonl. Sci.} {\bf 17(5)} (2007) 429-470.

\bibitem{GK} Y.~Kodama, J.~Gibbons,  A method for solving the dispersion-less KP hierarchy and its exact solutions. II. {\em Phys. Lett. A} {\bf 135} (1989), no. 3, 167-170. 
        

\bibitem{Kra} R.~Krasny, A study of singularity formation in a vortex sheet by the point-vortex approximation, {\em J. Fluid Mech.} {\bf 167} (1986), 65-93.

\bibitem{fminsearch} J.C. Lagarias,  J. A. Reeds, M. H. Wright, and  P. E. Wright, Convergence Properties of the Nelder-Mead Simplex 
Method in Low Dimensions, {\em SIAM Journal of Optimization} {\bf 9} (1998) 
Number 1, pp. 112-147.     

\bibitem{LL} P.D.~Lax,  C.D.~Levermore,  The small dispersion limit of the Korteweg--de Vries equation. III. {\em Comm. Pure Appl. Math.} {\bf 36} (1983), no. 6, 809-829. 

\bibitem{LRT}C.~Lin, E.~Reissner, and H.S.~Tsien, On two-dimensional non-steady motion of a slender body in a compressible fluid. {\em J. Math. Physics.} {\bf 27} (1948) 220-231

\bibitem{LinChen} J. E.~Lin,  H.H.~Chen,  Constraints and conserved quantities of the Kadomtsev--Petviashvili equations. {\em Phys. Lett. A} {\bf 89} (1982), no. 4, 163-167. 

\bibitem{Liu} Y.~Liu, Blow-up and instability of solitary-wave solutions to a generalized Kadomtsev--Petviashvili equation. {\em Tamsui Oxf. J. Manag. Sci.} {\bf 353} (2001) 191-208.
%\bibitem{Pomeau}

\bibitem{Manakov} S.V.~Manakov, The inverse scattering transform for the time dependent Schr\"odinger equation and the Kadomtsev-Petviashvili equation, {\em Physica D}  {\bf 3} (1981) 42-427.

\bibitem{MS08} S.V.~Manakov and P.M.~Santini, {On the solutions of the d{K}{P}
  equation: the nonlinear Riemann--Hilbert problem, longtime behaviour, implicit
  solutions and wave breaking}, {\em Nonlinearity} \textbf{41} (2008), 1.

\bibitem{MS12}
S.V.~Manakov and P.M.~Santini,   {Wave breaking in the solutions of the dispersionless
  {K}adomtsev--{P}etviashvili equation at a finite time}, {\em Theoret. and Math.
  Phys.} \textbf{172} (2012), 1117.


\bibitem{MST} L.~Molinet,  J.C.~Saut, N.~Tzvetkov, Global well-posedness for the KP-I equation.
{\em Math. Ann.} {\bf 324} (2002), no. 2, 255-275.

\bibitem{MST07}  L.~Molinet,  J.C.~Saut, N.~Tzvetkov,  Remarks on the mass constraint for KP-type equations. {\em SIAM J. Math. Anal.} {\bf 39} (2007), no. 2, 627-641.
        %   \bibitem{Pomeau}  Pomeau, Yves; Le Berre, Martine; Guyenne, Philippe; Grilli, Stephan Wave-breaking and generic singularities of nonlinear hyperbolic equations. Nonlinearity 21 (2008), no. 5, T61-T79.

\bibitem{Raimondo}  A.~Raimondo,  Frobenius manifold for the dispersionless Kadomtsev--Petviashvili equation. {\em Comm. Math. Phys.} {\bf 311}  (2012), no. 3, 557-594.

\bibitem{Rozanova}  A. Rozanova,  The Khokhlov-Zabolotskaya-Kuznetsov equation. C. R. Math. Acad. Sci. Paris 344 (2007), no. 5, 337Ð342.  Preprint:https://hal.archives-ouvertes.fr/hal-00112147

\bibitem{Saut93} J.-C.~Saut, Remarks on the generalized Kadomtsev--Petviashvili equations. {\em Indiana Univ. Math. J.} {\bf 42}
(1993) 1011-1026.


\bibitem{Suleimanov}  B. I.~Suleimanov, Solution of the Korteweg--de Vries equation which arises near the breaking point in problems with a slight dispersion.  {\em JETP Lett.} {\bf 58 (11)} (1993) 849.
\bibitem{TT} K.~Takasaki, T.~Takebe, Integrable hierarchies and dispersionless limit. {\em Rev. Math. Phys.} {\bf 7} (1995), no. 5, 743808.

\bibitem{TF} S.~Turitsyn and G.~Falkovitch, Stability of magneto-elastic solitons and self-focusing of sound in
antiferromagnets. {\em Sov. Phys. J. Exp. Theor. Phys.} {\bf 62} (1985) 146-152.

\bibitem{Venakides}S.~Venakides, The zero dispersion limit of the Korteweg--de Vries equation for initial data with nontrivial reflection coefficient. {\em Comm. Pure Appl. Math.} {\bf 38} (1985), 125-155.

\bibitem{WAS} X.P.~Wang, M.J.~Ablowitz, H.~Segur, Wave Collapse and 
Instability of Solitary Waves of a Generalized Kadomtsev--Petviashvili 
Equation, {\em Physica D} {\bf 78} (1994) 241-265.   

\bibitem{ZK69}
 E.A.~Zabolotskaya and R.V.~Khokhlov, {Quasi-plane waves in the nonlinear
  acoustics of confined beams}, {\em Sov. Phys. Acoustics} \textbf{15} (1969),
  35--40.

\bibitem{ZS} V.E.~Zakharov, A.B.~Shabat,  A scheme for integrating the 
  nonlinear equations of mathematical physics by the method of the 
  inverse scattering problem. I, {\em Funct. Anal. Appl.} {\bf 8(3)} (1974) 226-235.      
\end{thebibliography}
\end{document}